\documentclass{article}
\usepackage{graphicx} 
\usepackage{subcaption}
\usepackage[
    backend=biber,
    style=authoryear,
    sorting=nyt,
    natbib=true,
    uniquename=false
    ]{biblatex}
\usepackage{amsmath,amsfonts,amssymb,amsthm,mathrsfs,mathtools}
\usepackage{bm} 
\usepackage{todonotes}
\usepackage{accents} 
\usepackage{appendix}

\usepackage[ruled,vlined]{algorithm2e}
\usepackage{algpseudocode}

\theoremstyle{plain}
\newtheorem{thm}{Theorem}[section]
\newtheorem{lem}{Lemma}[section]

\usepackage{bm}

\usepackage{fullpage}
\usepackage{setspace}

\theoremstyle{definition}
\newtheorem{assu}{Assumption}[section]
\newtheorem{defn}{Definition}[section]

\newtheorem{exmp}{Example}[section]

\theoremstyle{definition}
\newtheorem{remark}{Remark}[section]

\AtEndEnvironment{assu}{\pushQED{\qed}\popQED}
\AtEndEnvironment{exmp}{\pushQED{\qed}\popQED}
\renewcommand{\qedsymbol}{$\square$}

\AtBeginEnvironment{thm}{\par\addvspace{\bigskipamount}}
\AtBeginEnvironment{lem}{\par\addvspace{\bigskipamount}}
\AtBeginEnvironment{prop}{\par\addvspace{\bigskipamount}}
\AtBeginEnvironment{cor}{\par\addvspace{\bigskipamount}}
\AtBeginEnvironment{assu}{\par\addvspace{\bigskipamount}}
\AtBeginEnvironment{defn}{\par\addvspace{\bigskipamount}}
\AtBeginEnvironment{conj}{\par\addvspace{\bigskipamount}}
\AtBeginEnvironment{exmp}{\par\addvspace{\bigskipamount}}
\AtBeginEnvironment{remark}{\par\addvspace{\bigskipamount}}

\usepackage[colorlinks=true,citecolor=blue,linkcolor=blue,urlcolor=blue]{hyperref} 
\usepackage{cleveref}

\crefname{defn}{Definition}{Definitions}
\Crefname{defn}{Definition}{Definitions}
\crefname{lem}{Lemma}{Lemmas}
\Crefname{lem}{Lemma}{Lemmas}
\crefname{thm}{Theorem}{Theorems}
\Crefname{thm}{Theorem}{Theorems}
\crefname{cor}{Corollary}{Corollaries}
\Crefname{cor}{Corollary}{Corollaries}
\crefname{assu}{Assumption}{Assumptions}
\Crefname{assu}{Assumption}{Assumptions}
\crefname{exmp}{Example}{Examples}
\Crefname{exmp}{Example}{Examples}
\crefname{remark}{Remark}{Remarks}
\Crefname{remark}{Remark}{Remarks}
\crefname{appendix}{Appendix}{Appendices}
\Crefname{appendix}{Appendix}{Appendices}

\usepackage{appendix}

\title{Locally Equivalent Weights for\\ Multilevel Regression and Poststratification\thanks{Corresponding author: \texttt{rgiordano@berkeley.edu} .} }
\author{Ryan Giordano\\UC Berkeley \and Alice Cima\\UC Berkeley \and Jared Murray\\UT Austin \and Erin Hartman\\UC Berkeley \and Avi Feller\\UC Berkeley}

\usepackage[]{graphicx}
\usepackage[]{color}

\makeatletter
\def\maxwidth{ %
  \ifdim\Gin@nat@width>\linewidth
    \linewidth
  \else
    \Gin@nat@width
  \fi
}
\makeatother

\definecolor{fgcolor}{rgb}{0.345, 0.345, 0.345}

\usepackage{framed}
\makeatletter
 {\par\unskip\endMakeFramed%
 \at@end@of@kframe}
\makeatother

\definecolor{shadecolor}{rgb}{.97, .97, .97}
\definecolor{messagecolor}{rgb}{0, 0, 0}
\definecolor{warningcolor}{rgb}{1, 0, 1}
\definecolor{errorcolor}{rgb}{1, 0, 0}
\newenvironment{knitrout}{}{} 

\usepackage{alltt}

\def\argmax#1{\mathrm{argmax}_{#1}\,}

\def\expect#1#2{\mathbb{E}_{#1}\left[ #2\right]\,}
\newcommand{\expecty}[2][\x]{\expect{\p(#2 \vert #1)}{#2}}
\def\cov#1#2{\mathrm{Cov}_{#1}\left( #2\right)\,}
\def\var#1#2{\mathrm{Var}_{#1}\left( #2\right)\,}
\def\cumulant#1#2{\mathcal{K}_{#1}\left( #2\right)\,}

\newcommand{\fracat}[3]{\left. \frac{#1}{#2} \right|_{#3}}
\newcommand{\norm}[1]{\left\Vert#1\right\Vert}
\newcommand{\abs}[1]{\left|#1\right|}

\def\ord#1{\mathcal{O}\left( #1\right)}
\def\ordp#1{\mathcal{O}_p\left( #1\right)}
\def\gauss#1{\mathcal{N}\left( #1\right)}
\def\trans{\intercal} 
\def\id{\bm{I}} 
\def\iid{\overset{iid}{\sim}}
\def\plim{\overset{prob}{\rightarrow}}
\def\dlim{\overset{dist}{\rightarrow}}
\def\eqcheck{\overset{\textrm{check}}{\approx}}
\def\rdom#1{\mathbb{R}^{#1}}
\def\ind#1{\mathbb{I}\left( #1\right)}
\def\ordp#1{\mathcal{O}_p(#1)}
\def\trace#1{\mathrm{Trace}\left( #1 \right)\,}
\def\minev#1{\lambda_{\mathrm{min}}\left( #1 \right)\,}

\newcommand{\tiltil}[1]{\accentset{\approx}{#1}} 

\def\zerov{\bm{0}}

\def\sur{S}
\def\tar{T}

\def\nsur{N_{\sur}}
\def\ntar{N_{\tar}}

\def\sumsur{\sum_{i \in [\nsur]}}
\def\sumtar{\sum_{j \in [\ntar]}}
\def\meansur{\frac{1}{\nsur} \sum_{i \in [\nsur]}}
\def\meantar{\frac{1}{\ntar} \sum_{j \in [\ntar]}}

\def\psur{{\mathcal{P}_{S}}}
\def\ptar{{\mathcal{P}_{T}}}

\def\p{\mathcal{P}}
\newcommand{\post}[1][\Y]{\p(\betav \vert #1)}
\newcommand{\postd}[1][\delta \r]{\p(\betav \vert \Y; #1)}
\def\deltatil{\tilde{\delta}} 
\newcommand{\postdtil}{\postd[\deltatil \r]}

\def\bayes{\mathrm{Bayes}}
\def\mrp{\mathrm{MrP}}
\def\mrplew{\mathrm{MrPlew}}
\def\ols{\mathrm{OLS}}
\def\cal{\mathrm{CW}}

\newcommand{\muhat}[1][]{\hat{\mu}^{#1}}
\newcommand{\w}[1][]{w^{#1}}
\newcommand{\W}[1][]{\bm{W}^{#1}}

\def\f{f}
\def\m{m} 
\def\mhat{\hat{\m}}

\def\betav{\bm{\beta}}
\def\betahat{\hat{\betav}} 
\def\betavhat{\hat{\betav}} 
\def\betastar{\betav^*} 
\def\betadom{\rdom{D_{\betav}}}

\def\g{g} 
\def\info{\mathcal{I}}
\def\infohat{\hat{\info}}
\def\mrpvar{V} 
\def\mrpvarhat{\hat{V}} 

\def\A{A} 
\newcommand{\betagrad}[1][]{\nabla_{\betav}^{#1}}
\def\Agrad#1{\nabla_\eta^{#1}A} 
\def\betaball{\mathcal{B}_{\Delta}} 
\def\deltamax{\delta_{+}} 
\def\deltadom{[0, \deltamax]} 

\def\z{\bm{z}}

\def\x{\bm{x}}
\def\X{\bm{X}}
\def\Xtar{\X_{\tar}}

\def\piv{\bm{\pi}}

\def\r{r}  
\def\R{\bm{R}}  
\def\rset{\mathcal{R}} 
\def\rmax{\mathcal{R}_{\mathrm{max}}} 

\def\rdim{D_r}  

\def\imbalance{\mathrm{Imbalance}}

\def\A{\mathcal{A}} 

\def\y{y}
\def\Y{\bm{Y}}
\def\yhat{\hat{y}}
\def\ytil{\tilde{y}}
\def\ytiltil{\breve{y}}
\def\Ytil{\tilde{\Y}}
\def\Ytiltil{\breve{\Y}}

\def\resid{\mathcal{E}}

\def\xcovhat{\hat{\mathbf{M}}_{xx}}
\def\xscovhat{\widetilde{\mathbf{M}}_{xx}}
\def\xycovhat{\hat{\mathbf{M}}_{xy}}
\def\resv{\bm{\varepsilon}}

\def\betacov{\bm{\Sigma}}

\def\mlogit{\m^{\textrm{logit}}}
\def\vhat{\hat{v}}
\def\V{\bm{V}}
\def\Vtar{\V_{\tar}}
\def\v{v}

\def\alphav{\bm{\alpha}}

\newcommand{\DefineMacros}{
\newcommand{\AlexanderNSur}{4,364}
\newcommand{\AlexanderNTar}{64,784,329}
\newcommand{\AlexanderYbar}{0.462}
\newcommand{\AlexanderMrpMu}{0.288}
\newcommand{\AlexanderRakingMu}{0.263}
\newcommand{\AlexanderNumBoots}{100}
\newcommand{\AlexanderMCMCTimeMins}{95}
\newcommand{\AlexanderMrplewTimeSecs}{664}
\newcommand{\AlexanderMrPSd}{1.06}
\newcommand{\AlexanderRakingSd}{1.39}
\newcommand{\AlexanderMrPPostSd}{1.09}
\newcommand{\StoriesNSur}{4,803}
\newcommand{\StoriesNTar}{406,886}
\newcommand{\StoriesYbar}{0.539}
\newcommand{\StoriesMrpMu}{0.522}
\newcommand{\StoriesRakingMu}{0.522}
\newcommand{\StoriesNumBoots}{100}
\newcommand{\StoriesMCMCTimeMins}{15}
\newcommand{\StoriesMrplewTimeSecs}{5}
\newcommand{\StoriesMrPSd}{0.523}
\newcommand{\StoriesRakingSd}{0.523}
\newcommand{\StoriesMrPPostSd}{0.527}
\newcommand{\LaxNSur}{6,341}
\newcommand{\LaxNTar}{994,486}
\newcommand{\LaxYbar}{0.333}
\newcommand{\LaxMrpMu}{0.457}
\newcommand{\LaxRakingMu}{0.345}
\newcommand{\LaxNumBoots}{100}
\newcommand{\LaxMCMCTimeMins}{39}
\newcommand{\LaxMrplewTimeSecs}{7}
\newcommand{\LaxMrPSd}{1.27}
\newcommand{\LaxRakingSd}{0.583}
\newcommand{\LaxMrPPostSd}{3.07}
\newcommand{\AlexanderColpert}{\texttt{\detokenize{decade_married_rk2009+:educ_group>BA}}}
\newcommand{\AlexanderPctflipped}{2}
\newcommand{\AlexanderShrink}{0}
\newcommand{\LaxColpert}{\texttt{\detokenize{edu.cat2}}}
\newcommand{\LaxPctflipped}{10}
\newcommand{\LaxShrink}{0}

}

\newcommand{\IntroPlot}{

\begin{knitrout}
\definecolor{shadecolor}{rgb}{0.969, 0.969, 0.969}\color{fgcolor}\begin{figure}[H]

{\centering \includegraphics[width=0.98\linewidth,height=0.441\linewidth]{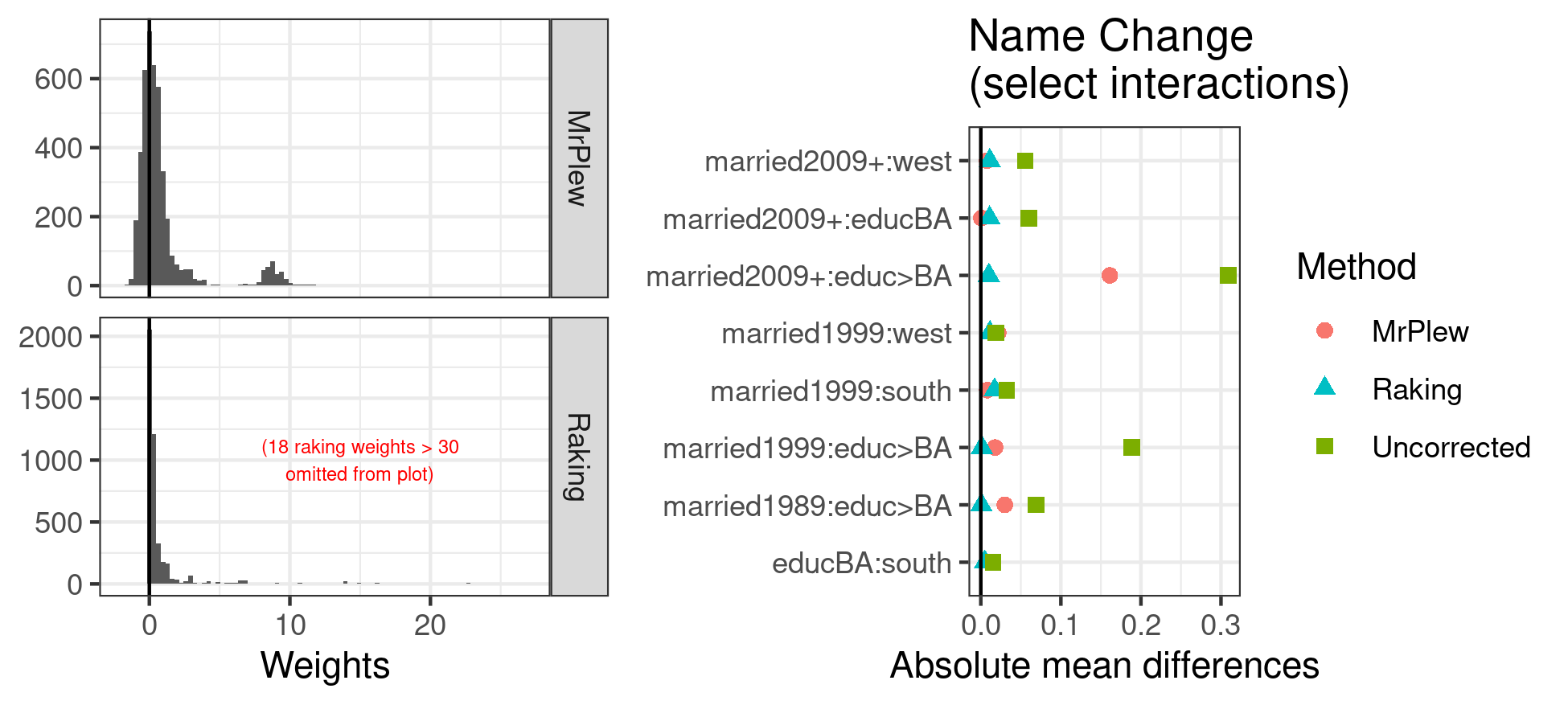} 

}

\caption[]{Preview of MrP Diagnostics made possible by MrPlew for the Name Change analysis}\label{fig:introplot}
\end{figure}

\end{knitrout}
}

\newcommand{\ExperimentTable}{
\begin{tabular}{lccc}
\hline
 & Name Change & Same-Sex Marriage & Election Forecasting \\
\hline
$\nsur$ & $\AlexanderNSur{}$ & $\LaxNSur{}$ & $\StoriesNSur{}$ \\
$\overline{\y}$ & $\AlexanderYbar{}$ & $\LaxYbar{}$ & $\StoriesYbar{}$ \\
$\muhat[\mrp]$ & $\AlexanderMrpMu{}$ & $\LaxMrpMu{}$ & $\StoriesMrpMu{}$ \\
$\muhat[\cal]$ & $\AlexanderRakingMu{}$ & $\LaxRakingMu{}$ & $\StoriesRakingMu{}$ \\
\hline
\end{tabular}

}

\newcommand{\AlexanderResidualTable}{
\begin{table}[!h]
\centering
\caption{\label{tab:alexanderresidtable}Mean response and residuals by interaction value for Name Change}
\centering
\begin{tabular}[t]{|r||l||l|}
\hline
\AlexanderColpert{} & $\overline{y}$ & $\overline{y - \hat{y}}$\\
\hline
0 & 0.412 & -0.001\\
\hline
1 & 0.560 & 0.002\\
\hline
\end{tabular}
\end{table}

}

\newcommand{\VariancePlot}{

\begin{knitrout}
\definecolor{shadecolor}{rgb}{0.969, 0.969, 0.969}\color{fgcolor}\begin{figure}[H]

{\centering \includegraphics[width=0.98\linewidth,height=0.441\linewidth]{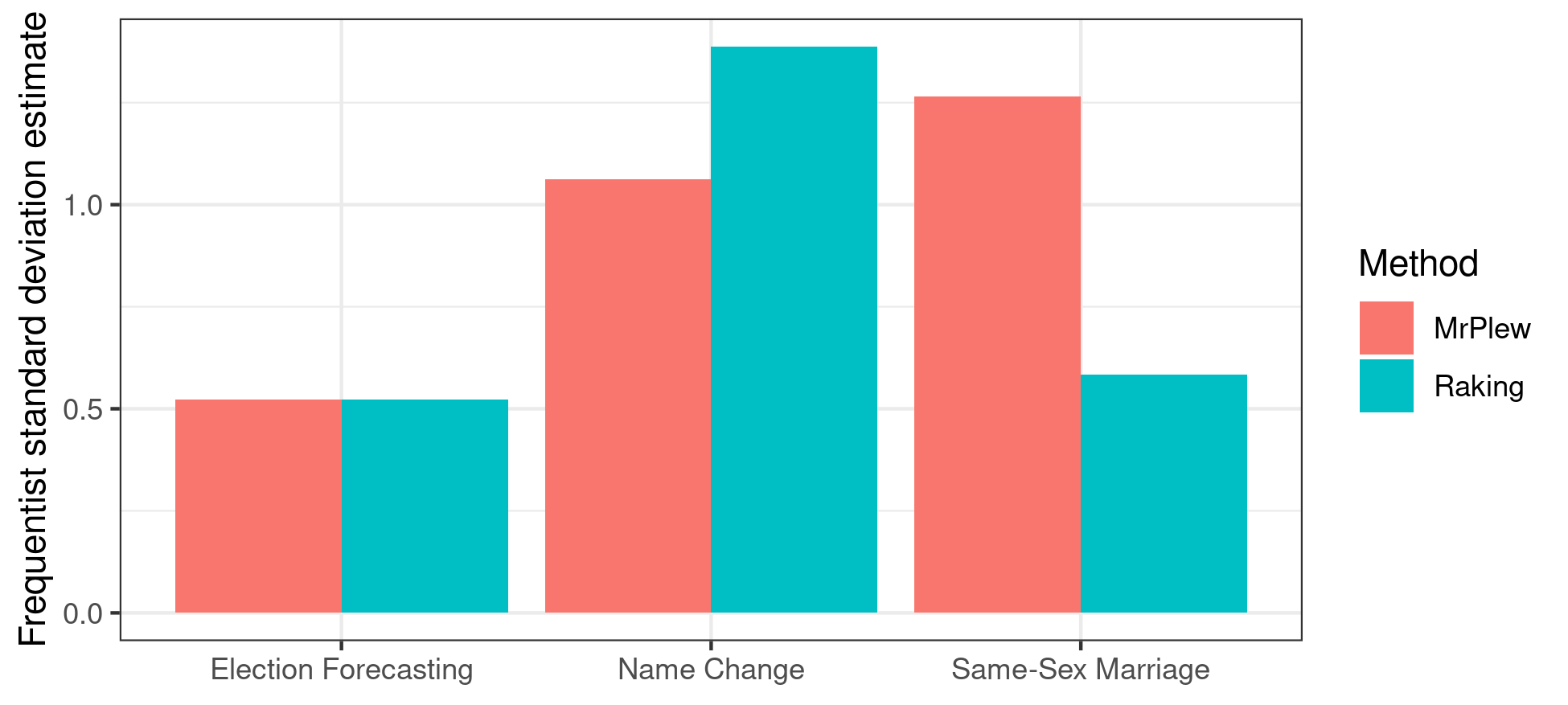} 

}

\caption[]{Estimates of the frequentist standard deviation of $\sqrt{\nsur} \muhat[\mrp]$ and $\sqrt{\nsur} \muhat[\cal]$}\label{fig:varplot}
\end{figure}

\end{knitrout}
}

\newcommand{\BootstrapPlot}{

\begin{knitrout}
\definecolor{shadecolor}{rgb}{0.969, 0.969, 0.969}\color{fgcolor}\begin{figure}[H]

{\centering \includegraphics[width=0.98\linewidth,height=0.441\linewidth]{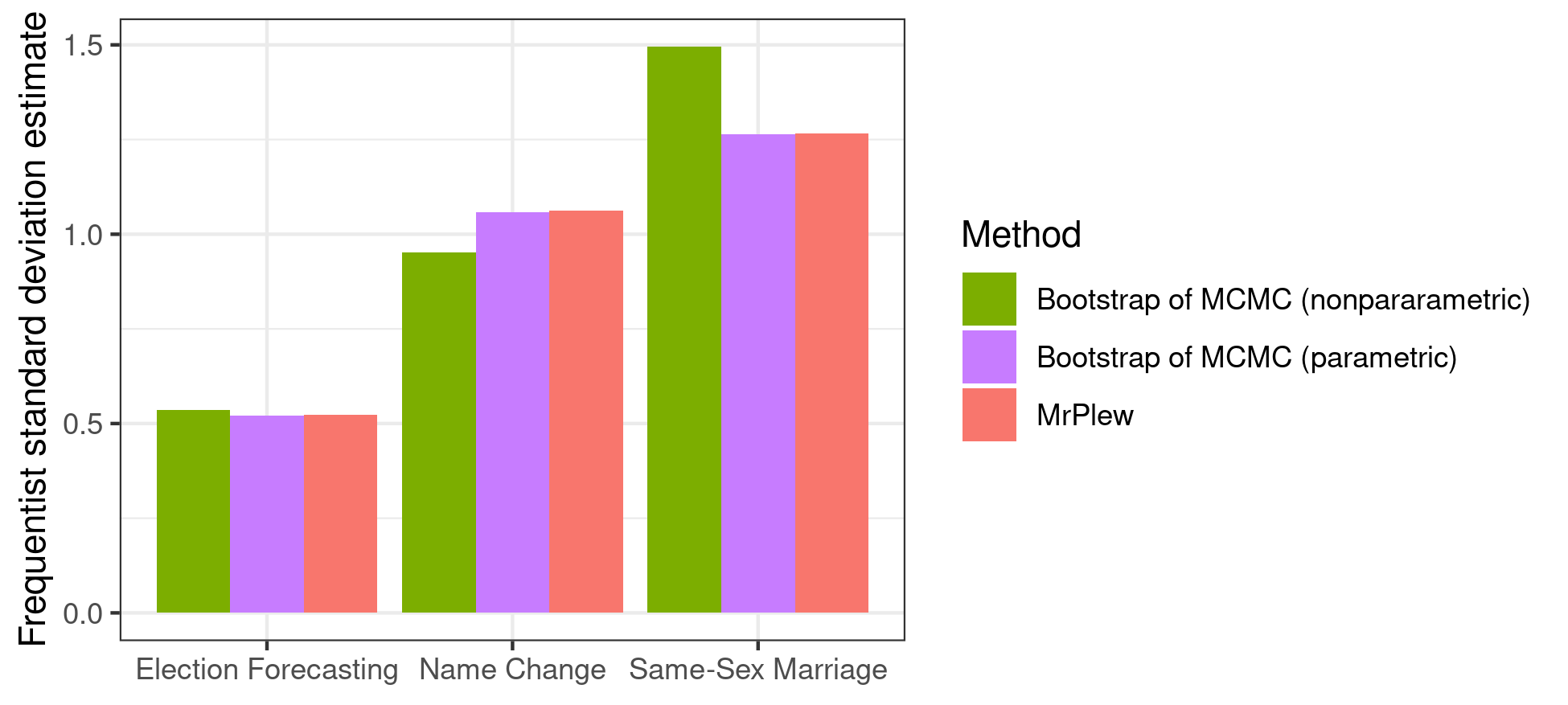} 

}

\caption[]{Estimates of the frequentist standard deviation of $\sqrt{\nsur} \muhat[\mrp]$}\label{fig:bootplot}
\end{figure}

\end{knitrout}
}

\newcommand{\WeightsPlot}{

\begin{knitrout}
\definecolor{shadecolor}{rgb}{0.969, 0.969, 0.969}\color{fgcolor}\begin{figure}[H]

{\centering \includegraphics[width=0.98\linewidth,height=0.441\linewidth]{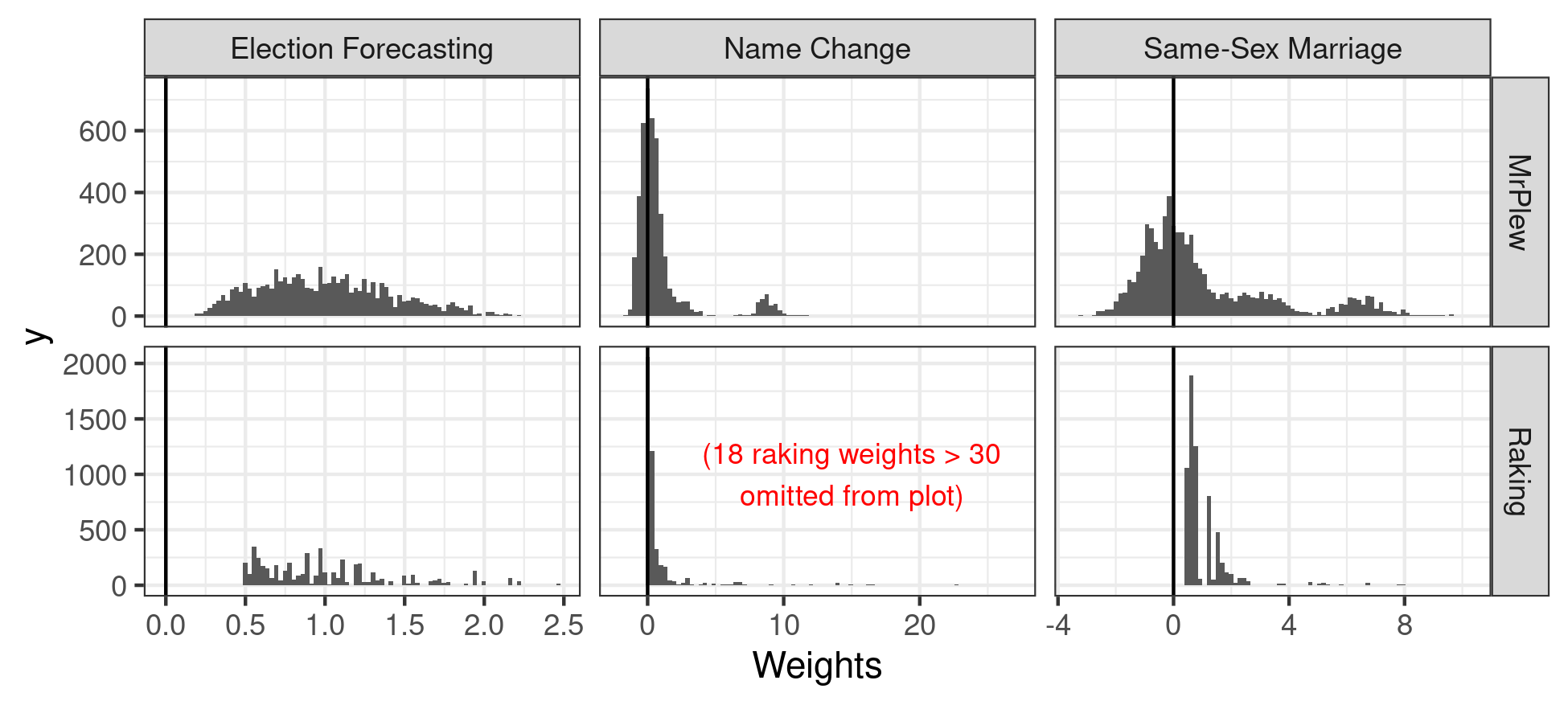} 

}

\caption[]{Comparison of weights}\label{fig:weightsplot}
\end{figure}

\end{knitrout}
}

\newcommand{\AlexanderBalance}{

\begin{knitrout}
\definecolor{shadecolor}{rgb}{0.969, 0.969, 0.969}\color{fgcolor}\begin{figure}[H]

{\centering \includegraphics[width=0.98\linewidth,height=0.441\linewidth]{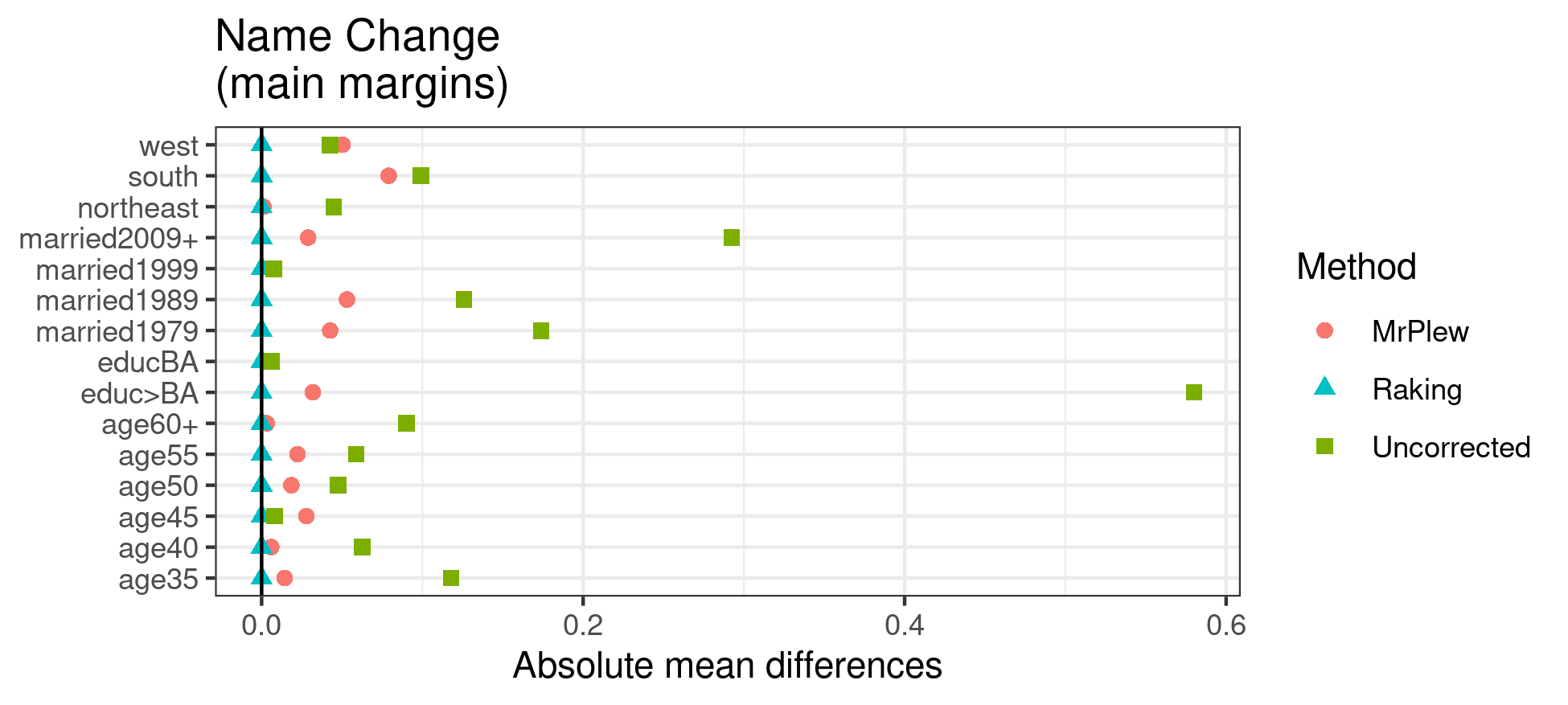} 

}

\caption[]{Balance}\label{fig:alexanderbalanceplot}
\end{figure}

\end{knitrout}
}

\newcommand{\LaxphilipsBalanceInteractions}{

\begin{knitrout}
\definecolor{shadecolor}{rgb}{0.969, 0.969, 0.969}\color{fgcolor}\begin{figure}[H]

{\centering \includegraphics[width=0.98\linewidth,height=0.441\linewidth]{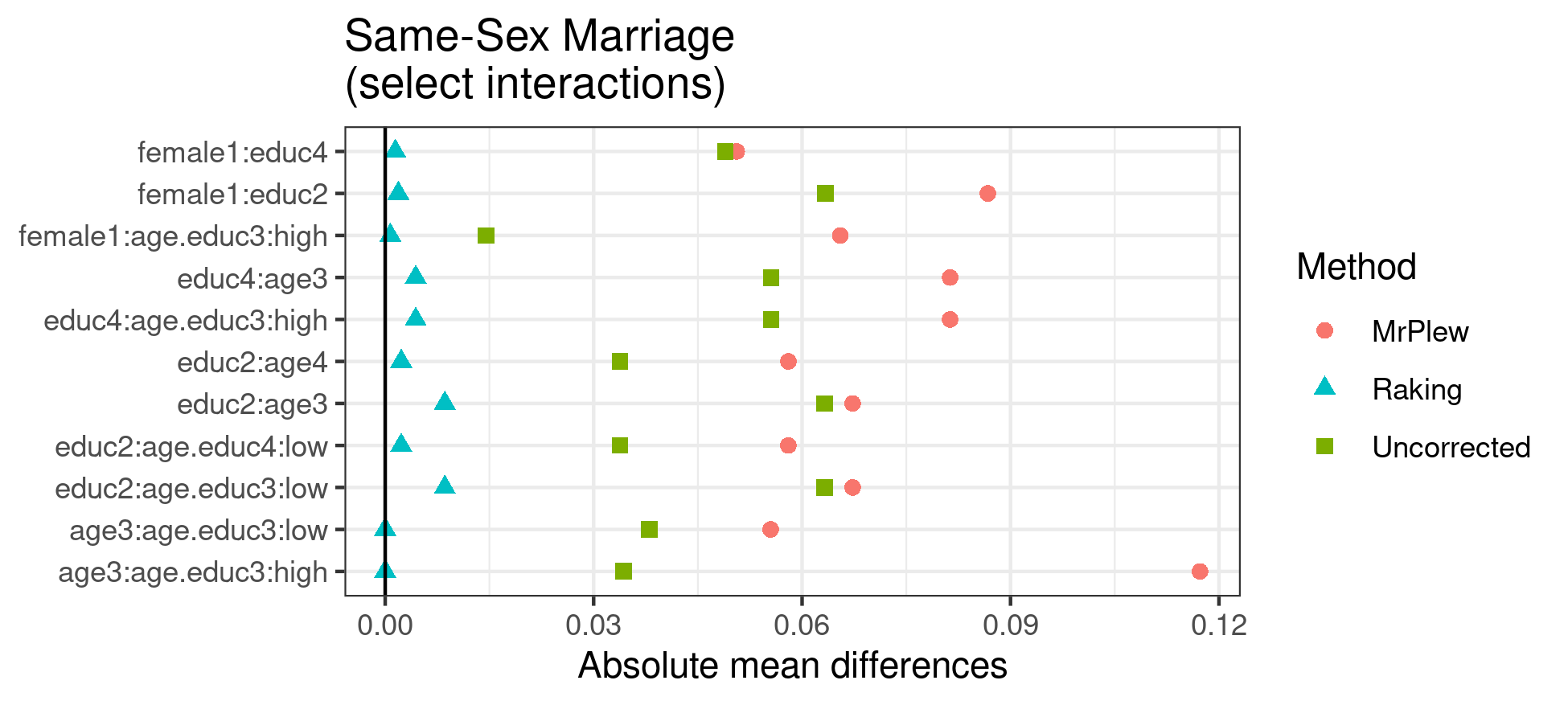} 

}

\caption[]{Balance}\label{fig:laxphilipsbalanceplotinteractions}
\end{figure}

\end{knitrout}
}

\newcommand{\StoriesBalanceInteractions}{

\begin{knitrout}
\definecolor{shadecolor}{rgb}{0.969, 0.969, 0.969}\color{fgcolor}\begin{figure}[H]

{\centering \includegraphics[width=0.98\linewidth,height=0.441\linewidth]{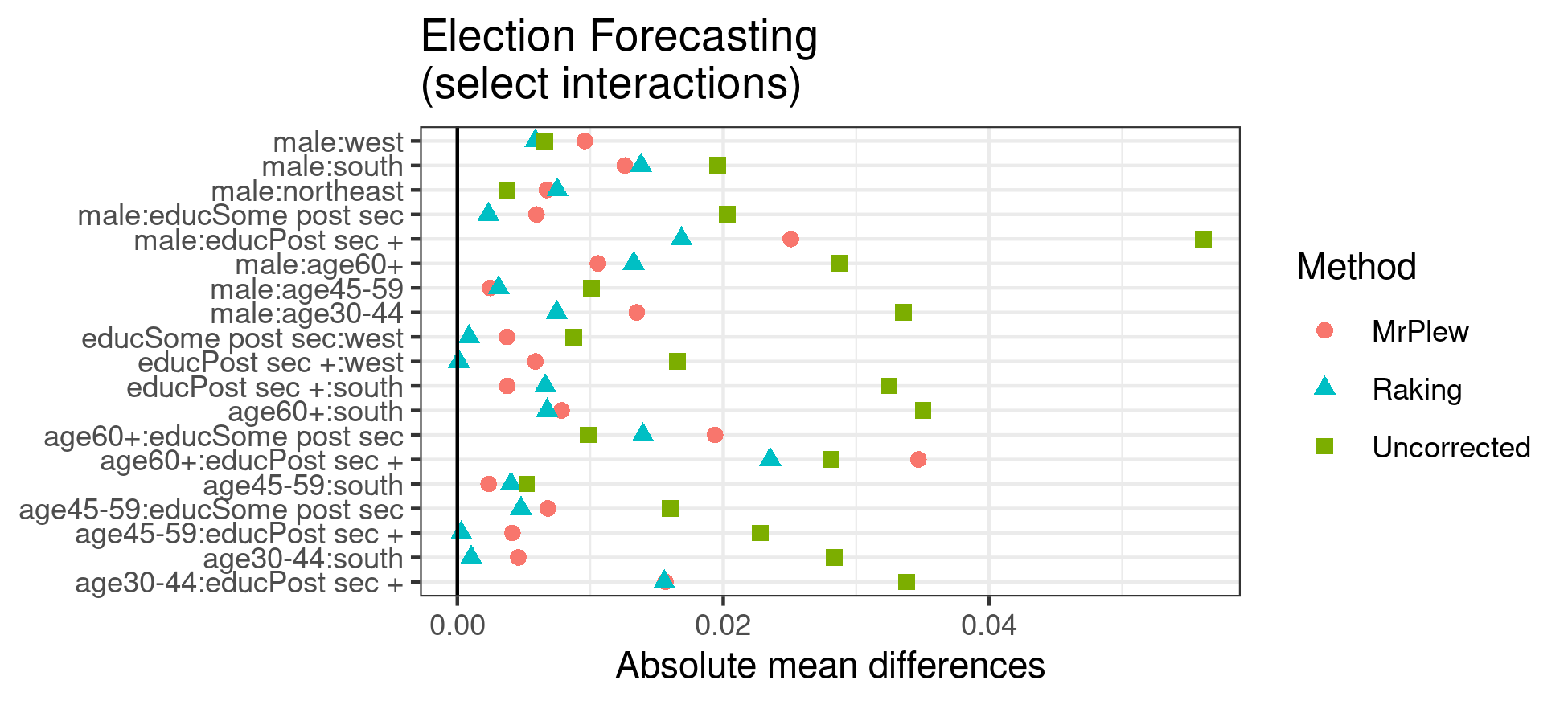} 

}

\caption[]{Balance}\label{fig:storiesbalanceplotinteractions}
\end{figure}

\end{knitrout}
}

\newcommand{\AlexanderRefitPlot}{

\begin{knitrout}
\definecolor{shadecolor}{rgb}{0.969, 0.969, 0.969}\color{fgcolor}\begin{figure}[H]

{\centering \includegraphics[width=0.98\linewidth,height=0.441\linewidth]{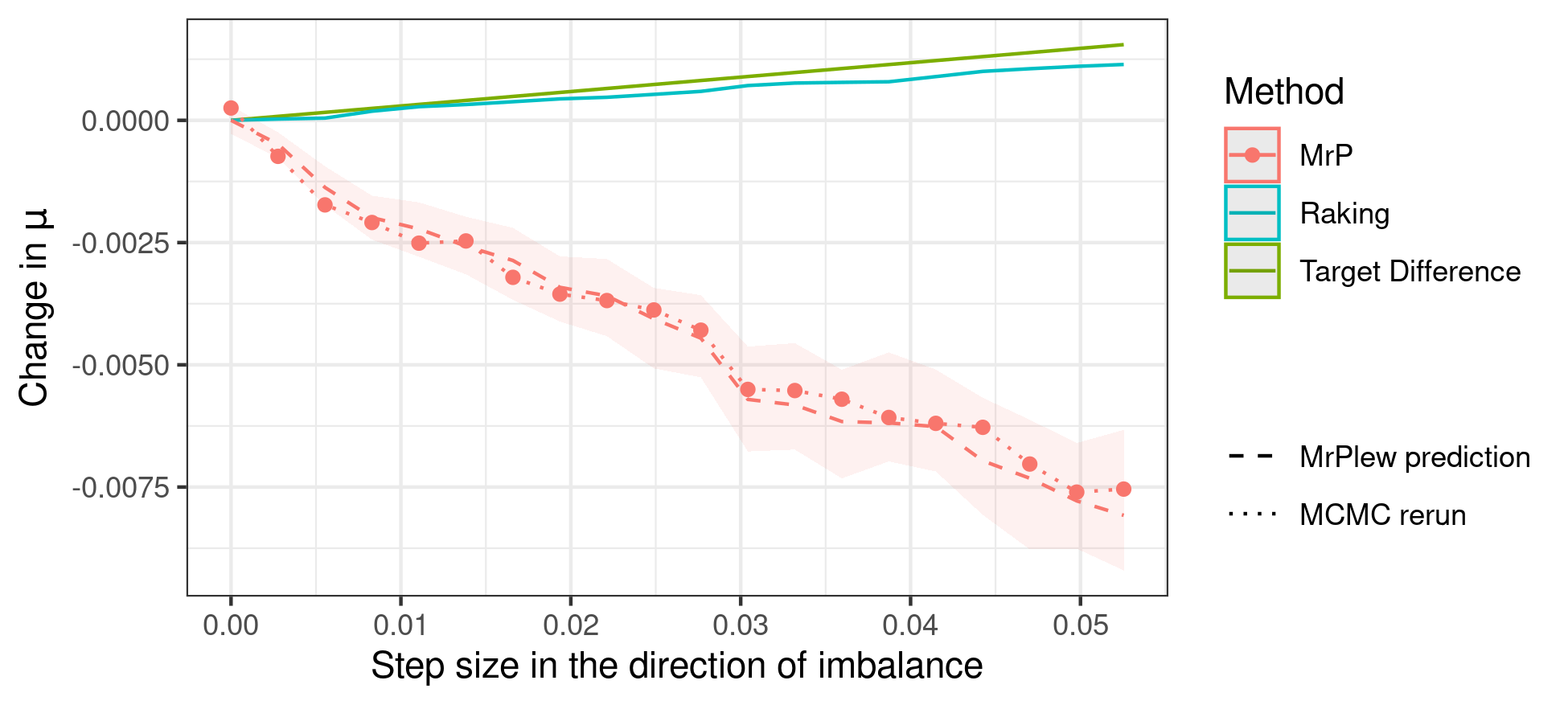} 

}

\caption[]{Refit}\label{fig:alexanderrefitplot}
\end{figure}

\end{knitrout}
}

\newcommand{\LaxRefitPlot}{

\begin{knitrout}
\definecolor{shadecolor}{rgb}{0.969, 0.969, 0.969}\color{fgcolor}\begin{figure}[H]

{\centering \includegraphics[width=0.98\linewidth,height=0.441\linewidth]{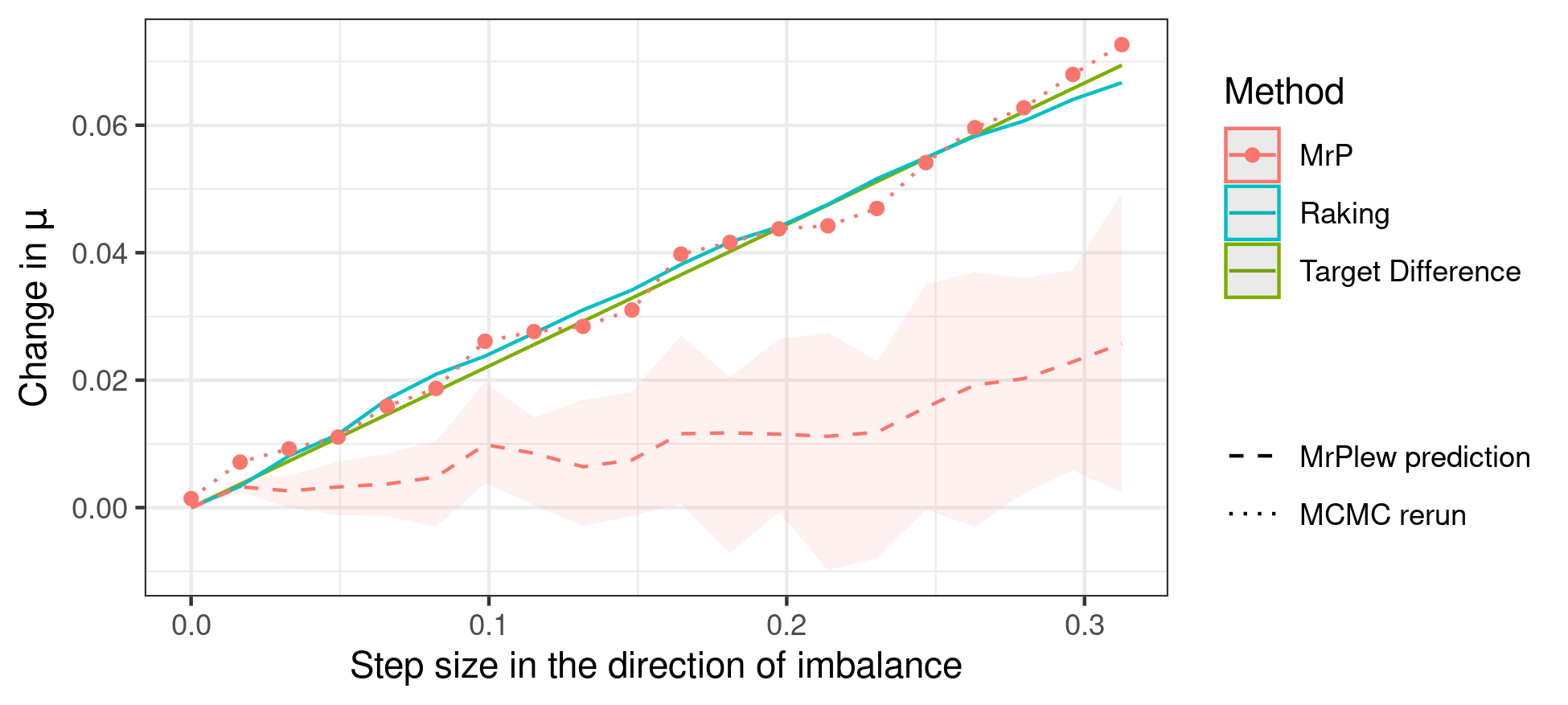} 

}

\caption[]{Refit}\label{fig:laxphilipsrefitplot}
\end{figure}

\end{knitrout}
}

\newcommand{\AlexanderPoolingPlot}{

\begin{knitrout}
\definecolor{shadecolor}{rgb}{0.969, 0.969, 0.969}\color{fgcolor}\begin{figure}[H]

{\centering \includegraphics[width=0.98\linewidth,height=0.441\linewidth]{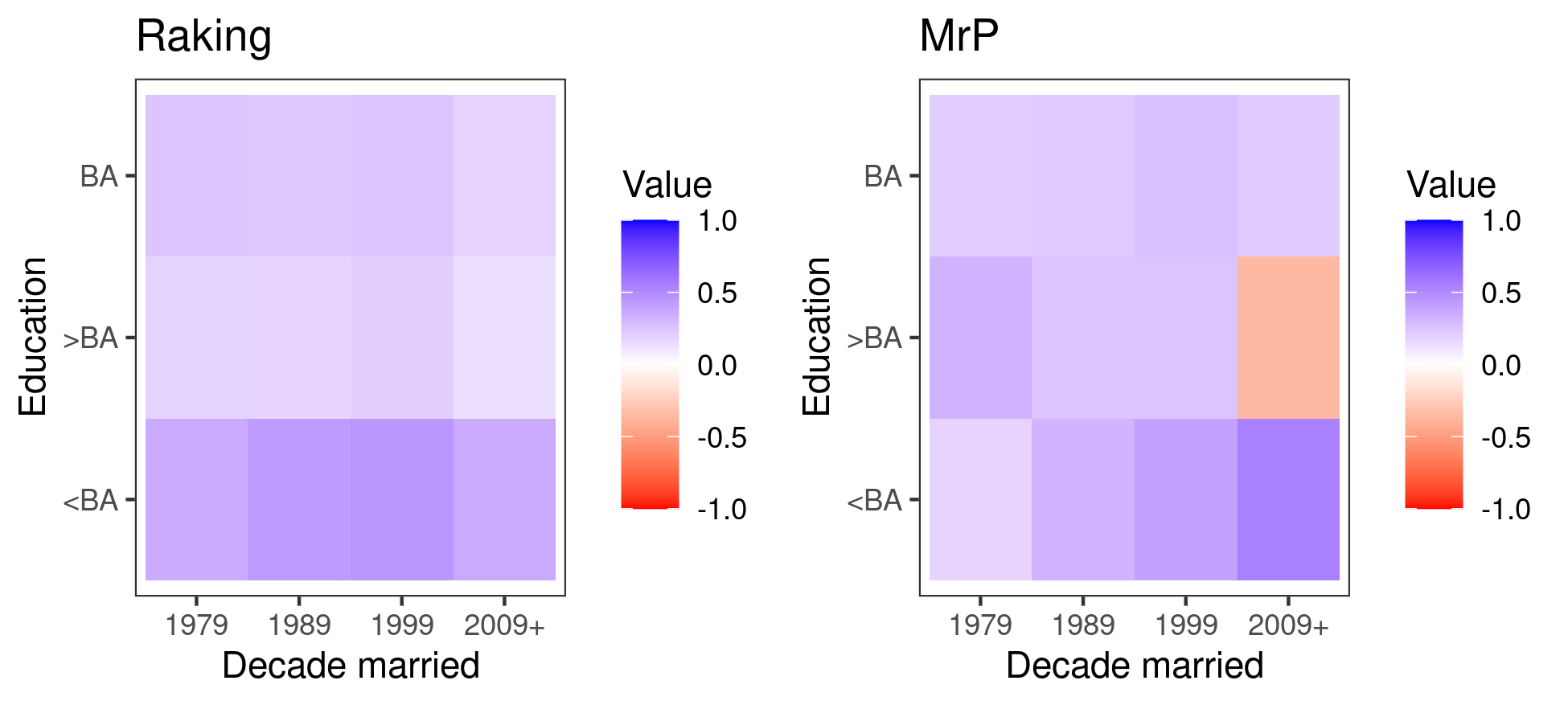} 

}

\caption[]{Subgroup contribution for Name Change}\label{fig:alexanderpoolingplot}
\end{figure}

\end{knitrout}
}

\newcommand{\PoolingComparisonPlot}{

\begin{knitrout}
\definecolor{shadecolor}{rgb}{0.969, 0.969, 0.969}\color{fgcolor}\begin{figure}[H]

{\centering \includegraphics[width=0.98\linewidth,height=0.441\linewidth]{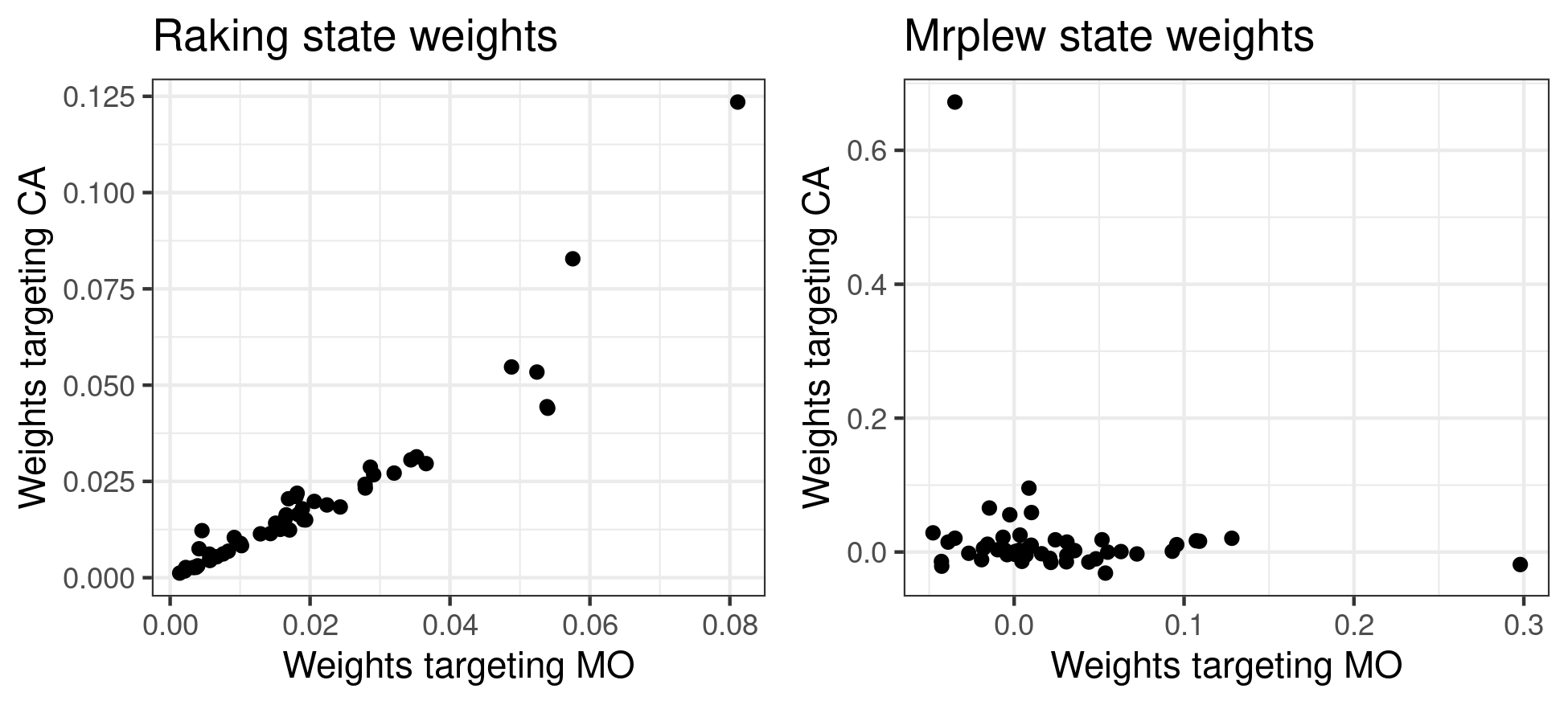} 

}

\caption[]{Different variability in per-state subgroup contribution weights for Same-Sex Marriage. Each point is a state.}\label{fig:poolingcompplot}
\end{figure}

\end{knitrout}
}

\newcommand{\ImportanceSamplingPlot}{

\begin{knitrout}
\definecolor{shadecolor}{rgb}{0.969, 0.969, 0.969}\color{fgcolor}\begin{figure}[H]

{\centering \includegraphics[width=0.98\linewidth,height=0.441\linewidth]{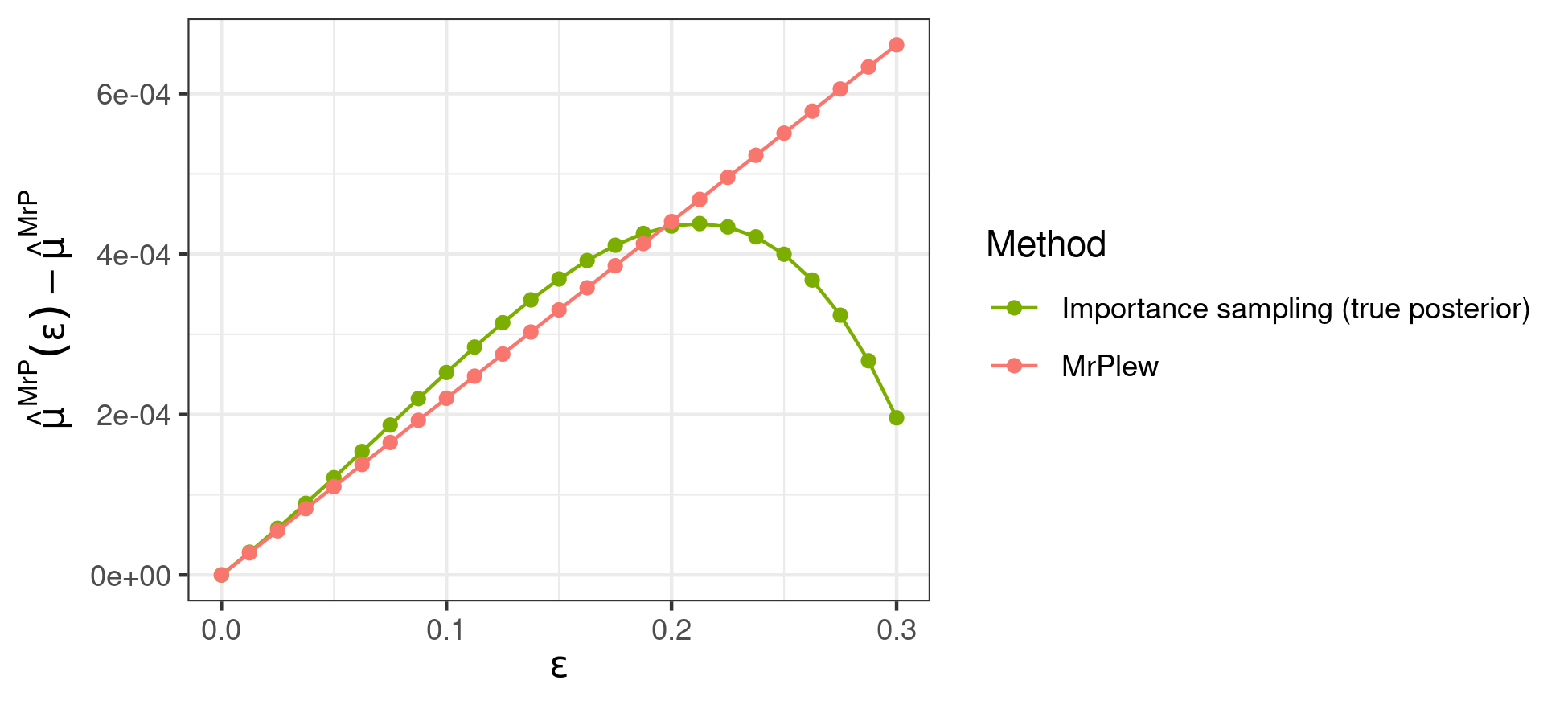} 

}

\caption[]{Local behavior of the Same-Sex Marriage perturbation.  We only consider $\epsilon$ with at least 1000 effective self-normalized importance samples.}\label{fig:isplot}
\end{figure}

\end{knitrout}
}

\newcommand{\AlexanderModelsBalancePlot}{

\begin{knitrout}
\definecolor{shadecolor}{rgb}{0.969, 0.969, 0.969}\color{fgcolor}\begin{figure}[H]

{\centering \includegraphics[width=0.98\linewidth,height=0.980\linewidth]{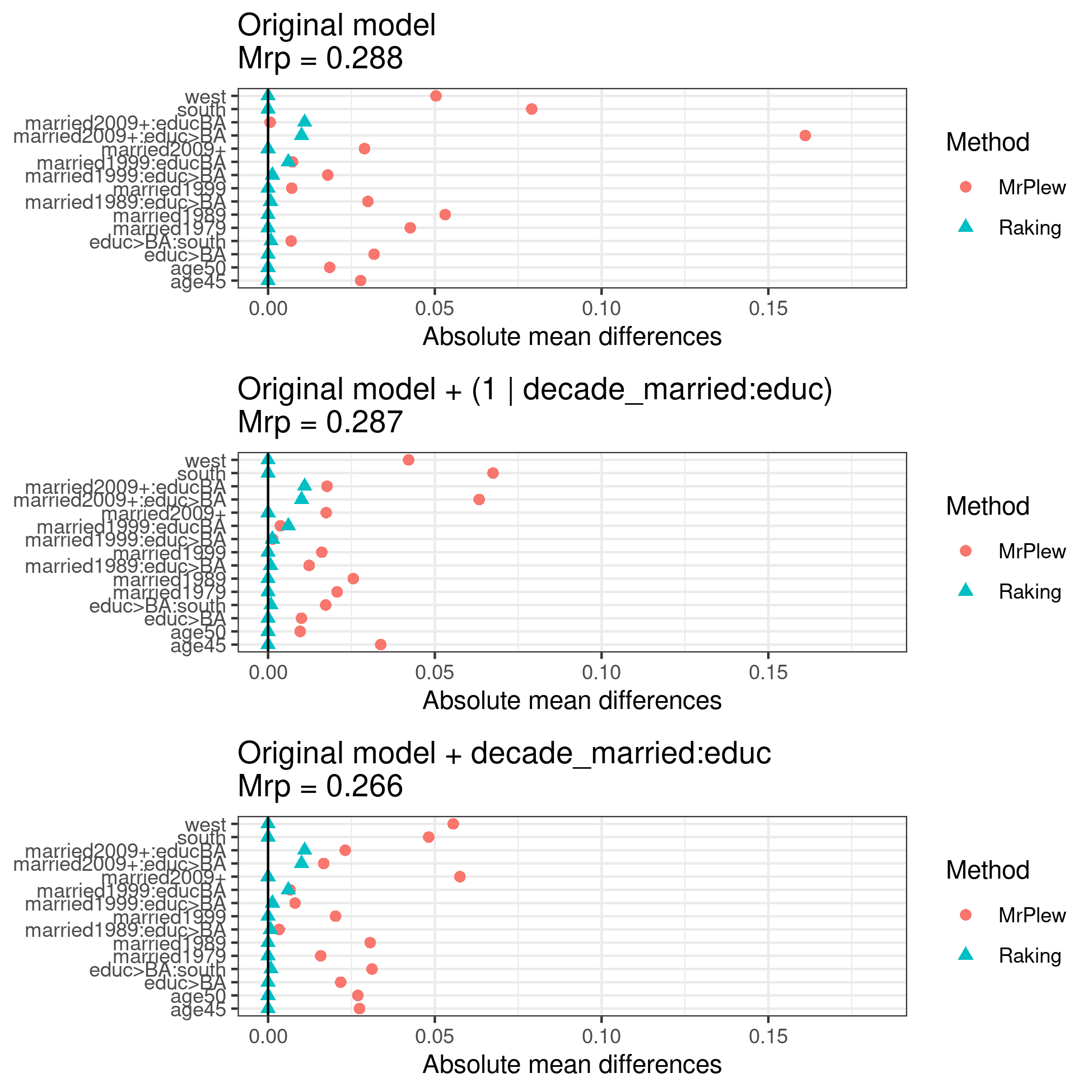} 

}

\caption[]{The effect of adding random and fixed effects for an imbalanced interaction in the Name Change analysis.  Any regressor showing an imbalance in any of the models above a certain threshold is shown in each plot, so regressors that are not shown can be assumed to be reasonably balanced.}\label{fig:alexandermodelsbalance}
\end{figure}

\end{knitrout}

}

\DefineMacros{}

\usepackage{parskip}

\usepackage{float}
\extrafloats{100}

\begin{document}

\maketitle


\begin{abstract}
\noindent Multilevel regression and poststratification (MrP) has become a
workhorse method for estimating population quantities from non-probability
surveys, and is the primary model-based alternative to traditional survey
calibration weighting methods, such as raking.  For simple linear regression
models, MrP methods admit ``equivalent weights'', allowing for direct
comparisons between MrP and traditional calibration weighting. Such weights,
however, have been unavailable for the most widely used MrP models, such as
logistic regression. In this paper, we develop a natural generalization, ``MrP
locally equivalent weights'' (MrPlew), which represent MrP as a weighting-style
estimator that is locally equivalent to calibration weights near the observed
responses. This enables a suite of standard weighting diagnostics, including
frequentist sampling variability, covariate balance, and subgroup contribution.
We formally justify the use of MrPlew in these cases: we prove the MrPlew-based
variance estimator is asymptotically equivalent to the infinitesimal jackknife
for common exponential family models, and we introduce a novel class of model
checks based on invariance to data perturbations that generalize covariate
balance and subgroup contribution to nonlinear models. We further show that
MrPlew can be computed easily using existing MCMC samples and provide
open-source software to compute MrPlew using the output of standard software. We
illustrate our approach for several canonical studies that use MrP, including
via a logistic regression outcome model, showing that implied covariate balance
can sometimes be worse for MrP than for raking. Given the ease of computing, we
recommend making MrPlew a standard part of the MrP model interrogation workflow.
\end{abstract}


\onehalfspacing
\section{Introduction}
Multilevel regression and poststratification \parencite[MrP;][]{Gelman1997} has
become a workhorse method for estimating population quantities from
non-probability surveys, and is the primary model-based alternative to
traditional survey calibration weighting procedures such as raking. MrP adjusts
for survey nonresponse and nonprobability sampling by modeling the relationship between the survey response
and observed covariates (``multilevel regression'') applied to a specific target
(``poststratification''). Estimates are typically obtained from approximate posterior
draws from a nonlinear model, computed via Markov chain Monte Carlo (MCMC).

Calibration weighting (CW) instead constructs estimates as weighted averages of survey responses, where the weights are chosen to exactly or
approximately balance observed covariates subject to some user-defined
dispersion penalty
\parencite{Deville1992,haziza2017weights}. A key advantage
of CW methods relative to MrP methods is their interpretability: visual and
quantitative inspection of the weights gives direct insight into the CW procedure
itself. For example, the variability of CW weights directly determines frequentist sampling variance. Practitioners can use CW weights to check
``covariate balance'' to assess whether the weights correct for differences in
observable quantities between survey and target populations.  CW weights also
directly measure the contribution of particular subgroups, such as
states or demographic groups, to the final estimate.  In short, the weights
themselves serve as a rich set of diagnostic tools.

In certain special cases, MrP estimates can be re-written as CW procedures with
\textit{equivalent weights}, or as a weighting representation of the MrP procedure; examples include simple linear regression models,
Gaussian process (GP) estimation with fixed kernels, and regression trees with
fixed structure
\parencite{gelman:2007:surveystruggles,Park2009,benmichael2024multilevel}. Such
weights enable the whole suite of diagnostics available for calibration
weighting as well as apples-to-apples comparisons for understanding differences
between MrP and CW estimates.
Unfortunately, globally valid equivalent weights are unavailable for the most
widely used MrP models, such as logistic regression and hierarchical models with
estimated random effect variances
\parencite{lopez:2022:mrpintro}. For these models, the MrP estimate is
essentially a computational black box.

In this paper, we propose a natural generalization, ``MrP locally equivalent
weights'' (MrPlew), which represent MrP as a weighting-style estimator similar
to CW, but only for response vectors near the original set of responses, in a
sense we make precise. While MrPlew is not (and cannot be in general) a globally
equivalent weighting representation, we formally justify the use of MrPlew
\emph{as if they were globally equivalent weights} for use in common CW
diagnostics such as frequentist sampling variability, comparisons of the weighted sample to the target population (which we refer to as covariate balance), and
subgroup contribution.  We further show that MrPlew can be computed easily using
existing MCMC samples and provide open-source software to compute MrPlew using
the output of standard MCMC software like \texttt{brms} \parencite{brms}.

To develop the theoretical framework for MrPlew, we make two main technical
contributions. First, we formally prove the asymptotic equivalence of the
MrPlew-based variance estimator with the infinitesimal jackknife variance
estimator \parencite{giordano:2024:bayesij}. This result justifies the use of
MrPlew for assessing frequentist sampling variability, enabling an
apples-to-apples comparison with the sampling variance of CW estimators. Second,
we define a novel class of model checks based on \emph{invariance to data
perturbations}; in the linear case, these reduce to familiar diagnostics. For
the nonlinear case, we formally prove that the MrPlew weights can be used to
assess this invariance locally in an asymptotic regime, uniformly over a large
class of potential perturbations. This result justifies the use of MrPlew for
assessing local covariate balance and subgroup contribution.  We discuss the gap
between our local theoretical results and practically interesting larger
perturbations, and recommend a method for checking the validity of our
diagnostics in practice.

Finally, we apply MrPlew-based model diagnostics to several real-world MrP
analyses. First, we consider an analysis that extrapolates a Twitter survey of
name changes after marriage to the entire US population
\parencite{alexander:2019:namechange,cohen:2019:namechange}. Second, we consider
a textbook example analyzing the 2020 US presidential election
\parencite{alexander:2023:telling}. Third, we consider an example of correcting
sampling bias in a national survey on support for same-sex marriage
\parencite{lax:2009:gay,kastellec:2010:laxmrp}.  Across the three examples, we
directly compare the weights themselves, the frequentist variability, and the
implied covariate balance, finding meaningful divergences, most strikingly in
covariate balance on unmodeled interactions. 

\subsection{Introductory Example: Name Change}
\label{sec:intro_example}


To preview our results, we replicate an MrP analysis from
\citet{alexander:2019:namechange} of the Marital Name Change Survey
\citep[MNCS;][]{cohen:2019:namechange}. The original survey is a convenience
sample from Twitter respondents, and the goal is to estimate the corresponding
rate in the overall US population. Following \citet{alexander:2019:namechange},
we limit the survey data to (self-reported) women married to men ($\nsur =
4{,}364$), and take our response of interest $\y$ to be a binary indicator of
whether the woman retains her surname when marrying; the survey average is
$\bar{y} = 0.46$.

For the ``multilevel regression'' part of MrP, we estimate a hierarchical
logistic regression with age group, education level, state of residence, and
decade married:
$$
\begin{aligned}
\y \sim{}&
    \texttt{logit(~(1 | age\_group) + (1 | educ\_group) + (1 | state) + (1 | decade\_married)~)}.
\end{aligned}
$$
We fit the model using \texttt{brms::brm} with the default priors.
For the ``poststratification'' part of MrP, we match the overall US population
using statistics from the 2017--2022 ACS survey \parencite{ipumsusa}. Finally,
we estimate corresponding calibration weights via \textit{raking}, with some
coarsening of the categories.
\Cref{sec:experiments} gives full details of the analysis and further results. 

\AlexanderBalance{}

\Cref{fig:alexanderbalanceplot} shows that there are substantial differences
between the covariate distribution in the survey and in the target population.
Both MrP and raking then make substantial adjustments, shifting the point
estimate from $\bar{y} = 0.46$ to $\hat{\mu}^{\text{MrP}} = 0.29$ and
$\hat{\mu}^{\text{CW}} = 0.26$. This is where comparisons between CW and MrP
would typically stop; the Bayesian analyst would instead proceed with standard
model checks as part of the Bayesian workflow \citep[see, for
example][]{kennedy2023model,kuh2024using,lopez:2022:mrpintro}.

\IntroPlot{} 
Our goal is to enable a suite of additional diagnostics based on
MrP locally equivalent weights. First, \Cref{fig:alexanderbalanceplot} shows
that the implied covariate balance, measured by comparing the MrPlew weighted covariate means to the population means, for MrP is excellent for the main margins;
raking balances these exactly by construction. However, \Cref{fig:introplot}
shows substantial imbalance on a key interaction term between decade married and
education level. Even though raking and MrP only explicitly adjust for the
margins of these groups, raking nonetheless balances this interaction while MrP
does not. \Cref{sec:experiments} includes additional analyses to explore the
impact of this imbalanced interaction. 

\Cref{fig:introplot} also compares the raking weights and the locally equivalent
MrP weights themselves. The latter weights are clustered much more around zero,
with a substantial proportion of negative weights. Many raking weights, however,
are extreme ($w^{\text{CW}}_i> 30$), which appears to drive up the overall
variability of raking versus MrP. Using our results, we can assess this
directly:
after scaling by $\sqrt{\nsur}$, the estimated frequentist sampling standard
deviation is $\AlexanderRakingSd$ for raking and $\AlexanderMrPSd$ for the MCMC
MrP procedure. 


To the best of our knowledge, such direct comparisons of the strengths and
weaknesses of MrP and calibration weighting methods on the same dataset were not
previously possible.

\subsection{Literature review}\label{sec:litreview}


\paragraph{Survey calibration and soft calibration.}
Calibration weighting adjusts survey estimates by finding weights that balance
observed covariates between sample and population
\parencite{deville:1993:generalizedraking,fuller:2011:sampling,Deville1992,haziza2017weights}.
\textcite{wu:2001:modelcalibration} extend this to \emph{model calibration},
where the calibration targets are derived from a fitted model rather than from
the sampling design. A key recent development is ``soft'' calibration, which
relaxes exact balance to allow approximate balance on higher-order interactions
\parencite{Park2009,Guggemos2010,Wang2019,benmichael2024multilevel}. In
particular, \textcite{gao:2023:softcalibration} make the connection to
random-effects models precise: under a shared random-effects structure, the
optimal calibration weights impose exact calibration on fixed effects and
approximate calibration on random effects.

\paragraph{Equivalent weights.}
Survey researchers have long known that regression adjustment has an equivalent
weighting form \parencite[e.g.,][]{Park2009}. In a foundational paper,
\textcite{gelman:2007:surveystruggles} applies this insight to a Bayesian
Normal-Normal outcome model, showing that the MrP estimate can be re-written as
a calibration weighting estimator with globally valid ``equivalent weights.''
\textcite{benmichael2024multilevel} extend this, showing that multilevel
calibration weights are equivalent to the MAP estimate of a multilevel outcome
model under certain conditions; see
\textcite{chattopadhyay:2023:regressionimpliedweights} and
\textcite{bruns2025augmented} for parallel results from the causal inference
literature.

\paragraph{Calibrated Bayes and Bayesian survey inference.}
\textcite{little:2004:tomodel} frames the central tension in survey inference as
``to model or not to model?'' and proposes \emph{calibrated Bayes}
\parencite{little:2006:calibratedbayes,little:2012:calibratedbayes} as a
resolution: use Bayesian models, but choose them so that the resulting
procedures have good frequentist properties. A growing body of work pursues this
vision, incorporating design information into Bayesian models through weighted
pseudo-posteriors
\parencite{savitsky:2016:bayesinformative,wang:2018:bayesinformative},
model-based survey weights
\parencite{si2020bayesian,si:2015:bayesnonparametric}, and Bayesian analogs of
raking \parencite{si:2021:bayesraking}. Most recently,
\textcite{gelman:2024:mrpw} propose a framework to incorporate design weights
into a Bayesian MrP model by jointly modeling the outcome $y$ and the sampling
weight $w$ given covariates $x$, then poststratifying on $(x, w)$. This is the
complement to our approach, incorporating design weights into a broader Bayesian
model, instead of finding locally equivalent weights of the original Bayesian model.

\paragraph{Extensions and diagnostics for MrP.}
There has been an explosion of interest in MrP; see the recent textbook from
\textcite{lopez:2022:mrpintro}. Extensions include deep interactions \citep{Ghitza2013}, tree-based methods
\parencite{Montgomery2018,bisbee2019barp}, 
and doubly robust-style combinations of weighting and outcome modeling
\parencite{chen2020dr_surveys,benmichael2024multilevel}. Despite this progress,
model diagnostics for MrP remain limited. \textcite{kennedy2023model} and
\textcite{kuh2024using} propose cross-validation-based diagnostics, but these
assess predictive accuracy rather than how the survey data are being used by the
estimator. \textcite{meng:2018:paradises} introduces the data defect index,
measuring the correlation between sample inclusion and the outcome.  Our
contribution can be understood as expanding the suite of tools for assessing
model stability as part of veridical data science for the Bayesian workflow
\parencite{yu:2020:vds,gelman:2020:bayesianworkflow}.

\paragraph{Local robustness.}
Our contributions follow in the tradition of the Bayesian local robustness
literature, which studies the effect of infinitesimal perturbations to Bayesian
posteriors \parencite{basu:1996:local,gustafson:1996:localposterior,
gustafson:2012:localrobustnessbook,gelman:2004:serialdilution,
giordano:2018:covariances,
giordano:2023:bnp,cabral:2025:robustness,di:2025:likelihood},
as well as related local sensitivity ideas in the frequentist case influence
literature \parencite{belsley:1980:regression,cook:1977:detectionofinfluential,
cook:1986:assessment, kass:1989:approximatebayesiansensitivity,
zhu:2007:perturbation,koh:2017:blackboxinfluence,giordano:2019:swiss,
thomas:2018:reconcilingcurvatureandis}. Our core technical results build on
\textcite{giordano:2024:bayesij}.  Specifically, we prove frequentist
consistency using the proof of consistency of the infinitesimal jackknife (IJ)
estimator for Bayesian posterior expectations found in \textcite[Theorem
2]{giordano:2024:bayesij}, and our result for covariate balance follows from an
extension of the series expansions of posterior expectations of
\textcite[Theorem 1]{giordano:2024:bayesij} to hold uniformly over a set of
posteriors.


\section{Methods}

\subsection{Problem setup}
We frame our problem of interest in terms of estimating a population quantity
from a non-representative sample, though the same formal problem arises in
observational causal inference and in domain adaptation for regression more
broadly. We observe scalar survey responses, denoted $\y$, and vector-valued
regressors, denoted $\x$; in the context of the Marital Name Change Survey in
\Cref{sec:intro_example}, $\y_i$ is a binary indicator of whether a married
woman retains her surname and $\x_i$ collects demographic regressors like age
and education level. We additionally observe regressors $\x$ from a target population,
for which the responses are unobserved; in the running example, the target is
the corresponding population of the United States. The problem is to infer the
expected value of the response in the target population, under the assumption
that the conditional distribution of the response given observed regressors is
the same in the survey and target populations.

\paragraph{Notation.}
Let $\{(\x_i, \y_i) : i \in [\nsur]\}$ denote the survey data, where $[\nsur] =
\{1,\ldots,\nsur\}$, $\y_i$ is a scalar survey response, $\x_i \in \rdom{P}$ is
a vector of regressors, and $\nsur$ is the number of survey observations.  Let
$\{\x_j : j \in [\ntar]\}$ denote regressors observed for $\ntar$ units drawn
from the target population; the corresponding responses are not observed.  We
write $\Y = (\y_1, \ldots, \y_{\nsur})^\trans$, $\X$ for the $\nsur \times P$
matrix of survey regressors, and $\Xtar$ for the $\ntar \times P$ matrix of
target regressors.  The responses $\y_i$ may be continuous or discrete, though
we are particularly interested in the binary case.

With some abuse of notation, the symbol $\x$ (without index) will denote a
generic random variable in both the survey and target distributions.  To
avoid ambiguity, we adopt the following convention for expectations: for
a random variable $\z$ with distribution $\p(\z)$ and measurable
function $\phi$,
$$
\expect{\p(\z)}{\phi(\z)} := \int \phi(\z) \, \p(d\z),
$$
with all other quantities taken as fixed.  For example,
$\expect{\p(\y \vert \x)}{\x \y}$ is a function of $\x$ but not of $\y$.
Covariances follow the analogous convention.

\paragraph{Formal setup.}
The surveys literature has several distinct traditions for formalizing
adjustment procedures; see \textcite{elliott2017inference}.  Following
the MrP literature \parencite{lopez:2022:mrpintro}, we anchor our
discussion in the super-population approach \parencite{chang2008using},
though we expect our results to extend to the quasi-randomization
approach \parencite{kott2010using}.  Our contributions are best
understood as \emph{diagnostics}: the assumptions below motivate the
estimators we study, but the diagnostics themselves do not rely on them.

We assume $\x$ has distribution $\psur(\x)$ in the survey and
$\ptar(\x)$ in the target, with these two distributions allowed to
differ.  The CW and MrP estimators we study are motivated by three
assumptions.

\begin{assu}[Invariance]\label{assu:invariance}
The conditional distribution $\p(\y \vert \x)$ is the same in the survey
and target populations.
\end{assu}

This assumption rules out any unmeasured factors that shift the response beyond
what the observed regressors $\x$ capture.  In the Name Change application,
invariance is plausible only to the extent that age, education level, state, and
decade of marriage fully capture the decision to change names.

\smallskip
\begin{assu}[Overlap]\label{assu:overlap}
$\ptar(\x)$ is absolutely continuous with respect to $\psur(\x)$.
\end{assu}

This assumption requires that any $\x$ that occurs with positive probability in
the target also occurs with positive probability in the survey.  Overlap can be difficult to justify in convenience samples, where large strata of the target may be thinly represented or entirely absent in the survey.

\smallskip
\begin{assu}[IID sampling]\label{assu:iid}
$(\x_i, \y_i) \iid \psur(\x) \p(\y \vert \x)$ for $i \in [\nsur]$ and
$\x_j \iid \ptar(\x)$ for $j \in [\ntar]$.
\end{assu}

We state our asymptotic results under IID survey sampling for
simplicity; the IID assumption on the target is not essential.

Our goal is to estimate
$$
\mu := \expect{\ptar(\x, \y)}{\pi(\x) \y}
$$
for some known weighting function $\pi(\x)$.  In the simplest case
$\pi(\x) \equiv 1$ and $\mu$ is the target-population mean of $\y$. This is the estimand in our Name Change example, corresponding to the
overall U.S.\ rate of women retaining their surname.  More generally,
$\pi(\x)$ accommodates subgroup means (e.g., rates within a state) and
contrasts between subgroups, both of direct substantive interest.  For
compactness, we write $\pi_j := \pi(\x_j)$ and $\piv = (\pi_1, \ldots,
\pi_{\ntar})^\trans$.  If we observed $\y_j$ for $j \in [\ntar]$, a
natural estimator would be
$$
\tilde\mu := \meantar \pi_j \y_j \quad\quad\textrm{(Infeasible)},
$$
but $\y_j$ is unobserved outside the survey. We now turn to feasible estimators of this quantity under the assumptions above.


\subsection{Survey adjustment: Calibration weighting}
\label{sec:calibration_weighting}

\subsubsection{Overview}
The classical approach to survey adjustment is \emph{calibration weighting}
\parencite{Deville1992}, with estimators of the form
\begin{align}
\muhat[\cal] := \meansur \w[\cal]_i \y_i
\quad\textrm{for some }\W[\cal] := (\w[\cal]_{1}, \ldots, \w[\cal]_{\nsur})^\trans.
\label{eq:calibration}
\end{align}
There is a rich literature on estimators of this form, reviewed in \Cref{sec:litreview}; see \textcite{haziza2017weights} for a modern overview.


To build intuition, consider the infeasible case in which we observe target responses $\y_j$. Then the weights should satisfy:
\begin{align}
\meantar \pi_j \y_j - \meansur \w[\cal]_i \y_i \eqcheck 0,
\quad\quad\textrm{(Infeasible)}
\label{eq:y_balance}
\end{align}
where $\eqcheck$ denotes an approximate equality treated as a check on the model.

We cannot compute \cref{eq:y_balance} because we do not observe target responses
$\y_j$.  However, we can write an analogous equation for the regressors.  Let
$\r(\x)$ denote some measurable function of $\x$, with $\r_i = \r(\x_i)$ and
$\r_j = \r(\x_j)$ in the survey and target, respectively. Ideally, $\r(\cdot)$
is predictive of either the outcome $\y$ or selection into the survey; see \citet{sarndal2005estimation} or
\citet{benmichael2021_balancing} for discussion.
Then \emph{covariate balance} for $\r(\cdot)$ measures the difference
\begin{align}
    \imbalance(\r, \W[\cal]) :=
    \meantar \pi_j \r_j -
    \meansur \w[\cal]_i \r_i
    \eqcheck 0.
    \label{eq:imbalance}
\end{align}

We use the term covariance balance to refer to the comparison of the weighted
sample distribution of covariates to the target population, however it is also
sometimes referred to as external consistency \citep{haziza2017weights}.
\emph{Raking} \parencite{Deming1940,Deville1992} finds the weights $\W[\cal]$
that minimize a dispersion criterion (e.g., entropy or variance) subject to
$\imbalance(\r, \W[\cal]) = 0$ for a user-chosen set of covariates $\r$. In the
Name Change application, we fit raking weights using the
\texttt{survey::calibrate} function with \texttt{calfun = "raking"}
\parencite{lumley:2025:survey}, with entropy as the dispersion criterion and ACS
population counts as targets. Alternative calibration estimators consider
different dispersion functions and imbalance constraints. An important example
we return to below is \emph{soft calibration}
\parencite{gao:2023:softcalibration,benmichael2024multilevel}, which bounds
$\imbalance(\r, \W[\cal])$ for some $\r$ rather than requiring exact balance.

A key feature of calibration weighting estimators is that they are
\emph{design-based}: the weights are estimated without using the survey
responses $\y$, which enter only through the weighted average in
\Cref{eq:calibration}.  We assume throughout that the calibration
weights are design-based in this sense.

\subsubsection{Diagnostics for calibration weighting}
\label{sec:diagnostics}
Since CW estimators are design-based and linear in the responses $\y$,
practitioners can directly leverage the weights to understand how the estimator
is constructed and to diagnose potential problems. We briefly review three
common diagnostics for CW estimators, which we will later adapt to MrP.

\paragraph{Direct inspection of the weights and subgroup contribution.}  The
weights themselves offer the most important initial diagnostic, for example by
identifying observations with extreme weights.
A natural diagnostic looks at this same influence but aggregated up to different
groups; for example, how much do survey respondents from each state contribute
to the overall estimate? Formally, for some partition $\A = \{\A_1, \ldots,
\A_K\}$ of the regressor space (e.g.~into US states), we can compute
$
   \meansur \w[\cal]_i \ind{\x_i \in \A_k},
$
where $\ind{\cdot}$ is the indicator function.  We provide concrete examples of
this in \Cref{sec:partial_pooling}.

\paragraph{Covariate balance.}  As outlined above, the covariate balance metric
in \Cref{eq:imbalance} is a key diagnostic for CW estimators. Raking explicitly
sets $\imbalance(\r, \W[\cal]) = 0$ for the covariates used in calibration, such
as the margins, but the same metric lets us check imbalance on regressors
\emph{outside} the raking set — e.g., interactions or nonlinear transformations
of the original regressors.

\paragraph{Frequentist variability.} The weights also directly determine the
frequentist variability of the CW estimator $\muhat[\cal]$. Conditional on
$\W[\cal]$, the variance is:
\begin{align}
\var{\p(\Y \vert \W[\cal], \X)}{\muhat[\cal]} =
    \frac{1}{\nsur^2} \sumsur \left(\w[\cal]_i\right)^2 \var{\p(\y \vert \x_i)}{\y_i }.
    \label{eq:freq_var_general}
\end{align}
Ignoring structure in the conditional response variances, the conditional
variance of $\muhat[\cal]$ is minimized when the weights are equal and grows as
they become more dispersed.

\medskip
Our goal is to develop analogs of these diagnostics for Bayesian survey
adjustment methods, which we turn to next.

\subsection{Survey adjustment: Multilevel regression and poststratification}
\label{sec:mrp}

\subsubsection{Overview}
Unlike CW, MrP explicitly models the outcome $\y$ as a function of $\x$,
targeting the conditional mean rather than the density ratio
\parencite{gelman:1997:poststratification,lopez:2022:mrpintro}.  Given a posited
outcome model $\p(\y \vert \x)$ and a corresponding estimate $\yhat(\x) \approx
\expect{\p(\y \vert \x)}{\y}$, the analyst sets $\yhat_j = \yhat(\x_j)$ for each
target $\x_j$ and forms
\begin{align}
    \muhat[\textrm{Generic MrP}] ={}& \meantar \pi_j \yhat_j.\label{eq:mrp_general}
\end{align}
To motivate \Cref{eq:mrp_general}, note that if our estimates $\yhat_j$ are accurate in
the sense that
$$
\expect{\ptar(\x_j)}{ \pi_j \yhat_j} \approx
\expect{\ptar(\x_j)}{ \pi(\x_j) \expect{\p(\y \vert \x_j)}{\y}} =
\expect{\ptar(\x, \y)}{\pi(\x) \y}  = \mu,
$$
then $\muhat[\textrm{Generic MrP}]$ is (nearly) unbiased under sampling from the
target distribution.

The \emph{multilevel regression} step estimates the function $\yhat(\cdot)$ from
the survey data, typically via a multilevel model; the \emph{poststratification}
step averages $\yhat_j$ over the target $\x_j$ in \Cref{eq:mrp_general}.  We
focus on generalized linear models of the form $\expect{\p(\y \vert \x)}{\y} =
\m(\betav^\trans \x)$ for some inverse link $\m(\cdot)$, where $\x$ may contain
non-linear transformations of the observed regressors.
For OLS, $\m(\cdot)$ is the identity and $\yhat^{\ols}_j = \betahat^\trans
\x_j$; for logistic regression, $\mlogit(\z) := 1/(1 + \exp(-\z))$ is the
logistic link and $\yhat_j = \mlogit(\betahat^\trans \x_j)$.

\subsubsection{Bayesian computation}
We focus on \emph{Bayesian hierarchical} estimates of $\yhat_j$.  A Bayesian
estimator posits a likelihood $\p(\y \vert \x, \betav)$, in this case a
parametric model for the full conditional distribution of $\y$ given $\x$; for
any $\betav$, this gives a corresponding estimate $\expect{\p(\y \vert \x,
\betav)}{\y}$.  We specify a possibly hierarchical prior $\p(\betav)$, expressed
marginally over any hyperparameters throughout. While not central to our
discussion, we note that hierarchical priors can induce non-linearity in
posterior estimates as a function of $\Y$ via the implicit estimation of
variance parameters. Finally, we leave the posterior's dependence on $\X$
implicit, since its distribution is ancillary.

Writing $\ell(\y_i | \x_i, \betav) := \log
\p(\y_i \vert \x_i, \betav)$, let $\post$ denote the posterior of
$\betav$ given $\Y$ via Bayes' rule,
\begin{align}
    \post :=
    \frac{\p(\Y \vert \X, \betav) \p(\betav) }
                  {\int \p(\Y \vert \tilde{\betav}, \X)
                  \p(\tilde{\betav})  d\tilde{\betav}} =
    \frac{\exp\left( \sumsur \ell(\y_i \vert \x_i, \betav)\right)
          \p(\betav) }
                  {\p(\Y | \X)} .
        \label{eq:bayes_rule}
\end{align}
We then form
\begin{align}
\yhat^{\bayes}_j = \expect{\post}{\m(\betav^\trans \x_j)}
\quad\textrm{and}\quad
\muhat[\mrp] ={} \meantar \pi_j \yhat^{\bayes}_j.\label{eq:mrp_bayes}
\end{align}
The MrP family also includes optimization-based and other approaches, but we
focus on this Bayesian estimator throughout.

In practice, $\post$ is not available in closed form, and
\Cref{eq:mrp_bayes} is estimated by Markov chain Monte Carlo (MCMC).
Given posterior draws $\betav_k$, $k \in [M]$, we construct MCMC estimates
$$
\yhat^{\textrm{MCMC}}_j = \frac{1}{M} \sum_{k \in [M]} \m(\betav_k^\trans \x_j)
\approx \yhat^{\bayes}_j,
$$
and plug these into \Cref{eq:mrp_general} in place of
$\yhat^{\bayes}_j$. We treat $\muhat[\mrp]$ as if computed from the true
posterior, flagging MCMC issues where relevant.

Returning to the Name Change application, we fit the hierarchical logistic model
of \Cref{sec:intro_example} via \texttt{brms::brm} \parencite{brms}
with default priors, running four MCMC chains of 2,000 iterations
each (500 warmup).  Posterior means $\yhat^{\textrm{MCMC}}_j$ are then
averaged over the ACS target population to form $\muhat[\mrp]$.

\subsubsection{Globally equivalent weights for linear outcome models}
\label{sec:linear_equivalent_weights}

In general, the MrP estimator $\muhat[\mrp]$ depends on the survey responses
$\y_i$ through the posterior $\post$ and through the inverse link function
$\m(\cdot)$.  As a consequence, the mapping $\Y \mapsto \muhat[\mrp](\Y)$ can be
highly nonlinear.  As a result, we cannot in general directly apply the
diagnostics described in \Cref{sec:diagnostics}, which rely on the linearity of
$\muhat[\cal]$ in $\y_i$.

However, an important special case arises when the outcome model is linear. In a
foundational paper, \citet{gelman:2007:surveystruggles} shows that the MrP estimator
with a linear outcome model can be expressed in closed form as a CW estimator
with specific weights, known as \emph{equivalent weights}. For a simple OLS
outcome model \citep[see][]{Park2009}, the equivalent weights are given by:
\begin{align}
\muhat[\ols] ={}& \meantar \pi_j \yhat^{\ols}_j
\\={}& \meansur \w[\ols]_i \y_i
\quad\textrm{where}\quad
\w[\ols]_i := \frac{\nsur}{\ntar} \piv^\trans \Xtar (\X^\trans \X)^{-1} \x_i.
\label{eq:ols_w}
\end{align}
It follows that the map $\Y \mapsto \muhat[\ols](\Y)$ is linear, and that
$\w[\ols]_i$ can be used for all the diagnostics described in
\Cref{sec:diagnostics}.  Moreover, $\muhat[\ols]$ is still linear under ridge
penalization and therefore when $\betavhat$ is the posterior mean of a Bayesian
linear model with multivariate normal prior $\p(\betav)$; see
\Cref{exmp:normal_weights} below for more details. 


\section{MrP Locally Equivalent Weights}

We now turn to defining locally equivalent weights for MrP, beginning with the
general case and then providing closed-form examples. In the following sections,
we show how these weights can be used to construct diagnostics for MrP analogous
to those in \Cref{sec:diagnostics}.

\subsection{Motivation and definition}

To motivate our approach to locally equivalent weights, suppose we want
equivalent weights for OLS, $\w[\ols]$ as in
\Cref{sec:linear_equivalent_weights}, but do not have access to the closed-form
expression. A practical instance of such a case might be black-box software that
computes some linear function of the data, such as a regression tree, but
without providing access to the internal parameters of the model. Importantly,
suppose we can nevertheless repeatedly call this software to construct the
estimate $\muhat[\ols](\Ytil)$ for any $\Ytil$ in a small neighborhood of $\Y$.
We can then use these black box evaluations to compute $\w[\ols]_i$ via the
relation
$$
\w[\ols]_i = \nsur \fracat{\partial \muhat[\ols](\Ytil)}{\partial \ytil_i}{\Ytil = \Y}.
$$
This immediately recovers the OLS weights in \Cref{eq:ols_w}, noting that the
mapping $\Y \mapsto \muhat[\ols](\Y)$ is linear, and thus the Taylor series
expansion of $\muhat[\ols](\Ytil)$ around $\Y$ is exact.

Our key idea is to apply this same approach to MrP, even though the mapping $\Y
\mapsto \muhat[\mrp](\Y)$ is not linear, arguing that this gives the appropriate
notion of a Taylor series approximation to $\muhat[\mrp](\Ytil)$ around $\Y$.

\begin{defn}[MrP locally equivalent weights]\label{defn:mrp_lew}
The \emph{MrP locally equivalent weight} for observation $i$ is
\begin{align}
    \w[\mrp]_i :={}& \nsur \frac{\partial \muhat[\mrp](\Y)}{\partial \y_i}
    ={} \nsur
        \cov{\post}{\nabla_\y \ell(y_i | \x_i, \betav),
        \g(\betav)},
    \label{eq:mrp_w}
\end{align}
where $\g(\betav) := \meantar \pi_j \m(\betav^\trans \x_j)$ is the model's
implied target-population mean at a fixed $\betav$, and $\nabla_\y \ell(\y_i
\vert \x_i, \betav) := \partial \ell(\y_i \vert \x_i, \betav) / \partial \y_i$
denotes the derivative of the log-likelihood contribution with respect to
$\y_i$.
\end{defn}

The formula for the derivative in \Cref{eq:mrp_w} follows immediately from
well-known results in Bayesian local robustness \citep[e.g.~Theorem 1
of][]{giordano:2018:covariances}. Importantly, the weights $\w[\mrp]_i$ can be
easily estimated using MCMC samples from the posterior $\post$ with minimal
additional computation. Further, as we highlight below, for generalized linear
models such as logistic regression, this partial derivative simplifies to
$\x_i^\trans \betav$, which is also easy to compute from standard model output.

\subsection{Justification and interpretation}
While the locally equivalent weights are straightforward to define and compute,
justifying their use and interpreting their meaning requires more care;
technical results follow in the next sections. In short, the weights $\W[\mrp]$
are justified as the unique first-order representation of $\muhat[\mrp]$ near
$\Y$, and this local equivalence is what licenses the CW-style diagnostics
developed in \Cref{sec:variance,sec:balance}.

First, we must make ``local'' precise by extending the definition of $\post$ in
\Cref{eq:bayes_rule} to accommodate $\Ytil$ in a Euclidean neighborhood of $\Y$.
Such $\Ytil$ may not be valid inputs for the log likelihood $\ell(\y \vert \x,
\betav)$, but as long as both $\p(\Ytil)$ and
$\expect{\post[\Ytil]}{\m(\betav^\trans \x_i)}$ are finite for all $\x_i$, the
mapping $\Ytil \mapsto \muhat[\mrp](\Ytil)$ remains well-defined. We assume here
that $\muhat[\mrp](\Ytil)$ is defined and smooth in a neighborhood of $\Y$; we
give precise conditions in \Cref{sec:variance,sec:balance} and discuss the
binary-response case in \Cref{sec:non_local_robustness}.

In general, $\muhat[\mrp]$ is not globally linear in $\Y$, so we cannot hope
that $\muhat[\mrp](\Y) \approx \meansur \w[\mrp]_i \y_i$. What we do have is a
\emph{locally affine} approximation, the first-order Taylor expansion of
$\muhat[\mrp]$ around $\Y$:
\begin{align}
    \muhat[\mrplew](\Ytil) :={}& \muhat[\mrp](\Y) +
        \meansur \w[\mrp]_i (\ytil_i - \y_i)
    \label{eq:mrplew}
\\
\muhat[\mrp](\Ytil)  ={}& \muhat[\mrplew](\Ytil) +
        \resid(\Ytil, \Y),
        \label{eq:mrp_taylor}
\end{align}
where $\resid(\Ytil, \Y)$ is the appropriate residual. For a globally linear
estimator such as $\muhat[\ols]$, the residual vanishes: $\resid(\Ytil, \Y) = 0$
and $\muhat[\mrplew](\Ytil) = \muhat[\ols](\Ytil)$, recovering the OLS weights
of \Cref{eq:ols_w}.
While $\resid(\Ytil, \Y) \ne 0$ in general, under regularity conditions
established below, we show that $\resid(\Ytil, \Y)$ is of order
$\norm{\Ytil-\Y}_2^2$. In other words, when $\Ytil$ is close to $\Y$, the
residuals are small and $\muhat[\mrplew]$ is a locally affine approximation to
$\muhat[\mrp]$.

Under such an approximation, we define below several senses in which we can
justify using the weights $\W[\mrp]$ for diagnostics \emph{as if}
$\muhat[\mrp](\Y)$ were a CW estimator with weights $\W[\mrp]$. In
\Cref{sec:variance}, we justify using the weights to compute frequentist
variance, leveraging recent results for the Bayesian infinitesimal jackknife. In
\Cref{sec:balance}, we justify using the weights to compute covariate balance:
we show that judicious choices of $\Ytil$ produce \emph{nonlinear
generalizations} of more standard covariate balance checks and give conditions
under which the corresponding nonlinearity is appropriately small.

\subsection{Illustration: Closed-form examples}

To build intuition, we now derive locally equivalent weights for two closed-form
Bayesian examples: a conjugate linear model, and logistic regression in the
asymptotic limit. We revisit these examples in the context of covariate balance
in \Cref{sec:balance}.

While illustrative, these closed-form analyses are only possible due to exact
(\Cref{exmp:normal_weights}) or asymptotic (\Cref{exmp:logistic_weights})
linearity.  For more complex models, such as those in our applications in
\Cref{sec:experiments}, no closed form is available, and we instead compute
$\W[\mrp]$ directly from \Cref{eq:mrp_w} using MCMC samples.

\begin{exmp}\label{exmp:normal_weights}
First, consider the conjugate linear model
$$
\p(\y \vert \x, \betav) = \gauss{\betav^\trans \x, \sigma^2}
\quad\textrm{and}\quad
\p(\betav) = \gauss{\zerov, \betacov},
$$
for known $\sigma \ne 0$ and invertible $\betacov$. In this case, the posterior
has a closed form and $\muhat[\mrp](\Y)$ is linear in $\Y$. Building on
\textcite{gelman:2007:surveystruggles}, we show in \Cref{app:simple_normal} that
\begin{align}
\W[\mrp] ={} \frac{1}{\ntar} \X
\left(\frac{1}{\nsur} \X^\trans \X +
      \frac{\sigma^2}{\nsur} \betacov^{-1} \right)^{-1} \Xtar^\trans \piv.
\label{eq:simple_normal_normal_w}
\end{align}
For a tight prior, $\betacov \to \zerov$, the weights $\W[\mrp]$ shrink and
incorporate less of the covariance $\frac{1}{\nsur} \X^\trans \X$. Conversely,
in the flat-prior limit, $\betacov^{-1} \to \zerov$ with $\Xtar = \X$,
$\W[\mrp]$ reduces to the projection of $(\nsur / \ntar) \piv$ onto the column
space of $\X$.
\end{exmp}

\begin{exmp}\label{exmp:logistic_weights}
Next, consider logistic regression in the asymptotic regime.  Suppose
that $\y$ is binary, with
$$
\expect{\p(\y \vert \x, \betav)}{\y} = \mlogit(\betav^\trans \x)
\quad\textrm{where}\quad
\mlogit(z) = 1 / (1 + \exp(-z)).
$$
We assume the prior $\p(\betav)$ is smooth and $\nsur$ is large, so that $\post$
is well-approximated by Bernstein-von Mises
\parencite[Ch.~10.2]{vaart:2000:asymptoticstatistics}. Write $\betahat$ for the
MLE, $\yhat_i := \mlogit(\betahat^\trans \x_i)$ and $\vhat_i :=
\yhat_i(1-\yhat_i)$ for the approximate mean and variance, with analogous
definitions for the target ($\yhat_j, \vhat_j$). Let $\V$ and $\Vtar$ denote
diagonal matrices with survey and target variance estimates on the diagonal,
respectively. As we show in \Cref{app:simple_logistic}, in this case
\begin{align}
\W[\mrp] \approx \frac{1}{\ntar} 
    \X \left(\frac{1}{\nsur} \X^\trans \V \X\right)^{-1} \Xtar^\trans \Vtar \piv,
    \label{eq:logistic_mrplew}
\end{align}
where $\approx$ reflects the delta method and Bernstein-von Mises approximations.

\Cref{eq:logistic_mrplew} mirrors the flat-prior limit of
\Cref{eq:simple_normal_normal_w} ($\betacov^{-1} = \zerov$), with extra factors
of $\V$ and $\Vtar$. The absence of the prior covariance reflects the fact that
\Cref{eq:logistic_mrplew} is computed in the asymptotic limit and thus any prior
influence vanishes. The variance factors $\V$ and $\Vtar$ reflect the fact that
$\muhat[\mrp](\Y)$ is nonlinear in $\Y$, since different $\Y$ gives rise to
different $\V$ and so different $\W[\mrp]$.
\end{exmp}


\section{Formal results}
We now formally justify the use of MrPlew weights for estimating frequentist
variability and for assessing covariate balance.

\subsection{Variance estimation}\label{sec:variance}


We begin with frequentist variability. As we will see, this is conceptually
distinct from Bayesian posterior uncertainty, which captures uncertainty about
the parameter $\betav$ given the data. The frequentist variance, by contrast,
captures the variability of the estimator across repeated samples.

Our proof will rely on some technical assumptions required for consistency of
the Bayesian infinitesimal jackknife \parencite{giordano:2024:bayesij}. To state
these conditions, we need to introduce some notation. Let $\Agrad{k}$ denote the
$k$-th derivative of the log partition function  $\A(\cdot)$ of an exponential
family, and let $\r^{\otimes k}$ denote the $\rdim \times \ldots \rdim$ array of
products of $\r$.   Let $\norm{\cdot}_2^2$ of a multidimensional array denote
the squared Euclidean norm of the stacked array, i.e., the sum of the squares of
the array entries. Finally, let $\dlim$ denote convergence in distribution and
$\plim$ convergence in probability, both with respect to IID samples from
$\psur(\x, \y)$.

\begin{assu}[Canonical exponential family]\label{assu:linear}
    Assume that the likelihood is given by
    a one-parameter natural exponential family with sufficient statistic
    $\y$ and natural parameter $\eta = \betav^\trans \x$.
    (This model may be misspecified.)
    Specifically, the log likelihood for $(\x_i, \y_i)$ is given by
    $$
    \ell(\y_i \vert \x_i, \betav) = \y_i \eta_i - A(\eta_i)
    \textrm{ for }\eta_i := \betav^\trans \x_i,
    $$
    where $A(\cdot)$ is the log partition function and the density is
    assumed to be relative to some fixed base measure on $\y$.
\end{assu}

Popular models satisfying this restriction are generalized linear models with
canonical link functions. We primarily focus on logistic regression; other
examples include Poisson regression and Normal regression with known residual
variance.


Next, we assume that the limit of the
maximum likelihood estimator of $\betav$ exists and is identifiable as $\nsur \to \infty$.
\begin{assu}[Identifiability of the MLE]\label{assu:mle} Assume that $\expect{\psur(\x)}{\A(\betav^\trans
\x)}$ is finite for all finite $\betav$, and define $\ell(\betav) :=
\expect{\psur(\y, \x)}{\ell(\y_i \vert \x_i, \betav)}$. Assume that $\betastar
:= \argmax{\betav \in \betadom} \ell(\betav)$ exists, is unique, and that
$
\info :=
    -\fracat{\partial^2 \ell(\betav)}
            {\partial \betav \partial \betav^\trans}{\betastar}
$ is positive definite.
\end{assu}
If we can exchange integration and differentiation in the definition of
$\ell(\betav)$ (sufficient conditions are given in \Cref{assu:ij} and \Cref{lem:ij} in
\Cref{app:ij}), then the information matrix from \Cref{assu:mle} takes the form
$$
\info = \expect{\psur(\x)}{\Agrad{2}({\betastar}^\trans \x) \x \x^\trans}.
$$
Since $\Agrad{2}({\betastar}^\trans \x) > 0$ by standard properties of
exponential families, positive definiteness of $\info$ reduces to a mild
moment condition on $\x$. This positive-definiteness condition fails if $\x$ lies almost surely in
a proper linear subspace of $\rdom{\rdim}$, but holds whenever the support of
$\x$ spans $\rdom{\rdim}$.

We impose a mild set of conditions on the survey data-generating distribution $\psur(\x, \y)$ and the prior $\p(\betav)$.

\begin{assu}[Regularity for the infinitesimal jackknife]\label{assu:ij_lite}
Under \Cref{assu:linear,assu:mle}, assume that as $\nsur \to \infty$ (with observations IID from $\psur(\x, \y)$), the following stay fixed:
\begin{itemize}
    \item The dimension of $\betav$,
    \item The target observations $\{\x_1, \ldots, \x_{\ntar}\}$, and
    \item The weighting function $\pi(\cdot)$.
\end{itemize}
Additionally assume that:
\begin{itemize}
    \item With probability one under $\psur(\x)$, $\x$ is bounded.
    \item The prior $\p(\betav)$ has bounded support.
    \item Marginally, $\expect{\psur(\y)}{\y^2} < \infty$.
    \item The prior $\p(\betav)$ is proper and has a density that is nonzero and
          four times continuously differentiable in a neighborhood of
          $\betastar$.
\end{itemize}
The boundedness assumptions on $\x$ and $\betav$ can be relaxed to weaker moment conditions; see \Cref{assu:ij} in \Cref{app:ij}.
\end{assu}


We now state the main result that a sample-variance estimator $\mrpvarhat$,
analogous to the CW variance formula in \Cref{eq:freq_var_general}, is
consistent for the frequentist variance of $\muhat[\mrp]$.

\begin{thm}[Infinitesimal jackknife-based variance for MrPlew]\label{thm:ij}
Let \Cref{assu:linear,assu:mle,assu:ij_lite} hold. For $i \in [\nsur]$, let
$\yhat_i := \expect{\post}{\m(\betav^\trans \x_i)}$ and $\varepsilon_i := (\y_i
- \yhat_i)$. Define
\begin{align}
\mrpvarhat := \meansur \left(
    \nsur \w[\mrp]_i \varepsilon_i  - \nsur \overline{\w[\mrp] \varepsilon}
    \right)^2
\quad\textrm{where}\quad
\overline{\w[\mrp] \varepsilon} := \meansur \w[\mrp]_i \varepsilon_i,
\label{eq:mrplew_variance}
\end{align}
and where $\mrpvarhat$ is the sample variance of $\nsur \w[\mrp]_i \varepsilon_i$.
Then, as $\nsur \rightarrow \infty$,
$$
\sqrt{\nsur}\left(\muhat[\mrp] - \muhat[\infty] \right)
\dlim \gauss{0, \mrpvar}
\quad\textrm{and}\quad
\mrpvarhat \plim \mrpvar
$$
for some variance $\mrpvar \ge 0$ and
$\muhat[\infty] = \meantar \pi_j \m({\betastar}^\trans \x_j)$.
\end{thm}
See \Cref{app:ij} for the proof.

The proof of \Cref{thm:ij} operates by showing that $\mrpvarhat$ is
asymptotically equivalent to the infinitesimal jackknife covariance;
\textcite[Theorem 2]{giordano:2024:bayesij} show this is consistent. The bulk of
the proof verifies that the model satisfies the conditions of the IJ consistency
theorem under \Cref{assu:linear,assu:mle,assu:ij_lite}.

\begin{remark}[Failure of variance estimation when the canonical exponential
    family does not hold]\label{rem:assu_linear_necessary} \Cref{assu:linear} is
    essential for \Cref{thm:ij}, though not for \Cref{thm:balance} below: if the
    model's sufficient statistic is not $\y$ alone, then $\mrpvarhat$ is
    generally an \emph{inconsistent} estimator of the frequentist variance. For
    example, if we model $\p(\y \vert \x, \betav, \sigma) = \gauss{\x^\trans
    \betav, \sigma^2}$ with unknown variance $\sigma^2$, the sufficient
    statistic is $(\y, \y^2)$, \Cref{assu:linear} is not satisfied, and
    $\mrpvarhat$ is inconsistent; see \Cref{exmp:gaussian_glm} in
    \Cref{app:simple_examples} for further discussion of this counterexample.
    This is not a problem for variance estimation itself, since one can use the
    infinitesimal jackknife covariance instead. However, this counterexample
    emphasizes that $\W[\mrp]$ does not automatically inherit CW-style
    interpretations for variance.
\end{remark}

\subsection{Covariate balance}\label{sec:balance}
\subsubsection{Motivation and intuition}

We now turn to covariate balance. The key idea is to reframe the standard
balance check as a sensitivity diagnostic that assesses whether the weighting
procedure would have been able to detect a small, directed perturbation of the
response variable.  This reframing admits a local interpretation to which we can
apply the Taylor series approximation in \Cref{eq:mrp_taylor}, justifying the
use of $\W[\mrp]$ in a local analogue of the standard CW balance check.



To build intuition, consider a perturbation of a continuous response $\y$ of the form:
\begin{align}
\ytil = \y + \delta \r(\x)
\label{eq:ytil}
\end{align}
for some measurable function of the covariates $\r = \r(\x)$ and small $\delta$.
For continuous $\y$, the additive construction in \cref{eq:ytil} can produce
valid observations. For binary $\y$ or other bounded responses, the additive
construction of \cref{eq:ytil} no longer produces a valid response in general.
We will \emph{define and analyze} $\ytil$ of \cref{eq:ytil} through the
generalized posterior below, deferring the question of which binary processes
approximately satisfy \cref{eq:ytil_expectation} to
\Cref{sec:parametric_bootstrap}.

The corresponding change in the (unknown) target mean is:
$$
\begin{aligned}
\meantar \pi_j \ytil_j - \meantar \pi_j \y_j = 
    \underbrace{ \delta \meantar \pi_j \r_j}_{\text{actual change}}.
\end{aligned}
$$
For linear calibration weights, the corresponding change in the estimator is:
\begin{align}
\meansur \w[\cal]_i \ytil_i - \meansur \w[\cal]_i \y_i = 
    \underbrace{\delta \meansur \w[\cal]_i \r_i}_{\text{inferred change}}.
    \label{eq:cal_balance_change}
\end{align}
Subtracting the inferred change from the actual change and dividing by $\delta$
immediately recovers the standard covariate balance check in
\Cref{eq:imbalance}:
\begin{align}
\frac{1}{\delta} \imbalance(\r, \W[\cal])
={}&
    \meantar \pi_j \r_j -
    \meansur \w[\cal]_i \r_i \eqcheck 0.
\label{eq:imbalance_pert}
\end{align}


The main advantage of this reframing is that we can immediately extend this
argument to nonlinear MrP via its Taylor expansion. Specifically, as long as we
can sensibly define $\muhat[\mrp](\Ytil)$ and the Taylor expansion in
\Cref{eq:mrp_taylor} is valid, we have an analogous \emph{local} change:
\begin{align}
    \muhat[\mrp](\Ytil) - \muhat[\mrp](\Y) ={}&
        \meansur \w[\mrp]_i (\ytil_i - \y_i) + \resid(\Ytil, \Y)
    \nonumber
    \\={}& \underbrace{\delta \meansur \w[\mrp]_i \r_i}_{\text{inferred change}} + 
    \resid(\Ytil, \Y).
\label{eq:balance_linear}
\end{align}
\Cref{thm:balance} below shows that, under regularity conditions, the residual
term $\resid(\Ytil, \Y)$ is of order $\delta^2$, uniformly over a Donsker class
of functions $\r$. After dividing by $\delta$ and dropping the residual term, we
can therefore form a \emph{local} balance check that replaces $\w[\cal]$ with
$\w[\mrp]$ but otherwise mirrors the linear version.

%


\subsubsection{Main result}

\begin{defn}[Perturbed likelihood and generalized posterior]\label{defn:balance}
    Let $\deltamax > 0$ denote an
    upper bound on the perturbation.  Under \Cref{assu:linear}, for all $\delta
    \in \deltadom$, we formally define the perturbed likelihood

    $$
    \begin{aligned}
    \ell(\y \vert \x, \betav; \delta \r)
    ={}& (\y + \delta \r(\x)) \x^\trans \betav - A(\betav^\trans \x)
    \\={}& \ell(\y \vert \x, \betav) + \delta \r(\x) \x^\trans \betav \\
    \p(\Y \vert \betav; \delta \r) :={}&
        \exp\left( \sumsur \ell(\y_i + \delta \r_i \vert \x_i, \betav) \right).
    \end{aligned}
    $$
    We then define the \emph{generalized posterior}
    $$
    \postd :={}
    \frac{\p(\Y \vert \X, \betav; \delta \r) \p(\betav) }
                  {\int \p(\Y \vert \X, \tilde{\betav}; \delta \r)
                  \p(\tilde{\betav})  d\tilde{\betav}}
    $$
    when the denominator is finite, and let
    $\muhat[\mrp](\Ytil) := \expect{\postd}{\g(\betav)}$ ($\g(\betav)$ is defined in \cref{eq:mrp_w}).
\end{defn}

\begin{assu}[Donsker class]\label{assu:balance}
    Let $\rset$ denote a Donsker class of $\psur(\x)$-measurable functions for which
    $\sup_{\r \in \rset} \expect{\psur(\x)}{\norm{\x \r(\x)}_2^2} < \rmax^2 < \infty$.
\end{assu}
Donsker classes are classes of functions restricted enough to obey uniform laws
of large numbers \parencite[Chapter 19]{vaart:2000:asymptoticstatistics}.
Readers less familiar with Donsker classes can instead imagine that $\rset$ is a
finite set of functions without losing any essential understanding.

\begin{thm}\label{thm:balance}
Take $\Ytil = \Y + \delta \R$, where $\R = (\r(\x_1), \ldots, \r(\x_{\nsur}))^\trans$. Under
    \Cref{assu:linear,assu:mle,assu:ij_lite,assu:balance}, with probability
    approaching one as $\nsur \rightarrow \infty$, $\muhat[\mrp](\Ytil)$ exists,
    and satisfies
    $$
        \sup_{\r \in \rset}
        \frac{1}{\delta}
        \abs{
            \muhat[\mrp](\Ytil) - \muhat[\mrp](\Y) -
            \meansur \w[\mrp]_i \r_i
        } ={} O(\delta)
    \quad\textrm{as }\delta \rightarrow 0.
    $$
\end{thm}

See \Cref{app:balance} for a proof. \Cref{thm:balance} shows that, to leading
order in $\delta$, the change in the perturbed MrP estimator per unit $\delta$
is $\meansur \w[\mrp]_i \r_i$. The resulting MrP balance check $\imbalance(\r,
\W[\mrp])$ is thus the natural analogue of \Cref{eq:imbalance_pert}. The fact
that \Cref{thm:balance} holds uniformly over a wide class of functions
$\r(\cdot)$ justifies searching over a set of prespecified covariates to check
for imbalance, as is commonly done in practice.

Unlike \Cref{thm:ij}, the assumption in
\Cref{assu:linear} that $\y$ is a sufficient statistic of the model is not
essential for \Cref{thm:balance}, and an inspection of the proof
of \Cref{lem:bayesij_balance} in \Cref{app:balance} will show that 
versions of \Cref{thm:balance} should hold for a much broader class of
models, including generic multivariate exponential families.  As with
\Cref{thm:ij}, the boundedness conditions in \Cref{assu:ij_lite}
can be relaxed to a set of technical moment conditions (see
\Cref{assu:ij} in \Cref{app:ij}). 

\subsubsection{Closed-form examples}\label{sec:examples}

We next apply the balance check of \cref{thm:balance} to the two closed-form
Bayesian examples from \Cref{sec:examples}.

\begin{exmp}[Conjugate normal models]\label{exmp:normal_balance} We continue
\Cref{exmp:normal_weights}, the conjugate normal model, focusing on the
consequences of prior shrinkage for covariate balance; \Cref{app:simple_normal}
gives full details. Applying \Cref{thm:balance} and taking $\r(\x)$ to be the
components of $\x$, we have the following imbalance:
\begin{align}
\frac{1}{\delta} \imbalance(\x, \W[\mrp]) ={}&
\frac{1}{\ntar} \piv^\trans \Xtar 
    \left(
        \id_{P} - \left( 
            \frac{1}{\nsur} \X^\trans \X + 
            \frac{\sigma^2}{\nsur} \betacov^{-1}
        \right)^{-1}
        \left( \frac{1}{\nsur} \X^\trans \X \right)
    \right)
\label{eq:simple_linear_imbalance}
\end{align}
When $\betacov^{-1} \ne \zerov$, $\muhat[\mrp]$ only \emph{approximately}
balances $\x$. As $\betacov^{-1} \to \zerov$, however, the matrix product inside
\cref{eq:simple_linear_imbalance} converges to $\id_P$, so the imbalance
converges to zero. This reflects existing results on the implied imbalance of
ridge regression; see, for example, \citet{bruns2025augmented}.

This setup is formally equivalent to \emph{soft calibration}
\citep{gao2022soft}, where $\betav$ is treated as a random effect with
distribution $\p(\betav)$. Indeed, as with soft calibration, we show in
\cref{app:simple_normal} that $\W[\mrp]$ minimizes the expected mean squared
error of $\frac{1}{\ntar} \piv^\trans \Y_{\tar} - \frac{1}{\nsur} {\W}^\trans
\Y$ when $\Y$ and $\Y_{\tar}$ are generated from a shared draw of $\betav \sim
\p(\betav)$.
\end{exmp}


\begin{exmp}[Asymptotic logistic regression]\label{exmp:logistic_balance}

We continue \cref{exmp:logistic_weights}, the asymptotic logistic
regression setting; \Cref{app:simple_logistic} gives full details.
Importantly, we show that logistic regression generally balances the
\emph{variance-weighted} covariates, but not the covariates themselves.
This mirrors the classical result for model-assisted generalized
regression estimation under a logistic link \citep{firth1998robust}.
Consider variance-weighted regressors $\r(\x) = \v(\x) \x$, where
$\v(\x) = \var{\p(\y \vert \x)}{\y}$ is the conditional variance of
$\y$, estimated by $\vhat_i$ at $\x_i$. Then:
$$
\frac{1}{\delta} \imbalance(\v(\x) \x, \W[\mrp]) = \zerov.
$$
By contrast, logistic regression does not balance the covariates $\x$ themselves:
$$
\frac{1}{\delta} \imbalance(\x, \W[\mrp])
=
\frac{1}{\ntar} \piv^\trans \Xtar
\left(\id_P - \left(\frac{1}{\nsur} \X^\trans \V \X\right)^{-1} \frac{1}{\nsur} \X^\trans \X \right)
\ne \zerov. \textrm{ (in general)}
$$
Unlike the conjugate normal model above, this imbalance is not due to prior
influence, which vanishes in the asymptotic regime. Rather, it arises from a
mismatch in the scale of the model and the perturbation:
\Cref{eq:ytil_expectation} defines perturbations on the $\y$ scale, whereas
logistic regression operates on the log-odds scale. 

Somewhat surprisingly, although it is the presence of $\V$ that causes the
imbalance, this imbalance is not fundamentally driven by nonlinearity of the map
$\muhat[\mrp](\Y)$ as a function of $\Y$.  In \Cref{app:logistic_linearity}, we
consider the special case in which the density ratio is linear in $\x$ and show
that the logistic regression MrP estimator is then approximately linear in $\Y$
for large $\nsur$. Even in this case, however, the imbalance of $\x$ remains
nonzero.
\end{exmp}

\subsubsection{Using a parametric bootstrap to produce restricted responses}
\label{sec:parametric_bootstrap}

As we discuss above, when the response $y$ is binary or otherwise bounded, the
additive construction in \cref{eq:ytil} cannot be valid in general, and
yet it forms the basis for our theoretical results in \cref{thm:balance}.  In
this section, we bridge the gap and discuss how to generate valid
response vectors that \emph{approximately reproduce} imbalance identified
for the continuously perturbed $\ytil$.

The key idea is that the perturbation \cref{eq:ytil} can be a good approximation
to a restricted random variable $\ytiltil$ with \emph{conditional expectation}
shifted relative to $\y$:
\begin{align}
\expect{\p(\ytiltil \vert \x)}{\ytiltil} ={}& 
    \expect{\p(\y \vert \x)}{\y} + \delta \r(\x).
\label{eq:ytil_expectation}
\end{align}
Let $\Ytiltil$ denote the corresponding vector of perturbed responses.
\Cref{eq:ytil_expectation} makes sense since, for any model satisfying
\cref{assu:linear}, the posterior $\post[\Ytiltil]$ depends on $\Ytiltil$ only
through
$$
\betav \mapsto \betav^\trans \meansur \y_i \x_i \approx
    \betav^\trans \expect{\psur(\x)}{\expect{\p(\y \vert \x)}{\y} \x}
    \textrm{ for large }\nsur.
$$
So if $\ytiltil$ satisfies \cref{eq:ytil_expectation} then we can expect
$\post[\Ytiltil] \approx \postd[\Ytil]$, and we can thus think of
\cref{thm:balance} as approximating the behavior of $\post[\Ytiltil]$.  

We briefly describe a simple procedure based on a perturbed parametric bootstrap
to generate $\Ytiltil$ satisfying \cref{eq:ytil_expectation}, while still being
reasonably close to $\Y$; \Cref{app:binary} gives the full details. We begin
with an estimate $\mhat(\x_i)$ of $\expect{\p(\y\vert\x)}{\y}$, as one would for
performing the parametric bootstrap \parencite{efron:1994:bootstrap}. We then
draw $\ytiltil$ with mean $\mhat(\x_i) + \delta \r(\x_i)$, which requires
restricting $\delta$ small enough that $\mhat(\x_i) + \delta \r(\x_i) \in [0,1]$
for all $i$.  The final step is to correlate $\ytiltil_i$ with $\y_i$ while
maintaining the desired marginals, in order to keep $\Ytiltil - \Y$ as small as
possible; a procedure for doing so is described in \cref{app:binary}.  In
\cref{sec:non_local_robustness} we apply this technique to our experiments and
find that the resulting binary vectors match the corresponding predicted results
closely.

\subsubsection{Subgroup contribution}\label{sec:subgroup_contribution}
\def\s{s}
\def\stateset{\mathcal{I}_\s}

An important special case of the general perturbation results is for measuring
the contribution of particular subgroups to the final estimate. Consider the
Name Change application from \Cref{sec:intro_example}, where we group
respondents by education level: less than a college degree, a college degree, or
more than a college degree. If we index these groups from $s=1,\ldots,S$ (here,
$S = 3$), and let $\stateset$ denote the set of indices $i$ that are from
subgroup $\s$, then we can rewrite any CW estimator as
\begin{align*}
\muhat[\cal](\Y) ={}&
    \frac{1}{\nsur}
    \sum_{\s=1}^{S} \sum_{i \in \stateset} \w[\cal]_i \y_i.
\end{align*}
Here, $\w[\cal]_s := \sum_{i \in \stateset} \w[\cal]_i$ is the total weight
given to subgroup $\s$, and measures how much subgroup $\s$ contributes to the
overall estimate.  Comparing how $\w_\s$ varies from subgroup to subgroup can
provide intuition to the analyst about how the data is being used. 

For MrPlew weights, the quantity $\w[\mrp]_s := \sum_{i \in \stateset}
\w[\mrp]_i$ is precisely the left-hand side of the covariate balance check for
the indicator $\z_{is} = \ind{\x_i \textrm{ is in subgroup }\s}$.  Here,
$\z_{is}$ takes value $1$ if the survey observation $i$ is from subgroup $\s$,
and $0$ otherwise.   Thus, $\w[\mrp]_s$ admits an interpretation similar to that
of covariate balance: if the responses in subgroup $\s$ were all to increase in
expectation by a small $\delta$, we would expect $\muhat[\mrp]$ to increase by
$\delta \w[\mrp]_s / \nsur$.  This provides an intuitively meaningful measure of
the ``importance'' of subgroup $\s$ in the estimator $\muhat[\mrp]$.  This
interpretation is supported by \Cref{thm:balance} without modification, simply
by taking $\r(\x)$ to be the indicator of the subgroup categories.

\subsubsection{Why not perturb the log odds?}\label{sec:logodds}
A key feature of our generalized balance checks is that the perturbation to the
response in \Cref{eq:ytil,eq:ytil_expectation} is defined in the space of
responses, $\y$, rather than in the space of log-odds,
${\mlogit}^{-1}(\expect{\p(\y \vert \x)}{\y})$.  Before proceeding, we briefly
motivate and discuss this decision.

As a motivating question, in light of the failure of logistic regression to
balance the regressors as shown in \cref{exmp:logistic_balance}, one might
wonder why we do not use the following perturbation rather than
\cref{eq:ytil_expectation} to define our perturbations to the data:
\begin{align}
\expect{\p(\y \vert \x)}{\ytil} =
\mlogit\left(
{\mlogit}^{-1}(\expect{\p(\y \vert \x)}{\y}) + \delta' \r(\x)
\right).
\label{eq:ytil_logodds}
\end{align}
Since
$$
\frac{\partial}{\partial \delta'} \expect{\p(\y \vert \x)}{\ytil} = \v(\x) \r(\x),
$$
adding a small $\r(\x)$ to the log odds is equivalent to adding a small $\v(\x)
\r(\x)$ to the expectation.  Since the logistic regression models the log odds
as a linear combination of $\x$, it is perfectly able to capture perturbations
of the log odds in the direction of $\x$, and such directions correspond to
perturbations of the form $\v(\x) \x$ in the space of expectations.

We might argue that \cref{eq:ytil_expectation} is easier for a practitioner to
think about intuitively.  But a more fundamental reason is that the perturbation
\cref{eq:ytil_logodds} is \emph{defined in terms of a particular model}, a fact
which makes it difficult to meaningfully compute and compare perturbations in a
model-agnostic way.  For one, in order to actually induce the perturbation
\cref{eq:ytil_logodds}, one needs to have an estimate of $\v(\x)$ that can only
be provided by one of the very models we are trying to interrogate.  Different
logistic regression models, with different priors, will define different
perturbations to the data.  Also, the perturbation \cref{eq:ytil_logodds}
privileges the (fairly arbitary) log-odds parameterization when defining the
data perturbation.  Other link functions, such as the normal distribution
function for probit regression will exhibit unfavorable performance under the
perturbation \cref{eq:ytil_logodds}.  The perturbation \cref{eq:ytil_logodds}
appears to preclude direct comparison between methods that are not guaranteed to
produce valid log-odds estimates, such as OLS and raking.  
Thus, we believe that the failure of logistic regression to balance $\x$
asymptotically should be taken as a warning about logistic regression, not a
failure of our balance checks.  


\section{Applying MrPlew}\label{sec:experiments}

We demonstrate how to use MrPlew for three existing MrP analyses: (1) an
analysis that extrapolates a Twitter survey of name changes after marriage to
the entire US population
\parencite{alexander:2019:namechange,cohen:2019:namechange}; (2) an example of
correcting sampling bias in a national survey on support for same-sex marriage
\parencite{lax:2009:gay,kastellec:2010:laxmrp}; and (3) a textbook example
analyzing the 2020 US presidential election \parencite{alexander:2023:telling}.

In each case, we compare the original MrP analysis to raking on marginals of
coarsened versions of the same regressors used in the MrP analysis
({\small{\texttt{survey::calibrate}}} from \textcite{lumley:2025:survey}).  
Following \textcite{debell:2009:raking}, the raking covariates were coarsened so
that no marginal category contains fewer than 5\% of the survey observations.
For example, when MrP regressions included US state as a regressor, state
indicators were coarsened to geographic region (west, south, northeast, and
midwest) for raking. Finally, for simplicity we removed the small number of
observations with missing regressor values.%
\footnote{MrPlew weights remain computable under Bayesian data imputation,
provided the posterior covariance in \Cref{eq:mrp_w} incorporates the additional
variability due to imputation in the $\nabla_\y \ell(\y_i \vert \x_i, \betav)$
term.}


\subsection{Application descriptions}

We first describe the two additional datasets and analyses that we reproduce; \Cref{sec:intro_example} above gives additional details for the ``Name Change'' analysis.


\paragraph{``Same-Sex Marriage'' analysis.}
Our next analysis is based on the classic MrP primer from \textcite{kastellec:2010:laxmrp},
which provides code and data that is amenable to re-analysis with MCMC
\parencite{kastellec:2024:data}.
\citet{kastellec:2010:laxmrp} analyze five consistently-coded national
polls from 2004 that surveyed support for same-sex marriage; the authors call these polls a ``megapoll.'' The response variable $\y$
is a binary indicator where $1$ encodes support for same-sex marriage, and $0$
encodes either opposition or no expressed opinion; 
see the appendix of
\textcite{lax:2009:gay} for more details.

Following \textcite{kastellec:2010:laxmrp}, we then use MrP to calibrate the
megapoll responses to individual states using a 5\% Public Use Microdata Sample
(PUMS) from the 2000 US census. We fit an MCMC version of the following model,
which the original paper fit using marginal maximum likelihood:
$$
\begin{aligned}
\y \sim{}& \texttt{logit}\Big(~\texttt{(1 | race.female) + (1 | age.cat) + 
  (1 | edu.cat) + (1 | age.edu.cat) +}
  \\{}&\qquad\quad\;\; \texttt{(1 | state) + (1 | region) + (1 | poll) +
  p.relig.full + p.kerry.full}~\Big)
\end{aligned}
$$
See \textcite{kastellec:2010:laxmrp}, as well as the original research 
paper from \textcite{lax:2009:gay} for additional context for the surveys
and regressor definitions.
Following the ``MrP Primer'' chapter of \textcite{mastny:2018:study},
which re-analyzes the same data using \texttt{brms}, we set
standard Gaussian priors for each of the fixed effects and
the standard deviations.\footnote{In \texttt{brms}, which
follows conventions from the \texttt{stan} software package, the 
prior is a truncated half-normal for the standard deviations.}

Our main illustration uses MrP to predict support for same-sex marriage in
California, which is an example of using MrP for small area
estimation.\footnote{When forming predictions, we set the \texttt{poll} random
effect to zero.} In \Cref{sec:partial_pooling}, we also predict support in
Missouri, as a contrast to the estimate for California.
To estimate raking weights for California, we eliminated or coarsened
interactions that occurred less than 5\% of the time in either the survey or
target population. In the end, we used the following for raking: gender,
education level, age category, race (white, black, or other), white /
non-white interacted with gender, and age category interacted with secondary /
no secondary education.


\paragraph{``Election Forecasting'' analysis.}
Our third illustration is based on an analysis of the 2020 US presidential election
taken from
\citet[][Ch. 6, 8, and 16]{alexander:2023:telling}.  The survey dataset is from the Nationscape
project \parencite{tausanovitch:2022:nationscape}, which combines a large number
of surveys conducted between July 2019 and January 2021. The response variable $\y$
is a binary indicator where $1$ encodes support for Joe Biden in the 2020 US
presidential election, and $0$ encodes support for Donald Trump.  

Our objective with MrP is to adjust for the fact that the original survey is a convenience sample that is potentially unrepresentative of the entire US population.  The population dataset is taken from the
2019 American Community Survey (ACS) dataset, accessed through IPUMS
\parencite{ipumsusa}, and is selected as a proxy for the demographic profile of the entire US population.

The regressors are gender (encoded as male or female), four age groups in
roughly 15-year bins, three levels of education, and state.  We fit the
following hierarchical logistic regression:
$$
\begin{aligned}
\y \sim{}& \texttt{logit}\Big(~ \texttt{gender + (1 | age\_group) + (1 | state) + (1 | education\_level)} ~\Big),
\end{aligned}
$$
with normal priors for the scale and intercept terms.
For raking, we coarsened the states to regions as in the Name Change analysis
above.



\subsection{Comparing MrPlew and raking weights}

\begin{table}[tb]
\centering 
\ExperimentTable{}
\caption{High-level descriptions of the three applications.}\label{tab:experiments}
\end{table}

\Cref{tab:experiments} gives a high-level summary of the three main
applications, which exhibit different relationships between raking, MrP, and the
uncorrected survey mean $\overline{\y}$. In the Election Forecasting
application, the survey mean, raking, and MrP are all similar; in the Name
Change application, raking and MrP are similar to one another but quite
different from the survey mean; and in the Same-Sex Marriage application, raking
and the survey mean are similar, but MrP is very different.

\WeightsPlot{}

\Cref{fig:weightsplot} shows the corresponding MrPlew and raking weights. A
qualitative inspection of the weights largely aligns with the comparisons across
point estimates. In the Election Forecasting application, the MrP and raking
weights are broadly similar. By contrast, the MrP and raking weights display
very different patterns in the other two applications. Using the tools developed
above, we will assess whether these differences are meaningful.

\subsection{Frequentist variance}\label{sec:variance_experiments}
\VariancePlot{}

The weights in \Cref{fig:weightsplot} suggest that the frequentist variance of
MrP and raking may be similar in the Election Forecasting example, and that the
MrP variance may be higher than the raking variance in the Same-Sex Marriage
example; for the Name Change application, this is harder to assess visually
given the extreme outliers for raking.

\Cref{fig:varplot} shows the frequentist variance estimates for MrP from
\Cref{thm:ij}, compared to the corresponding raking-based variance estimate,
which we compute as:
$$
\mrpvarhat_{\cal} :=
\meansur \left(
    \nsur \w[\cal]_i \varepsilon_i  - \nsur \overline{\w[\cal] \varepsilon}
    \right)^2
\quad\textrm{where}\quad
\overline{\w[\cal] \varepsilon} := \meansur \w[\cal]_i \varepsilon_i,
$$
for $\varepsilon_i = (\y_i - \yhat_i)$ defined in \Cref{thm:ij} and used in our estimate
of $\mrpvarhat$.

Importantly, the frequentist standard error estimates in \Cref{fig:varplot}
enable an apples-to-apples comparison between MrP and raking, since these
estimates capture uncertainty under repeated sampling for both estimators. More
generally, frequentist variances and Bayesian posterior variances only coincide
asymptotically under posterior concentration and correct specification;
otherwise, these two quantities generally differ, especially in the presence of
weakly estimated random effects
\parencite{kleijn:2012:bvm,giordano:2024:bayesij}.

\Cref{fig:varplot} shows three distinct patterns. For the Election Forecasting
application, the estimated variances are quite close between MrP and raking,
consistent with the qualitative inspection of the weights in
\Cref{fig:weightsplot}. For the Name Change application, by contrast, the
extreme raking weights lead to higher variance for raking than MrP, likely
due to the fact that MrP uses more information than raking, and so
is able to produce a greater variety of weights. Finally, for
the Same-Sex Marriage application, we see the opposite pattern, with nearly
triple the frequentist standard deviation for MrP versus raking. 


Finally, \Cref{app:experiments} assesses the accuracy of the MrPlew variance
estimates by bootstrapping the MCMC procedure
\parencite{huggins:2022:bayesbag,giordano:2024:bayesij}. We consider parametric
and nonparametric bootstraps; both confirm that the standard deviation of
$\muhat[\mrp]$ across bootstrap draws closely matches the $\mrpvarhat$ estimates. 

\subsection{Covariate balance}\label{sec:balance_experiments}

Next, we use the results in \Cref{sec:balance} to examine differences in the implied covariate balance between MrPlew and
raking weights. For each application, we computed
covariate balance for each regressor used in raking, which are all binary or discrete, as well as all two-way interactions of those regressors. \Cref{app:experiments} gives complete results.

\begin{figure}[tbp]
\centering
\begin{subfigure}[t]{0.98\linewidth}
\centering
\includegraphics[width=\linewidth,height=0.441\linewidth]{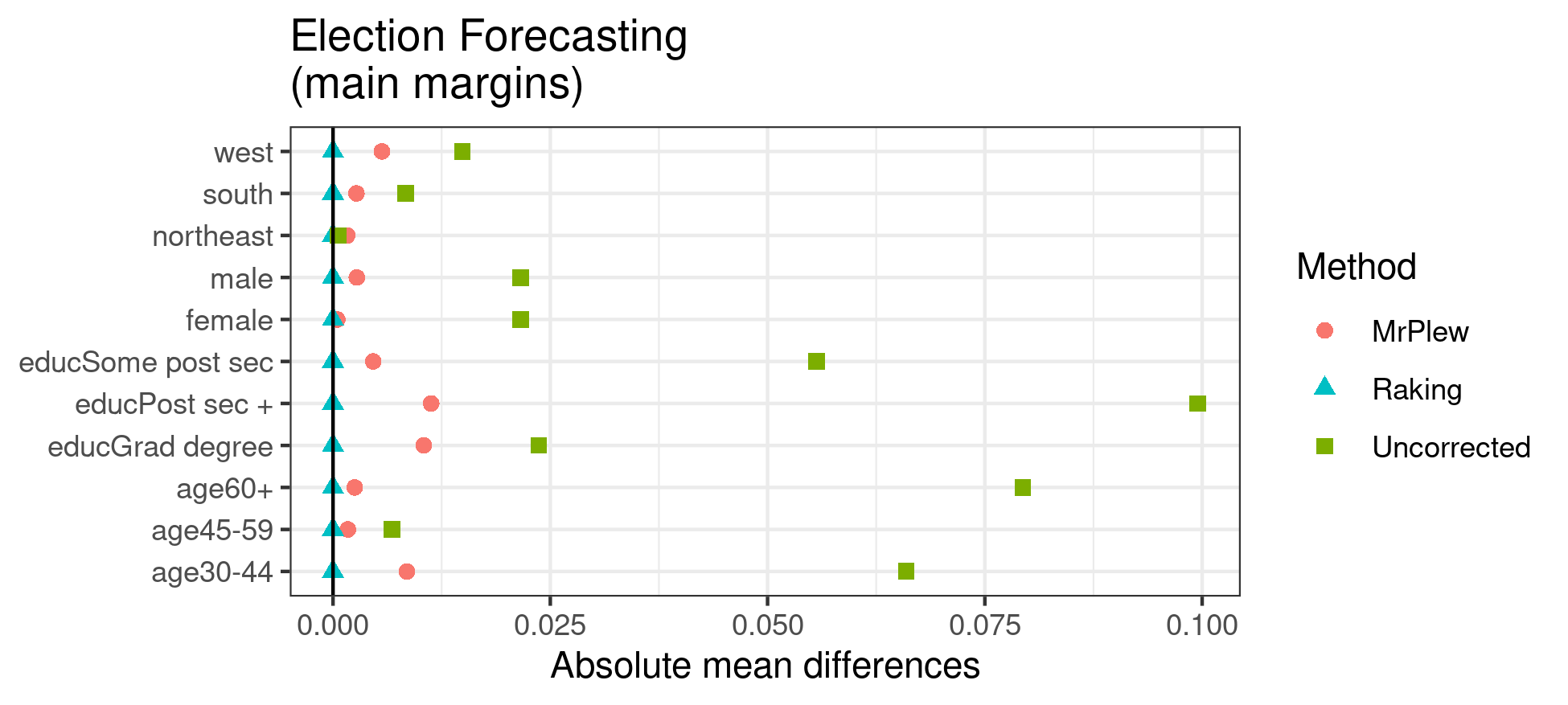}
\caption{Election Forecasting}
\label{fig:storiesbalanceplot}
\end{subfigure}

\vspace{0.75em}

\begin{subfigure}[t]{0.98\linewidth}
\centering
\includegraphics[width=\linewidth,height=0.441\linewidth]{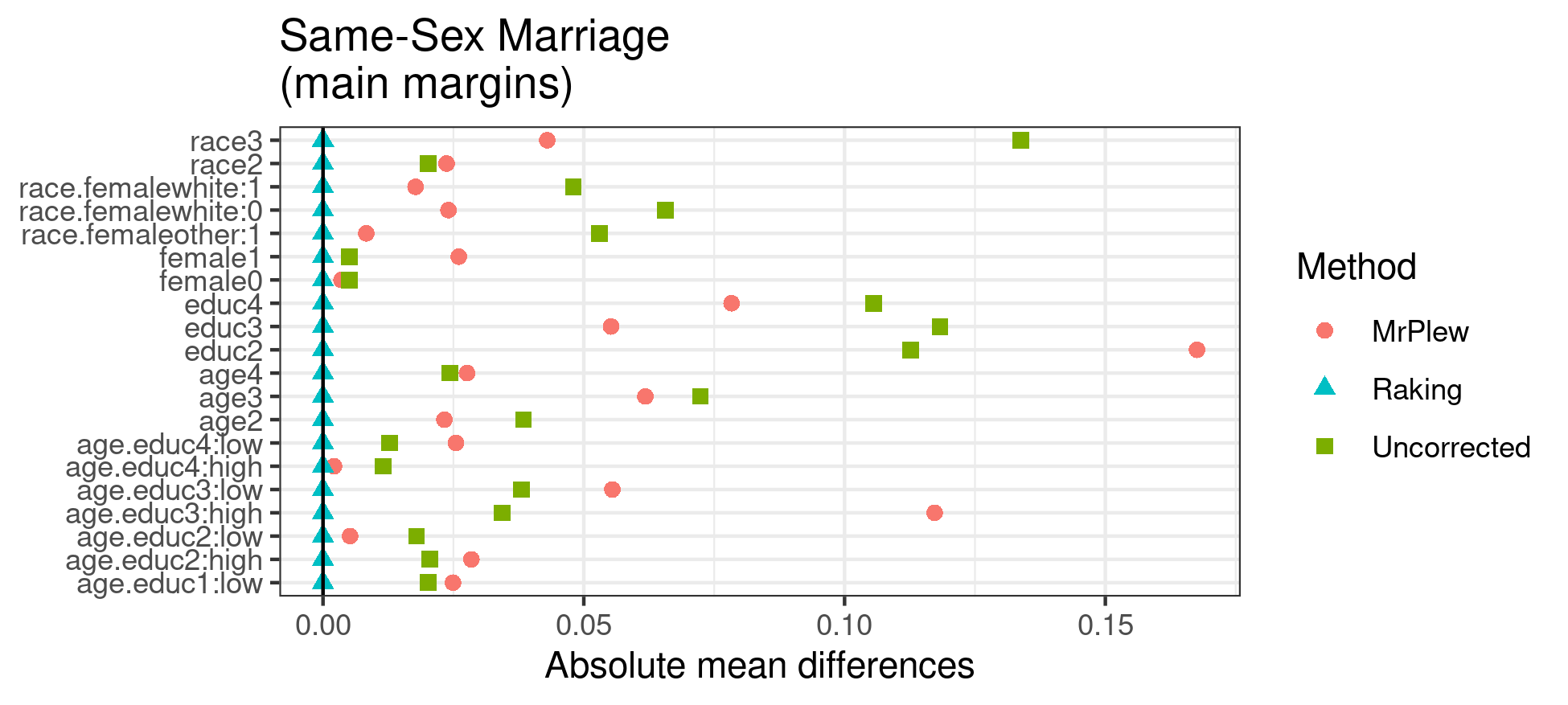}
\caption{Same-Sex Marriage}
\label{fig:laxphilipsbalanceplot}
\end{subfigure}
\caption{Implied covariate balance on raking marginals for the Election Forecasting and Same-Sex Marriage applications.}
\label{fig:balance_marginals}
\end{figure}


\Cref{fig:balance_marginals} shows the implied covariate balance for the
marginal regressors in the Same-Sex Marriage and Election Forecasting
applications; we discuss the covariate balance for the Name Change application
in \Cref{sec:intro_example}. As expected, raking exactly balances the marginal
covariates in both examples. For the Election Forecasting application, MrPlew
also achieves excellent (if not exact) covariate balance on the marginals. For
the Same-Sex Marriage application, however, the implied covariate balance for
MrP appears substantially worse than the \emph{uncorrected} covariate balance
across several age and education levels. We assess whether these local
imbalances translate to meaningful sensitivity in $\hat\mu^{\mrp}$ in
\Cref{sec:non_local_robustness}.

\Cref{app:experiments} includes additional balance plots reporting selected
two-way interactions of the regressors used in raking.\footnote{Due to the large
number of resulting balance checks, we assess covariate balance for interactions
that occurred in at least 5\% of both the survey and target populations. To
preserve space, we only plot interactions in which either MrPlew or raking
weights exhibited some minimal degree of imbalance.} For the Election
Forecasting application, both raking and MrP yield good balance on the
interactions as well. For the Same-Sex Marriage application, however, raking
largely balances two-way interactions but MrP again fails to achieve good
balance.




\subsection{Subgroup contribution}\label{sec:partial_pooling}

Finally, we apply the results in \Cref{sec:subgroup_contribution} to assess the
contribution of each subgroup to the final MrP estimate. We begin with the Name
Change application from \Cref{sec:intro_example}. As \Cref{fig:introplot} shows,
there is substantial imbalance in the interaction between education level and
decade married, and this interaction is plausibly associated with the outcome of
interest. \Cref{fig:alexanderpoolingplot} shows the subgroup contribution for
this interaction, comparing the raking and MrP estimates. The right-hand side
shows that the subgroup with negative weights is precisely the subgroup with the
largest imbalance in \Cref{fig:introplot}, suggesting that the imbalance in this
interaction could contribute to the difference between the MrP and raking
estimates for this application. We explore this interaction further in
\Cref{sec:alexander_add_regs}.

\AlexanderPoolingPlot{}

\begin{figure}[tb]
\centering
\begin{subfigure}[t]{0.98\linewidth}
\centering
\includegraphics[width=\linewidth,height=0.294\linewidth]{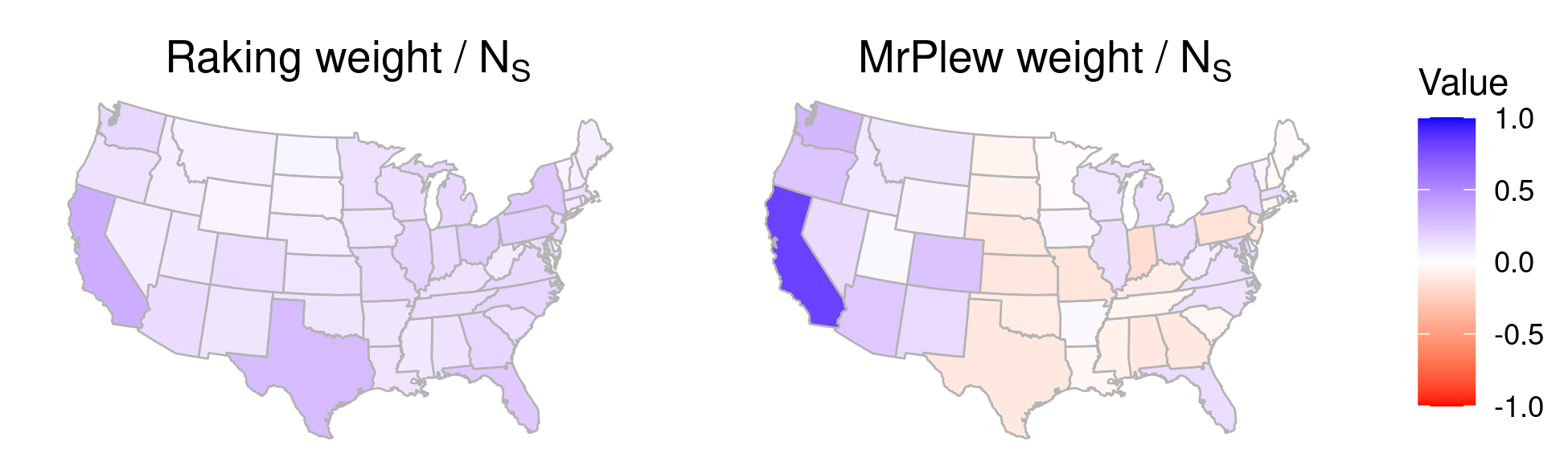}
\caption{Target: California}
\label{fig:laxphilipspoolingplot_ca}
\end{subfigure}

\vspace{0.75em}

\begin{subfigure}[t]{0.98\linewidth}
\centering
\includegraphics[width=\linewidth,height=0.294\linewidth]{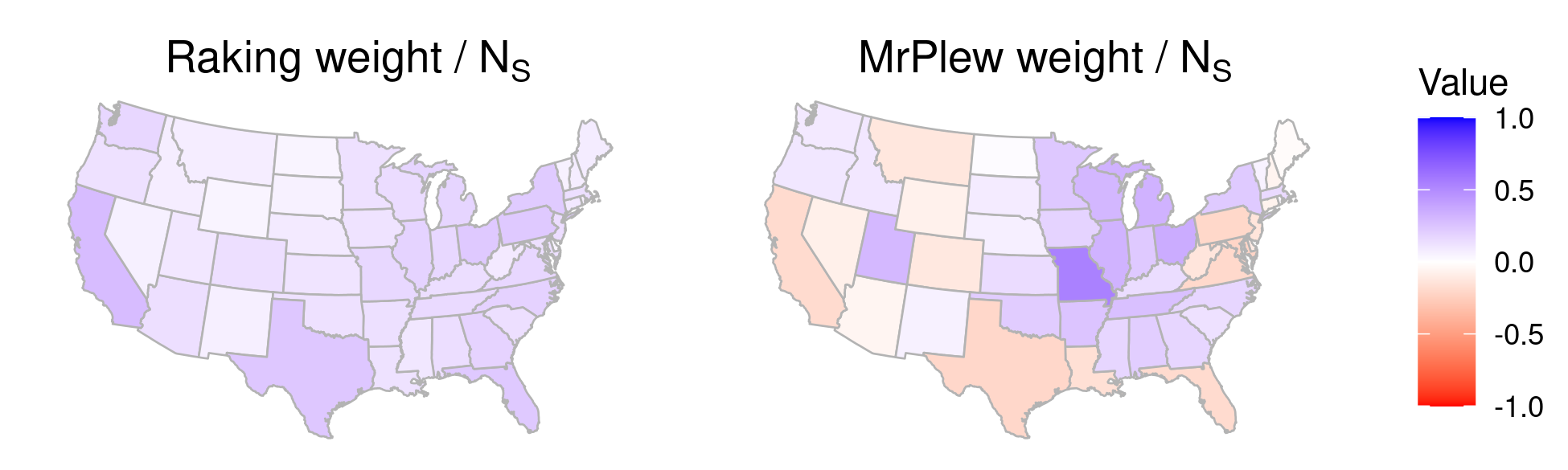}
\caption{Target: Missouri}
\label{fig:laxphilipspoolingplot_mo}
\end{subfigure}
\caption{Subgroup contribution for the Same-Sex Marriage application, with California and Missouri as targets.}
\label{fig:laxphilipspoolingplot}
\end{figure}

Next, we consider subgroup contributions for the Same-Sex Marriage application,
where we focus on states as the subgroups of interest. To demonstrate this
approach, we explore how state contributions change when the target geography
changes from California to Missouri. \Cref{fig:laxphilipspoolingplot_ca} shows
the relative contribution of each state for the raking and MrP estimates of
support for same-sex marriage in California. There are stark differences in the
variation of state weights between the two methods, with substantially more
variability for MrP than for raking.\footnote{ In \Cref{fig:poolingcompplot} of
\Cref{app:experiments}, we give additional evidence that the visual similarity
between the raking subgroup contribution plots is not a bug. Raking produces
very little state-to-state variation in weights between California and Missouri.
}

These differences partly reflect differences in the inputs for the two
estimators: MrP includes state as a predictor in the model, while raking only
includes non-geographic variables. They also reflect the fact that, as shown in
\Cref{fig:weightsplot}, locally equivalent weights for MrP can be negative,
unlike raking weights. As a result, many states actually have a \emph{negative}
contribution on the California-specific estimate, with the largest negative
weights for states in the Midwest. Finally, California itself receives much more
weight under MrP than under raking; this reflects the fact that California is a
large state with many survey observations and that MrP---but not
raking---includes state as a predictor.

\Cref{fig:laxphilipspoolingplot_mo} shows the same plot when we shift the target
from California to Missouri. We see a similar overall pattern, with much more
variability for MrPlew than for raking weights. However, the weights themselves
are now quite different, with states like California and Pennsylvania showing
negative weights for the MrP estimate for Missouri.




\subsection{Name Change example: Expanding the outcome model}
\label{sec:alexander_add_regs}

We fundamentally view MrPlew as an exploratory tool for model interrogation and
argue that the results in this section give useful context for assessing the use
of MrP in a given application. To illustrate one use of MrPlew in model
development, we return to the Name Change example from \Cref{sec:intro_example}.
Recall that \Cref{fig:introplot} showed substantial covariate imbalance in the
\AlexanderColpert{} interaction, which was not modeled in the original MrP
outcome model. 

Although balance is intuitively desireable, imbalance does not itself diagnose
model misspecification in general.  For example, \cref{exmp:logistic_balance}
demonstrates that correctly specified models can imbalance their own regressors,
and \cref{exmp:normal_weights} shows that calibration weights do not necessarily
even include $\Y$ at all and so cannot possibly diagnose problems with $\p(\y
\vert \x)$.  In the presence of imbalance, we recommend first considering
whether the imbalanced regressor plausibly aligns with the response, and then
imformally checking whether there is evidence in the fitted residuals
align with the imbalanced regressor.  In this case, \cref{tab:alexanderresidtable} shows
that, though the response does align with the imbalanced regressor,
the residuals of the model do not, suggesting that the imbalance
may not be affecting the MrP estimate.

\AlexanderResidualTable{}

\AlexanderModelsBalancePlot{}

Nevertheless, suppose we want to modify the model to reduce the imbalance. A
natural strategy to reduce imbalance is to expand the outcome model to include
the interaction term.\footnote{Including the regressor directly does not
generally yield exact balance: \Cref{exmp:logistic_balance} shows that logistic
regression balances $\v(\x)\r(\x)$, not $\r(\x)$. One could in principle include
$\vhat(\x)^{-1}\r(\x)$ for some plug-in $\vhat$, but this is suspect because
$\vhat$ depends on $\Y$ and so is not properly a measurable function of $\x$.}
Importantly, this interaction seems substantively plausible, so the choice to
include this term is not driven by MrPlew alone.
\Cref{fig:alexandermodelsbalance} shows the effect on covariate balance of
adding the imbalanced interaction to the model as a random effect (second panel)
and as a fixed effect (third panel). Both lead to substantial improvement in
covariate balance, even if the resulting point estimate is largely unchanged, as
predicted by the residual means in \cref{tab:alexanderresidtable}. Moreover,
after including the interaction term, MrPlew covariate balance checks do not
flag any additional imbalances of note.

\section{Discussion and open questions}

Our work raises a number of interesting questions which we hope will motivate
further research.

\subsection{Role of locally equivalent weights in model interrogation}

We view MrPlew balance checks as diagnostics for the common setting in which the
analyst trusts the model only approximately and wants to probe how it uses the
data. Balance checks are not, however, tests of model specification.
\Cref{exmp:logistic_balance,exmp:normal_balance} make this concrete. First, for
MrP with an OLS outcome model (\Cref{sec:linear_equivalent_weights}), the
\emph{globally} equivalent weights do not involve $\Y$ at all and so cannot
detect misspecification of $\p(\y \vert \x, \betav)$. Second, for MrP with a
logistic regression model (\Cref{exmp:logistic_balance}), a correctly specified
outcome model can still fail to achieve covariate balance.

Despite these limitations, we believe that the central idea of
\Cref{sec:balance} can be modified to produce actual specification checks:
design a perturbation of the data whose behavior is known under correct
specification, and then design local robustness tests of whether the model
exhibits the expected behavior.  For example, if the response is correctly
specified, then the MrP estimate should be approximately invariant to the
distribution $\psur(\x)$.  Pursuing this idea is the subject of ongoing work.


\subsection{Assessing non-local robustness}
\label{sec:non_local_robustness}

Local robustness is best understood as a computationally efficient approximation
to a non-local robustness question; see \textcite{giordano:2023:bnprejoinder}
and \citet[Appendix C]{giordano:2018:covariances}. In
\Cref{app:non_local_robustness}, we assess how well our local approximation
extrapolates in two applications, with mixed results: the imbalanced interaction
in the Name Change example extrapolates well, while the imbalanced education
level category in the Same-Sex Marriage example does not. The local results from
\Cref{thm:balance} remain valid, but understanding when and why the
approximation extrapolates is an important open question.


To assess non-local behavior, we first use MrPlew balance checks to identify
covariate directions that, in principle, can lead to large changes in the
estimate $\hat{\mu}^{\mrp}$. Using the procedure described in
\cref{sec:parametric_bootstrap}, we then repeatedly generate perturbed data sets
$\Ytiltil$ increasingly far from the original outcome and assess how well the
predictions from MrPlew track the realized changes in
$\hat{\mu}^{\mrp}(\Ytiltil)$. 
Suppose we have selected a perturbation $\delta \r(\x)$ and that we have
generated $\Ytiltil$ satisfying \Cref{eq:ytil_expectation}.  Then we expect that
\begin{align}
\underbrace{
    \meansur \w[\mrp]_i \ytiltil_i - \meansur \w[\mrp]_i \y_i
}_{\text{inferred change on binary vector}}
\approx
\underbrace{
    \delta \meansur \w[\mrp]_i \r_i
}_{{\substack{\text{inferred change on}\\\text{continuous perturbation}}}}.
\label{eq:ytil_expectation_imbalance}
\end{align}
\Cref{eq:ytil_expectation_imbalance} is a combination of \Cref{eq:cal_balance_change}
with the MrPlew weights, and \Cref{eq:ytil_expectation} with the law of large numbers as
$\nsur \rightarrow \infty$.
Next, suppose we have identified a $\Ytiltil$ satisfying \Cref{eq:ytil_expectation_imbalance},
and that our local approximation is good (e.g., see \Cref{eq:mrp_taylor}).  Then 
we expect that
\begin{align}
\underbrace{
    \muhat[\mrp](\Ytiltil) - \muhat[\mrp](\Y)
}_{\text{actual posterior change}}
\approx
\underbrace{
    \meansur \w[\mrp]_i \ytiltil_i - \meansur \w[\mrp]_i \y_i
}_{\text{inferred change on binary vector}}.
\label{eq:ytil_expectation_local_approximation}
\end{align}
Finally, if we have identified a perturbation that is imbalanced according to
$\W[\mrp]$, combining
\cref{eq:ytil_expectation_imbalance,eq:ytil_expectation_local_approximation} we
expect that
\begin{align}
\underbrace{
    \muhat[\mrp](\Ytiltil) - \muhat[\mrp](\Y)
}_{\text{actual posterior change}}
\quad \textrm{ is far from } \;
\underbrace{
    \meantar \pi_j \r_j
}_{\text{actual target change (\cref{eq:imbalance_pert})}}.
\label{eq:mrp_nonlocality}
\end{align}
That is, our local approximation \textit{should} be able to produce binary datasets
$\Ytiltil$ whose posterior change does not match the true target population
change.

In \Cref{app:non_local_robustness}, we ran experiments to check
\Cref{eq:mrp_nonlocality} for the Name Change and Same-Sex Marriage analyses.
For the Same-Sex Marriage example, we took $\r(\x)$ to be the imbalanced raking
marginal \LaxColpert{}; for the Name Change example, we took $\r(\x)$ to be the
imbalanced interaction \AlexanderColpert{}. As desired, the MrPlew prediction
for the generated binary $\Ytiltil$ indeed predicts a divergence with the truth:
the generated $\Ytiltil$ do in fact identify potentially problematic response
vectors. For the Name Change example, this approach shows that the local
approximation is quite good and extrapolates well. For the Same-Sex Marriage
example, however, the extrapolation is quite poor, suggesting we should be wary
of over-relying on the reported model checks.

It is known that linear approximations to posterior expectations may fail for
large changes, especially in posteriors with large numbers of poorly-estimated
random effects \parencite[Section 3]{giordano:2024:bayesij}. It is common in MrP
problems to have a large number of random effects, and so local approximations
should be taken with a grain of salt.  However, precisely why the nonlinearity
is so severe in the Same-Sex Marriage analysis but not in the Name Change
analysis is an open question; see \Cref{app:experiments} for additional
discussion.  Exploring posterior nonlinearity in more depth is an important
direction for future work.

\subsection{Creating asymptotically globally linear MrP estimators with DrP}
\label{sec:drp_globally_linear}

\newcommand{\drp}{\textrm{DrP}}
\newcommand{\rhohat}{\hat\rho}

As a complement to assessing curvature as discussed in
\cref{sec:non_local_robustness}, one might instead improve the usefulness of
locally linear diagnostics by modifying MrP estimators that are approximately
globally linear.  Here, we briefly argue that a simple technique for doing so is
provided by DrP \parencite{benmichael2024multilevel}.  An alternative approach
for creating globally linear MrP estimates is augmenting the MrP model with
regressors that predict $\ptar(\x) / \psur(\x)$, as discussed in
\cref{exmp:logistic_linearity} of \cref{app:logistic_linearity}. 

The DrP estimator is defined as follows.  Let $\rho(\x) := \ptar(\x) /
\psur(\x)$ denote the Radon-Nikodym derivative of the target regressor
distribution with respesct to the survey, and let $\rhohat(\cdot)$ denote an
estimate of $\rho(\cdot)$ that is independent of $\Y$ (e.g., formed with
covariates only or by sample splitting).  As usual, write $\rhohat_i =
\rho(\x_i)$.  For the present, take $\pi(\cdot) = 1$ for simplicity. Then the
DrP estimator based on $\muhat[\mrp](\Y)$ is
\parencite[Equation (12)]{benmichael2024multilevel}
$$
\muhat[\drp](\Y) = 
\muhat[\mrp](\Y) + \meansur \rhohat_i (\y_i - \mhat_i).
$$
Since $\muhat[\drp](\Y)$ is a function of $\Y$, we can form MrPlew weights for
the DrP estimate, and a simple argument sketched in \cref{app:drp_linearity}
shows that, if $\rhohat(\cdot)$ is a consistent estimator of $\rho(\cdot)$, then
\begin{align}
\w[\drp]_{i} = \rho(\x_i) + o_p(1),\label{eq:drp_linear}
\end{align}
where $o_p(1)$ is a quantity that vanishes as both $\nsur$ and $\ntar$ go to
infinity. Since $\rho(\x_i)$ does not depend on $\Y$, $\muhat[\drp](\Y)$ becomes
linear as $\nsur$ and $\ntar$ go to infinity, and $\w[\drp]_i$ converges to the
importance ratio $\rho(\x_i)$.  The authors are hopeful that modifications such
as this can improve the shortcomings of local robustness demonstrated in
\cref{sec:non_local_robustness}, though we leave rigorous theoretical and
experimental analysis of this approach for future work.



\section{Data and software availability}

Code to produce all the experiments of our paper
can be found at the git repo\\
\url{https://github.com/rgiordan/MrPLocallyEquivalentWeightsPaper}.  An
open-source software implementation of MrPlew is available
at \url{https://github.com/rgiordan/mrplew}.

\newpage
\printbibliography


\newpage
\begin{appendices}
\crefalias{section}{appendix}

\section{Details of closed-form examples}\label{app:simple_examples}

\subsection{Counterexample for \cref{thm:ij}}

\begin{exmp}\label{exmp:gaussian_glm}
Suppose that $\p(\y \vert \x, \betav, \sigma)$ is normal with mean
$\betav^\trans \x$ and variance $\sigma^2$.  Then the sufficient statistics are
$(\y, \y^2)$, and the model does not satisfy \cref{assu:linear}.  We will show
that, in general, the estimator $\mrpvarhat$ in \cref{eq:mrplew_variance} does not
provide consistent estimates of the frequentist variance.

In this case, the log likelihood is given (up to a
constant $C$ not depending on $\y_i$) by
$$
\begin{aligned}
\ell(\y_i \vert \x_i, \betav, \sigma) ={}&
-\frac{1}{2}\sigma^{-2} \left(
    \y_i^2  -
    2 \betav^\trans \x_i \y_i +
    \betav^\trans \x_i^\trans \x_i \betav
\right) + C \Rightarrow \\
\nabla_\y \ell(y_i | \x_i, \betav, \sigma) ={}&
    -\sigma^{-2} \left(\y_i  - 2 \betav^\trans \x_i \y_i\right)
\end{aligned}
$$
so $\w[\mrp]_i$ is given by (see \cref{eq:mrp_w})
$$
\w[\mrp]_i = \nsur
    \cov{\p(\betav, \sigma \vert \Y)}{
        -\sigma^{-2} \left(\y_i  - \betav^\trans \x_i \right),
        \g(\betav)}.
$$
In contrast, the empirical influence function of $\muhat[\mrp]$
\parencite{giordano:2024:bayesij} is given by
$$
\psi_i =
    \cov{\p(\betav, \sigma \vert \Y)}{
        \ell(\y_i \vert \x_i, \betav, \sigma),
        \g(\betav)},
$$
and $\nsur \psi_i \ne \w[\mrp]_i \varepsilon_i$ in general, even asymptotically.
For example, the coefficient in front of the $\y_i^2$ term is different
by a factor of two.
However, Theorem 2 of \parencite{giordano:2024:bayesij} shows that,
under mild regularity conditions,
$$
\meansur \left(N \psi_i -  (N \overline{\psi})\right)^2 \rightarrow \V.
$$
Since $\mrpvarhat$ of \cref{thm:ij} converges, in general, to a different limit
than that of the IJ estimator,
it cannot be consistent.

\end{exmp}


\subsection{Details for \cref{exmp:normal_weights,exmp:normal_weights}}
\label{app:simple_normal}

It will be convenient to define the following quantities:
$$
\xcovhat := \frac{1}{\nsur} \X^\trans \X
\quad\quad
\xscovhat := \xcovhat + \frac{\sigma^2}{\nsur} \betacov^{-1}
\quad\quad
\xycovhat := \frac{1}{\nsur} \X^\trans \Y.
$$

For \cref{eq:simple_normal_normal_w}, we have
\begin{align}
\post ={}& \gauss{\xscovhat^{-1} \xycovhat,
    \frac{\sigma^2}{\nsur} \xscovhat^{-1}
}
\nonumber \Rightarrow \\
\muhat[\mrp](\Y) ={}& \frac{1}{\ntar}\piv^\trans \Xtar \expect{\post}{\betav}
={}
\frac{1}{\ntar} \piv^\trans \Xtar \xscovhat^{-1} \xycovhat
\nonumber  \Rightarrow \\
\W[\mrp] ={}& \frac{1}{\ntar} \X \xscovhat^{-1} \Xtar^\trans \pi.
\end{align}

We can derive \cref{eq:simple_linear_imbalance}

$$
\begin{aligned}
\frac{1}{\delta} \imbalance(\x, \W[\mrp]) ={}&
\frac{1}{\ntar} \piv^\trans \Xtar - 
\frac{1}{\nsur} {\W[\mrp]}^\trans \X
\\={}&
\frac{1}{\ntar} \piv^\trans \Xtar - 
\frac{1}{\nsur \ntar} \piv^\trans \Xtar \xscovhat^{-1} \X^\trans \X
\\={}&
\frac{1}{\ntar} \piv^\trans \Xtar \left(\id_{P} - \xscovhat^{-1} \xcovhat \right)
\end{aligned}
$$

Finally, we prove that $\W[\mrp]$ minimizes the expected mean squared error
marginally over $\betav$.  First, marginally we have
$$
\cov{p(\Y \vert \X)}{\Y} = \sigma^2 \id_{\nsur} + \X \betacov \X^\trans
\quad\textrm{(marginally over $\betav$, correct specification)}.
$$
Then
$$
\begin{aligned}
\mathcal{E}(\W) 
:={}& 
\frac{1}{\ntar} \piv^\trans \Y_{\tar} - \frac{1}{\nsur} {\W}^\trans \Y
\\={}&
    \frac{1}{\ntar} \piv^\trans \Xtar \betav + \frac{1}{\ntar} \piv^\trans\resv_\tar - 
    \frac{1}{\nsur} {\W}^\trans \X \betav - \frac{1}{\nsur} {\W}^\trans \resv
\\={}&
    \left(
        \frac{1}{\ntar} \piv^\trans \Xtar - \frac{1}{\nsur} {\W}^\trans \X
    \right) \betav + \frac{1}{\ntar} \piv^\trans\resv_\tar - 
    \frac{1}{\nsur} {\W}^\trans \resv 
    \Rightarrow\\
\expect{\p(\Y, \Y_\tar \vert \X, \Xtar)}{\mathcal{E}(\W)^2} ={}&
\left(
        \frac{1}{\ntar} \piv^\trans \Xtar - \frac{1}{\nsur} {\W}^\trans \X
    \right) \betacov
    \left(
        \frac{1}{\ntar} \Xtar^\trans \piv - \frac{1}{\nsur} \X^\trans \W
    \right) + 
\\{}&
    \frac{1}{\nsur^2} {\W}^\trans \W \sigma^2 + \frac{1}{\ntar^2} \piv^\trans \piv \sigma^2
    \Rightarrow\\
\frac{\partial}{\partial \W} \expect{\p(\Y, \Y_\tar \vert \X, \Xtar)}{\mathcal{E}(\W)^2} ={}&
    \frac{2}{\nsur^2} \left(\sigma^2 \id_{\nsur} + \X \betacov \X^\trans \right) \W - 
    \frac{2}{\nsur \ntar} \X \betacov \Xtar^\trans \piv
    \Rightarrow\\
\W[*] ={}&
\frac{\nsur^2}{\nsur \ntar} 
    \left( \sigma^2 \id_{\nsur} + \X \betacov \X^\trans \right)^{-1}
    \X \betacov \Xtar^\trans \piv
\\={}&
\frac{\nsur}{\ntar}
    \X  
    \left( \sigma^2 \id_{P} + \betacov \X^\trans \X  \right)^{-1} \betacov
    \Xtar^\trans \piv
\\={}&
\frac{\nsur}{\ntar}
    \X  
    \left( \sigma^2 \betacov^{-1} + \X^\trans \X  \right)^{-1}
    \Xtar^\trans \piv
\\={}&
\frac{1}{\ntar} \X  \xscovhat^{-1} \Xtar^\trans \piv = \W[\mrp].
\end{aligned}
$$
In the preceding display, we used the push-through identity for
matrix inverses \parencite{henderson:1981:pushthrough}.


\subsection{Details for \cref{exmp:logistic_weights,exmp:logistic_balance}}
\label{app:simple_logistic}

Recall the setup in \cref{exmp:logistic_weights}.  We have
$$
\betahat = \argmax{\betav} \meansur \log \p(\y \vert \x, \betav),
$$
and so
$$
\begin{aligned}
\betagrad \meansur \log \p(\y_i \vert \x_i, \betahat) ={}& \meansur (\y_i - \yhat_i) \x_i
\quad\textrm{and}\\
\betagrad[2] \meansur \log \p(\y_i \vert \x_i, \betahat) ={}& -\meansur \vhat_i \x_i \x_i^\trans
= -\frac{1}{\nsur} \X^\trans \V \X.
\end{aligned}
$$

Note that
$$
\betagrad \g(\betavhat) = \meantar \pi_j \expect{\post}{\mlogit(\betav^\trans \x_j)}
$$
The Bernstein-von Mises theorem then states that, for large $\nsur$,
\begin{align}
\p\left( 
\sqrt{\nsur}(\betav - \betahat) 
\vert
\Y
\right)
\overset{\textrm{approx}}{\approx}  
\gauss{\zerov, \left(\frac{1}{\nsur} \X^\trans \V \X\right)^{-1}}.
\label{eq:logistic_bvm}
\end{align}

We can use the delta method and the approximation \cref{eq:logistic_bvm} to approximate
the MrPlew weights:
$$
\begin{aligned}
\w[\mrp]_i ={}& \nsur \cov{\post}{
    \g(\betav), \frac{\partial}{\partial \y_i} \log \p(\y_i \vert \x_i, \beta)}
\\={}&
    \nsur \cov{\post}{\g(\betav), \betav^\trans \x_i}
\\\approx{}&
\fracat{\partial \g(\betav)}{\partial \betav^\trans}{\betahat}
\left(\frac{1}{\nsur} \X^\trans \V \X\right)^{-1} \x_i.
\end{aligned}
$$
Next,
$$
\begin{aligned}
\fracat{\partial \g(\betav)}{\partial \betav}{\betahat} ={}
\meantar \pi_j \fracat{\partial \mlogit(\betav^\trans \x_j)}
                      {\partial \betav^\trans}{\betahat}
    ={}
\meantar \pi_j \vhat_j \x_j
= \frac{1}{\ntar} \Xtar^\trans \Vtar \piv.
\end{aligned}
$$
Combining gives \cref{eq:logistic_mrplew}.

Next, we consider covariate balance, beginning with the variance-weighted
function $\r(\x) = \v \x$, where we approximate $\v(\x_i)$ with $\vhat_i$ and
$\v(\x_j)$ with $\vhat_j$. 
$$
\begin{aligned}
\MoveEqLeft
\frac{1}{\delta} \imbalance(\v \x, \W[\mrp])
=\\{}&
\frac{1}{\ntar} \piv^\trans \Vtar \Xtar - 
\frac{1}{\nsur} {\W[\mrp]}^\trans \V \X
=\\{}&
\frac{1}{\ntar} \piv^\trans \Vtar \Xtar - 
\frac{1}{\ntar} \piv^\trans \Vtar \Xtar 
    \left(\frac{1}{\nsur} \X^\trans \V \X\right)^{-1} \frac{1}{\nsur} \X^\trans \V \X
=\\{}&
\frac{1}{\ntar} \piv^\trans \Vtar \Xtar - 
\frac{1}{\nsur} \piv^\trans \Vtar \Xtar = \zerov.
\end{aligned}
$$
In contrast, the imbalance for $\r(\x) = \x$ is generally not zero:
$$
\begin{aligned}
\MoveEqLeft
\frac{1}{\delta} \imbalance(\x, \W[\mrp])
=\\{}&
\frac{1}{\ntar} \piv^\trans \Xtar - 
\frac{1}{\nsur} {\W[\mrp]}^\trans \X
=\\{}&
\frac{1}{\ntar} \piv^\trans \Xtar - 
\frac{1}{\ntar} \piv^\trans \Vtar \Xtar 
    \left(\frac{1}{\nsur} \X^\trans \V \X\right)^{-1} \frac{1}{\nsur} \X^\trans \X
=\\{}&
\frac{1}{\ntar} \piv^\trans \Xtar
\left(\id_P - \left(\frac{1}{\nsur} \X^\trans \V \X\right)^{-1} \frac{1}{\nsur} \X^\trans \X \right)
\ne \zerov. \textrm{ (in general)}
\end{aligned}
$$


\subsection{Imbalance is not due to nonlinearity alone}
\label{app:logistic_linearity}
\begin{exmp}[Imbalance is not due to nonlinearity alone]
    \label{exmp:logistic_linearity}
One might wonder whether the failure of logistic regression to balance the
covariates in \cref{exmp:logistic_balance} is due to the nonlinearity of
the mapping $\Y \mapsto \muhat[\mrp](\Y)$. We show here that this is
not the case, and that in fact certain MrP estimators with nonlinear link
functions can be approximately linear in $\Y$ for large $\nsur$.

We consider again the asymptotic regime of logistic regression in
\cref{exmp:logistic_weights,exmp:logistic_balance}. Suppose we make the
following (improbable) assumption:
\begin{align}
    \textrm{Assume that there exists }\alpha\textrm{ such that }
\frac{\pi(\x) \ptar(\x)}{\psur(\x)} = \alphav^\trans \x.
\label{eq:density_ratio_lin}
\end{align}
\Cref{eq:density_ratio_lin} may be an unreasonable assumpton in general, but
with appropriate restrictions on the domain of $\x$, one could certainly
generate data according to \cref{eq:density_ratio_lin}, and so it is not
impossible that \cref{eq:density_ratio_lin} can be satisfied.  When
\cref{eq:density_ratio_lin} holds, then we show below
that, up to asymptotic approximations,
\begin{align}
\muhat[\mrp](\Y) \approx{}& \meansur \alphav^\trans \x_i \y_i,
\label{eq:logistic_linear}
\end{align}
where the approximation becomes exact as $\nsur \rightarrow \infty$.

\Cref{eq:logistic_linear} is perhaps surprising. Note that, even if one found
the assumption in \cref{eq:density_ratio_lin} plausible, it would be difficult
to directly estimate $\alphav$ because we do not directly observe $\ptar(\x)$
and $\psur(\x)$.  However, \cref{eq:logistic_linear} shows that $\w[\mrp]_i$ in
fact acts as a consistent estimator of $\alphav$.
\end{exmp}

\paragraph{Derivation.}
Recall that $\betahat$ is
the solution to the estimating equation
\begin{align}
\betagrad \meansur \log \p(\y_i \vert \x_i, \betahat) ={}
    \meansur (\y_i - \yhat_i) \x_i = \zerov.
    \label{eq:logistic_estimating_equation}
\end{align}
We will refer to the Radon-Nikodym derivative of $\pi(\cdot) \ptar(\cdot)$ with
respect to $\psur(\x)$ as the ``importance ratio,'' since it is the appropriate
weighting for a Horwitz-Thompson importance sampling estimator for converting a
sample from $\psur(\x)$ into a sample from $\ptar(\x)$.  In a causal inference
setting, the importance ratio is equivalent to the propensity score.

When \cref{eq:density_ratio_lin} holds, and when $\ntar$ and $\nsur$ are large
enough to invoke a law of large numbers, we can write
$$
\begin{aligned}
\muhat[\mrp](\Y) ={}&
    \meantar \pi_j \yhat_j & \textrm{(definition)}
\\\approx{}& \meantar \pi_j \x_j^\trans \betahat  & \textrm{(Bernstein von-Mises)}
\\\approx{}& \int \pi(\x) \x^\trans \ptar(\x) d\x \betahat & \textrm{(law of large numbers)}
\\={}& \alphav^\trans \left(
    \int \x \x^\trans \psur(\x) d\x
\right)  \betahat & \textrm{(\cref{eq:density_ratio_lin})}
\\\approx{}& \alphav^\trans \left(
    \meansur \x_i \x_i^\trans
    \right) \betahat  & \textrm{(law of large numbers)}
\\=& \alphav^\trans \meansur \x_i \yhat_i  & \textrm{(definition)}
\\=& \meansur \alphav^\trans \x_i \y_i.  & \textrm{\cref{eq:logistic_estimating_equation}}
\end{aligned}
$$
Up to the approximations in the preceding display, we thus have that $\Y \mapsto
\muhat[\mrp](\Y)$ is linear, and
\begin{align}
\w[\mrp]_i \approx \alphav^\trans \x_i = \frac{\pi(\x_i) \ptar(\x_i)}{\psur(\x_i)}.
\label{eq:implicit_weights}
\end{align}
Despite the linearity of $\Y \mapsto \muhat[\mrp](\Y)$, the conclusions of
\cref{exmp:logistic_balance} still hold: logistic regression, even
when linear, balances $\v(\x) \x$, not $\x$.


\subsection{Global linearity of DrP}
\label{app:drp_linearity}

We continue the discussion in \cref{sec:drp_globally_linear}
to support \cref{eq:drp_linear}.  First, we
can expand DrP as
$$
\muhat[\drp](\Y) = 
\muhat[\mrp](\Y) + \meansur \rhohat_i (\y_i - \mhat_i) = 
\meantar \mhat_j + \meansur \rhohat_i (\y_i - \mhat_i).
$$
The MrPlew weights for DrP are given by
$$
\w[\drp]_{i'} = 
    \frac{\nsur}{\ntar} \sumtar \frac{\partial \mhat_j}{\partial \y_{i'}} -
    \frac{\nsur}{\ntar} \sumsur \rhohat_i \frac{\partial \mhat_j}{\partial \y_{i'}} + 
    \rhohat_{i'},
$$
where we have used the fact that $\rhohat(\cdot)$ does not depend on $\y_{i'}$.

Applying a LLN to
both the survey and target populations gives
$$
\begin{aligned}
\w[\drp]_{i'} \approx{}&
    \int \nsur \frac{\partial \mhat(\x)}{\partial \y_{i'}} \ptar(\x) -
    \int \nsur  \frac{\partial \mhat(\x)}{\partial \y_{i'}} \rhohat(\x) \psur(\x) + \rhohat_{i'}
\\={}&
    \int \nsur \frac{\partial \mhat(\x)}{\partial \y_{i'}} (\rho(\x) - \rhohat(\x)) \psur(\x) 
    + \rhohat_{i'}.
\end{aligned}
$$
Under the assumption that $\var{\psur(\x)}{\nsur \frac{\partial
\mhat(\x)}{\partial \y_{i'}}} = O_p(1)$ (see \cref{lem:cov_ij}) and
$\rhohat(\cdot)$ is consistent in the sense that $\var{\psur(\x)}{\rho(\x) -
\rhohat(\x)} = o_p(1)$, Cauchy-Schwartz gives  
\begin{align}
\w[\drp]_{i'} = \rhohat_{i'} + o_p(1).
\end{align}
An argument analogous to \cref{thm:balance} can make this argument uniform in
$\i'$, and justify the replacement of $\rhohat_{i'}$ with $\rho(\x_i)$, though a
careful rigorous treatment is beyond the scope of the current paper.

\section{Proof of \cref{thm:ij}}\label{app:ij}

\def\assuij{\cref{assu:linear,assu:mle,assu:ij_lite}}

We first note that the boundedness of
$\x$ in \cref{assu:ij_lite} implies \cref{assu:ij}, which is more
technical but closer to what we actually need in the proof of
\cref{lem:ij}.

\begin{assu}\label{assu:ij}
    Let \cref{assu:ij_lite} hold (recall that this includes
    \cref{assu:linear,assu:mle}), but in place of 
    the boundedness of $\x$, assume that
\begin{itemize}
    \item $\expect{\psur(\x)}{\norm{\x}_2^2} < \infty$
    \item $\expect{\psur(\x)}{\y^2 \norm{\x}_2^2} < \infty$
    \item $\expect{\psur(\x)}{
        \norm{\Agrad{k}({\betastar}^\trans \x) \x^{\otimes k}}_2^2} < \infty$
        for $k \in \{0, \ldots, 4\}$
    \item $\expect{\psur(\x)}{
        \sup_{\betav \in \betaball}
        \norm{\Agrad{4}(\betav^\trans \x) \x^{\otimes 4}}_2^2} < \infty$ for
        some $\Delta > 0$.
\end{itemize}
In place of the assumption that $\p(\betav)$ has bounded support, assume that
\begin{itemize}
    \item $\expect{\p(\betav)}{\g(\betav)^2} < \infty$.
    \item $\expect{\psur(\x)}{
                \expect{\p(\betav)}{\ell(\y \vert \x, \betav)^2}
            } < \infty$.
    \item $\expect{\psur(\x)}{
                \expect{\p(\betav)}{
                    \g(\betav)^2 \ell(\y \vert \x, \betav)^2
                    }} < \infty$.
\end{itemize}
\end{assu}

For the remainder of the proof, let $\betaball := \{ \betav: \norm{\betav -
\betastar}_2^2 \le \Delta \}$ to be the $\Delta$--ball around $\betastar$.  We
first show a technical regularity lemma.

\begin{lem}\label{lem:ij} Under \assuij{}, there exists a $\Delta > 0$
    such that, for all
    $\betav \in \betaball$, $\ell(\betav)$
    is four times continuously differentiable and the exchange of partial
    differentiation with respect to $\betav$ and integration with respect to
    $\psur(\y, \x)$ is justified.

    Additionally, for $0 \le k \le 4$, there exists functions $M_k(\x, \y)$ such that
    \begin{align}
    \sup_{\betav \in \betaball} \norm{\betagrad[k] \ell(\betav)}_2^2 \le M_k(\x, \y)
    \quad\textrm{and}\quad\expect{\psur(\x, \y)}{M_k(\x, \y)} < \infty.
    \label{eq:ulln}
    \end{align}
\end{lem}

\begin{proof}

We first show that
$$
\begin{aligned}
\expect{\psur(\x)}{
    \sup_{\betav \in \betaball} \norm{\Agrad{k}(\betav^\trans \x) \x^{\otimes k}}_2^2
} \le \infty
\quad\textrm{ for }0 \le k \le 4.\label{eq:gradbound}
\end{aligned}
$$

Note that $\Agrad{k}(\betav)$ exists for all $k$ by standard properties of
exponential families.  Let $\betav_t := \betastar + t (\betav - \betastar)$ for
$t \in [0,1]$.  Then for any $k \le 3$,
$$
\begin{aligned}
\MoveEqLeft
\expect{\psur(\x)}{
    \sup_{\betav \in \betaball}
        \norm{\Agrad{k}(\betav^\trans \x) \x^{\otimes k}}_2^2
}
\\&={}
\expect{\psur(\x)}{
    \sup_{\betav \in \betaball} \norm{
        \left(
        \int_{0}^1
        \frac{\partial}{\partial t}
        \Agrad{k}(\betav_t^\trans \x) d\,t
        + \Agrad{k}({\betastar}^\trans \x)
        \right)
    \x^{\otimes k}
        }_2^2
    }
\\&={}
\expect{\psur(\x)}{
    \sup_{\betav \in \betaball} \norm{
        \left(
        \int_{0}^1
        \Agrad{k + 1}(\betav_t^\trans \x) d\,t
        (\betav - \betastar)^\trans \x
        + \Agrad{k}({\betastar}^\trans \x)
        \right)
    \x^{\otimes k}
        }_2^2
    }
\\&\le{}
\expect{\psur(\x)}{
    \sup_{\betav \in \betaball} \norm{
        \left(
        \int_{0}^1
        \Agrad{k + 1}(\betav_t^\trans \x) d\,t
        (\betav - \betastar)^\trans \x
     \right) \x^{\otimes k}
        }_2^2
    } +
    \expect{\psur(\x)}{
        \norm{
        \Agrad{k}({\betastar}^\trans \x)
    \x^{\otimes k}
    }_2^2 }
\\&\le{}
\expect{\psur(\x)}{
    \sup_{\betav \in \betaball} \norm{
        \left(
        \int_{0}^1
        \Agrad{k + 1}(\betav_t^\trans \x) d\,t
     \right) \x^{\otimes (k + 1)}
        }_2^2
    } \Delta +
    \expect{\psur(\x)}{
        \norm{
        \Agrad{k}({\betastar}^\trans \x)
    \x^{\otimes k}
    }_2^2 }
\\&\le{}
\expect{\psur(\x)}{
    \norm{
    \sup_{\betav \in \betaball}
        \Agrad{k + 1}(\betav^\trans \x)
     \x^{\otimes (k + 1)}
        }_2^2
    } \Delta +
    \expect{\psur(\x)}{
        \norm{
        \Agrad{k}({\betastar}^\trans \x)
    \x^{\otimes k}
    }_2^2 }
\end{aligned}
$$
The result follows by backward induction on $k$ starting
at $k=4$ using \cref{assu:ij}.

We now turn to bounding the gradients of $\ell(\betav)$. Since
$\betagrad[k] \ell(\betav) = \Agrad{k}({\betastar}^\trans \x) \x^{\otimes k}$
for $k \ge 2$, we have already uniformly bounded the second derivatives and higher.
For the first derivative, we can expand
$$
\begin{aligned}
\MoveEqLeft
\expect{\psur(\x, \y)}{
    \sup_{\betav \in \betaball}
    \norm{\nabla_{\betav} \ell(\y | \x, \betav) }_2^2
}
=\\{}&
\expect{\psur(\x, \y)}{
    \sup_{\betav \in \betaball}
    \norm{\y \x + \Agrad{1}({\betav}^\trans \x) \x }_2^2
}
\le\\{}&
\expect{\psur(\x, \y)}{\y^2 \norm{\x}_2^2} +
\sup_{\betav \in \betaball}
    \expect{\psur(\x)}{
    \norm{
    \Agrad{1}({\betav}^\trans \x) \x
    }_2^2} +
2 \sup_{\betav \in \betaball}
\expect{\psur(\x, \y)}{
    \y \Agrad{1}({\betav}^\trans \x) \x
}
\le\\{}&
\expect{\psur(\x, \y)}{\y^2 \norm{\x}_2^2} +
\sup_{\betav \in \betaball}
    \expect{\psur(\x)}{
    \norm{
    \Agrad{1}({\betav}^\trans \x) \x
    }_2^2} +
\\{}&\quad 2 \left(
    \sup_{\betav \in \betaball}
    \expect{\psur(\x)}{
        \norm{\Agrad{1}({\betav}^\trans \x) \x}_2^2
    }
    \expect{\psur(\y)}{\y^2}
\right)^{1/2}
 <{} \infty,
\end{aligned}
$$
by Cauchy-Schwartz, \cref{eq:gradbound}, and \cref{assu:ij}.
Similarly,
\begin{align}
\ell(\y | \x, \betav)^2 ={}&
\left(\y \x^\trans \betav - \A(\x^\trans \betav) \right)^2
\nonumber \\={}&
\y^2 \trace{\x \x^\trans \betav \betav^\trans} +
\A(\x^\trans \betav)^2 + 2 \A(\x^\trans \betav) \y \x^\trans \betav,
\end{align}
and so
\begin{align*}
\MoveEqLeft
\expect{\psur(\x, \y)}{
    \sup_{\betav \in \betaball}
    \ell(\y | \x, \betav)^2
}
\\={}&
\trace{
    \expect{\psur(\x, \y)}{ \y^2 \x \x^\trans}
    \sup_{\betav \in \betaball} \betav \betav^\trans
    } +
    \expect{\psur(\x, \y)}{
        \sup_{\betav \in \betaball} \A(\x^\trans \betav)^2
     }
     +\\&
    2
    \expect{\psur(\x, \y)}{
        \sup_{\betav \in \betaball}
        \A(\x^\trans \betav) \y \x^\trans \betav}
\\\le{}&
 \expect{\psur(\x, \y)}{ \y^2 \norm{\x}_2^2}
    \sup_{\betav \in \betaball} \norm{\betav}_2^2 +
    \expect{\psur(\x, \y)}{
        \sup_{\betav \in \betaball} \A(\x^\trans \betav)^2
     } +
     \\&
    2
    \left(
    \expect{\psur(\x)}{
        \A(\x^\trans \betav)^2}
    \expect{\psur(\x, \y)}{\y^2 \norm{\x}_2^2}
        \right)^{1/2}
    \sup_{\betav \in \betaball} \norm{\betav}_2
    < \infty,
\end{align*}
again by Cauchy-Schwartz, \cref{eq:gradbound}, and \cref{assu:ij}.

The exchange of differentiation and integration now follows from
Cauchy-Schwartz and the dominated convergence theorem, and for
each $0 \le k \le 4$ we can bound
$$
\norm{\betagrad[k] \ell(\betav)}_2^2 \le
\sup_{\betav\in\betaball} \norm{\betagrad[k] \ell(\betav)}_2^2 =: M_k(\x,\y),
$$
where we have shown that $\expect{\psur(\x, \y)}{M_k(\x,\y)} < \infty$.

\end{proof}


\begin{lem}\label{lem:loglik_bounded}
    Let \cref{assu:ij} hold.
    With probability approaching one, there exists a $\gamma > 0$ such
    that
    $$
    \sup_{\betav \in \betadom \setminus \betaball} \left(
        \meansur \ell(\y_i \vert \betastar, \x_i) -
        \meansur \ell(\y_i \vert \betav, \x_i)
    \right) \ge \gamma > 0.
    $$
    That is, the empirical log likelihood is strictly bounded away from the
    value achieved at $\betastar$ outside of $\betaball$.
\end{lem}

\begin{proof}

\def\betabound{\partial \betaball}
\def\ellhat{\hat{\ell}_{\nsur}}

Let $\betabound$ denote the boundary of the set $\betaball$. By \cref{assu:ij},
$\ell(\betastar) > \ell(\betav)$ for all $\betav \in \betabound$ because
$\betastar$ is a strict maximum.
Define $\ellhat(\betav) := \meansur \ell(\y_i \vert \betastar, \x_i)$. By
\cref{lem:ij}, $\ell(\y_i \vert \betav, \x_i)$ obeys a uniform law of large
numbers (ULLN) in $\betaball$, so that for sufficiently large $N$, with
probability approaching one we have, for some $\gamma < 0$,
$$
\ellhat(\betastar)  -
\sup_{\betav \in \betabound} \ellhat(\betav) > \gamma > 0.
$$

Recall that
$$
\betagrad[2] \ellhat(\betav)
= -\meansur \Agrad{2}(\x_i^\trans \betav) \x_i^{\otimes 2}.
$$
So, on $\betadom$, $\meansur \betagrad[2] \ellhat(\betav)$ is negative
semi-definite, since $\Agrad{2}(\betav^\trans \x_i) \succeq 0$ by standard properties
of the given exponential family.

For any point $\betav'' \in \betadom \setminus \betaball$, the line connecting
$\betastar$ to $\betav''$ must pass through some $\betav' \in \betabound$.
Since $\betagrad[2] \ellhat(\betav)$ is negative semidefinite, the slope
along any line cannot decrease, and so we must have that
$$
\ellhat(\betastar) \ge
\ellhat(\betav') + \gamma \ge
\ellhat(\betav'')+ \gamma,
$$
from which the result follows.

\end{proof}


\begin{lem}
Under \cref{assu:ij}, \cref{lem:loglik_bounded,lem:ij} imply that our problem
satisfies Assumptions 1 and 2 of \textcite{giordano:2024:bayesij} for the
quantity of interest $\g(\betav) := \meantar \pi_j \m(\x_j^\trans \betav)$.

Additionally, the quantities $\A(\betav^\trans \x_i)$ and
$\Agrad{1}(\betav^\trans \x_i)$ are third-order BCLT-okay as given by Definition
2 of \textcite{giordano:2024:bayesij}

\end{lem}
\begin{proof}
Each part of Assumption 1 of \textcite{giordano:2024:bayesij} is already
satisfied directly in \cref{assu:ij} or by \cref{lem:loglik_bounded,lem:ij}.
Note that we take $\Omega_\theta$ to be the set $\{ \theta: \p(\theta) > 0\}$.

For the remainder of the lemma, we must show that $\g(\betav)$, $\ell(\x \vert
\betav)$, $\g(\betav) \ell(\x \vert \betav)$, $\A(\betav^\trans \x_i)$, and
$\Agrad{1}(\betav^\trans \x_i)$ are BCLT-okay, as given by the three items in
Definition 2 of \textcite{giordano:2024:bayesij}.

For item 1 (almost sure differentiability), the assumption
follows from properties of exponential families.

For item 2 (order one average derivatives of prior expectations),
we have assumed in \cref{assu:ij} that
$\meansur \expect{\p(\betav)}{\g(\betav)^2} < \infty$,
and that we can apply the law of large numbers to the
quantities
$$
\begin{aligned}
\meansur \expect{\p(\betav)}{\g(\betav)^2 \ell(\y_i \vert \x_i, \betav)^2}
\quad\textrm{and}\quad
\meansur \expect{\p(\betav)}{\ell(\y_i \vert \x_i, \betav)^2}.
\end{aligned}
$$
Item 2 is satisfied for $\A(\betav^\trans \x_i)$, and
$\Agrad{1}(\betav^\trans \x_i)$ by the ULLN implied by
\cref{lem:ij}.

For item 3 (sample averages bounded in probability in $\betaball$), the
assumption follows from the ULLNs of \cref{assu:ij} and \cref{lem:ij}, together
with $\sup_{\betav \in \betaball} \g(\betav) < \infty$ by continuity.

\end{proof}

For the duration of this section, let $\betavhat$ denote the MAP
$\argmax{\betav} \post$. Note that earlier in the paper $\betavhat$ denoted the
OLS coefficient, but we will overload that notation for the moment.


\def\bcltresid{\mathcal{E}}
\def\ord#1{\mathcal{\tilde{O}}\left( #1 \right)}
\def\ordp#1{\mathcal{\tilde{O}}_p\left( #1 \right)}

{ 
\def\a{a_i}
\def\b{b}
\def\abar{\overline{a}_i}
\def\bbar{\overline{b}}
\def\abbar{\overline{ab}_i}

\begin{lem}\label{lem:cov_ij}

Let Assumption 1 of \textcite{giordano:2024:bayesij} hold, and suppose that
$\a(\betav)$, $\b(\betav)$, and $\a(\betav) \b(\betav)$ satisfy Assumption 2 of
\textcite{giordano:2024:bayesij}.  Here, $\a(\betav)$ may
depend on datapoint $i$ but $\b(\betav)$ does not. Then
$$
\nsur \cov{\post}{\a(\betav), \b(\betav)} ={}
\frac{1}{2} \betagrad \a(\betav) \infohat^{-1} \betagrad \b(\betav) +
\tilde{O}(\nsur^{-1}) \bcltresid^{cov}_i,
$$
where $\meansur \left( \bcltresid^{cov}_i \right)^2 = \ordp{1}$.
\end{lem}
\begin{proof}
The proof simply states a general version of an argument from
the proof of \textcite{giordano:2024:bayesij} Theorem 2.

Define
$$
\begin{aligned}
\Delta^a_i :={}& \a(\betavhat) - \expect{\post}{\a(\betav)} \\
\Delta^b :={}& \b(\betavhat) - \expect{\post}{\b(\betav)} \\
\abar(\betav) :={}& \a(\betav) - \a(\betavhat) \\
\bbar(\betav) :={}& \b(\betav) - \b(\betavhat) \\
\abbar(\betav) :={}& \abar(\betav) \bbar(\betav).
\end{aligned}
$$
We can then rewrite the covariance as
$$
\begin{aligned}
\cov{\post}{\a(\betav), \b(\betav)} ={}&
\expect{\post}{(\abar(\betav)  + \Delta^a_i )(\bbar(\betav) + \Delta^b)}
\\={}&
\expect{\post}{\abbar(\betav)} + \Delta^a_i \Delta^b +
\\{}&
    \expect{\post}{\abar(\betav)} \Delta^b +
    \Delta^a \expect{\post}{\bbar(\betav)}
\\={}& \expect{\post}{\abbar(\betav)} - \Delta^a_i \Delta^b.
\end{aligned}
$$

By Theorem 1 of \textcite{giordano:2024:bayesij},
$$
\begin{aligned}
\Delta^a_i ={} \ord{\nsur^{-1}} \bcltresid^{\a}_i
\quad\textrm{and}\quad
\Delta^b ={} \ord{\nsur^{-1}} \bcltresid^{\b},
\end{aligned}
$$
where the residuals $\bcltresid^{\a}_i$ and $\bcltresid^{\b}$ combine
leading-order and residuals terms, and satisfy $\meansur
\left( \bcltresid^{a}_i \right)^2 = \ordp{1}$
by the BCLT okay assumption and
\textcite{giordano:2024:bayesij} Theorem 1.

Then, noting that $\betagrad \abbar(\betavhat) = \zerov$, another
application of \textcite{giordano:2024:bayesij} Theorem 1 also gives that
$$
\begin{aligned}
\nsur \expect{\post}{\abbar(\betav)} ={}&
\frac{1}{2} \betagrad \a(\betavhat)^\trans \infohat^{-1} \betagrad \b(\betavhat) +
\ord{\nsur^{-1}} \bcltresid^{ab}_i,
\end{aligned}
$$
where again $\bcltresid^{ab}_i$ is finitely square-summable with probability
approaching one.
Combining gives
$$
\begin{aligned}
\nsur \cov{\post}{\a(\betav), \b(\betav)} ={}&
\frac{1}{2} \betagrad \a(\betavhat)^\trans \infohat^{-1} \betagrad \b(\betavhat) +
\ord{\nsur^{-1}} \left( \bcltresid^{ab}_i - \bcltresid^{a}_i \bcltresid^{b} \right).
\end{aligned}
$$
Furthermore, setting $\bcltresid^{cov}_i := \bcltresid^{ab}_i - \bcltresid^{a}_i \bcltresid^{b}$,
we have $\meansur \left( \bcltresid^{cov}_i \right)^2 = \ordp{1}$ by
Cauchy-Schwartz.
\end{proof}
}


\paragraph{Proof of \cref{thm:ij}.}

\def\ggrad{\nabla_{\betav} \g}

Let $\psi_i$ denote the classical empirical
influence function for $\muhat[\mrp]$ as defined in
\textcite{giordano:2024:bayesij},
$$
\psi_i := \nsur \cov{\post}{\g(\betav), \ell(\y_i \vert \x_i, \betav)}.
$$
Recall also from the statement of \cref{thm:ij} that
$$
\nsur \w[\mrp]_i \varepsilon_i =
\nsur \cov{\post}{\g(\betav), \x_i^\trans \betav} (\y_i - \yhat_i),
$$
where
$$
\yhat_i = \expect{\post}{\m(\betav^\trans \x_i)}.
$$

The
proof will show that $\psi_i \approx \nsur \w[\mrp]_i \varepsilon_i$
to a sufficiently high degree of accuracy,
and then to apply \textcite{giordano:2024:bayesij} Theorem 2.

Define $\ggrad(\betav)$ the derviative of $\g(\betav)$.  By
Theorem 1 of \textcite{giordano:2024:bayesij}
$$
\begin{aligned}
\nsur \w[\mrp]_i \varepsilon_i - \psi_i ={}&
    \nsur \cov{\post}{\g(\betav),
        A(\betav^\trans \x_i) - \yhat_i \betav^\trans \x_i}
\end{aligned}
$$
By \cref{lem:cov_ij}, we can write
$$
\begin{aligned}
\MoveEqLeft
\nsur \cov{\post}{\g(\betav),
    A(\betav^\trans \x_i) - \yhat_i \betav^\trans \x_i} =
\\&
\ggrad(\betavhat)^\trans \left(
    \Agrad{1}(\betavhat^\trans \x_i) \x_i - \yhat_i \x_i
    \right) +
    \ordp{N^{-1}} \bcltresid^{cov}_i.
\end{aligned}
$$
Additionally, by Theorem 1 of \textcite{giordano:2024:bayesij},
$$
\begin{aligned}
\yhat_i ={}& \expect{\post}{\m(\betav^\trans \x_j)}
\\={}&
    \m(\betavhat^\trans \x_j) +
    \ordp{N^{-1}} \bcltresid^{y}_i
\\={}&
    \Agrad{1}(\betavhat^\trans \x_j) +
    \ordp{N^{-1}} \bcltresid^{y}_i,
\end{aligned}
$$
where the final line follows from the exponential family fact
that $\Agrad{1}(\cdot) = \m(\cdot)$.  It follows that
$$
\Agrad{1}(\betavhat^\trans \x_i) \x_i - \yhat_i \x_i =
\ordp{N^{-1}} \bcltresid^{y}_i \x_i
$$
Combining,
$$
\begin{aligned}
\nsur \w[\mrp]_i \varepsilon_i - \psi_i ={}&
    \ord{N^{-1}} \left(
        \bcltresid^{cov}_i +
        \bcltresid^{y}_i \ggrad(\betavhat)^\trans \x_i \right)
\\=:{}&
\ord{N^{-1}} \bcltresid^{err}_i.
\end{aligned}
$$

From this it follows that

$$
\begin{aligned}
\overline{\w[\mrp] \varepsilon} ={}&
\meansur \nsur \w[\mrp]_i \varepsilon_i \\={}&
\overline{\psi} + \ord{N^{-1}} \meansur \bcltresid^{err}_i \\
\meansur \left( \nsur \w[\mrp]_i \varepsilon_i \right)^2 ={}&
    \meansur \psi_i^2 +
\\{}&
    2 \ord{N^{-1}} \meansur \bcltresid^{err}_i \psi_i +
    \ord{N^{-2}} \meansur \left(\bcltresid^{err}_i\right)^2
\end{aligned}
$$
and so
$$
\meansur \left(
    \nsur \w[\mrp]_i \varepsilon_i  - \nsur \overline{\w[\mrp] \varepsilon}
    \right)^2 =
\meansur \left( \psi_i - \overline{\psi}\right)^2 + \ordp{N^{-1}}.
$$
The final result follows from Theorem 2 of \textcite{giordano:2024:bayesij},
which states that the right hand side of the preceding
display is a consistent estimate of $\V$. \hfill\qedsymbol

\section{Proof of \cref{thm:balance}}\label{app:balance}

For any $\r$, take $\Y(\delta \r) = \Y + \delta \R$.

$$
\begin{aligned}
\muhat[\mrp](\Ytil) - \muhat[\mrp](\Y) =
\fracat{\partial \muhat[\mrp](\Y(\delta))}{\partial \delta}{\delta = 0} (\delta - 0)+
\frac{1}{2}
    \fracat{\partial^2 \muhat[\mrp](\Y(\delta \r))}{\partial \delta^2}{\delta = \deltatil}
    (\delta - 0)^2
\end{aligned}
$$
for some $\deltatil \in [0,\delta]$.  The first term is given by
$$
\fracat{\partial \muhat[\mrp](\Y(\delta \r))}{\partial \delta}{\delta = 0} =
\sumsur \frac{\partial \muhat[\mrp](\Y)}{\partial \y_i} \r_i =
\delta \meansur \w[\mrp]_i \r_i.
$$
The second term is given by
$$
\fracat{\partial^2 \muhat[\mrp](\Y(\delta \r))}{\partial \delta^2}{\delta = \deltatil} =
\fracat{\partial^2 \expect{\postd}{\g(\betav)}}{\partial \delta^2}{\delta = \deltatil}.
$$
So we would like to control the error
\begin{align}
\MoveEqLeft
\sup_{\r \in \rset}
\frac{1}{\delta}
\abs{
    \muhat[\mrp](\Ytil) - \muhat[\mrp](\Y) -
    \delta \meansur \w[\mrp]_i \r_i
} 
\nonumber\le\\&
\delta
\sup_{\r \in \rset}
\sup_{\deltatil \in [0,\delta]}
\frac{1}{2}
\abs{
    \fracat{\partial^2 \expect{\postd}{\g(\betav)}}{\partial \delta^2}{\delta = \deltatil}
}. \label{eq:balance_error}
\end{align}

Let $\cumulant{\p(\betav)}{\cdot}$ denotes the third--order cumulant of $\p(\betav)$,
$$
\begin{aligned}
\MoveEqLeft
\cumulant{\p(\betav)}{a(\betav), b(\betav), c(\betav)} :=
\\&
\expect{\p(\betav)}{
    (a(\betav) - \expect{\p(\betav)}{a(\betav)})
    (b(\betav) - \expect{\p(\betav)}{b(\betav)})
    (c(\betav) - \expect{\p(\betav)}{c(\betav)})
}.
\end{aligned}
$$
Then, for the right hand side of \cref{eq:balance_error}, note that
$$
\begin{aligned}
\fracat{\partial^2 \expect{\postd}{\g(\betav)}}{\partial \delta^2}{\delta} ={}&
\cumulant{\postd}{
    \g(\betav),
    \fracat{\partial \ell(\Y \vert \X, \betav; \delta \r)}{\partial \delta}{\delta},
    \fracat{\partial \ell(\Y \vert \X, \betav; \delta \r)}{\partial \delta}{\delta}
}
\\={}&
\cumulant{\postd}{
    \g(\betav),
    \sumsur \r_i \x_i^\trans \betav,
    \sumsur \r_i \x_i^\trans \betav
}
\\={}&
\sumsur \sum_{i' \in [\nsur]}
\r_{i'} \x_{i'}^\trans
\cumulant{\postd}{
    \g(\betav),
    \betav,
    \betav
} \x_i \r_i .
\end{aligned}
$$
Plugging in,
$$
\begin{aligned}
\MoveEqLeft
\sup_{\r \in \rset}
\sup_{\deltatil \in [0,\delta]}
\abs{
    \fracat{\partial^2 \expect{\postd}{\g(\betav)}}{\partial \delta^2}{\delta = \deltatil}
}
={}\\&
\sup_{\r \in \rset}
\sup_{\deltatil \in [0,\delta]}
    \abs{
        \sumsur \sum_{i' \in [\nsur]}
        \r_{i'} \x_{i'}^\trans
        \cumulant{\postdtil}{
            \g(\betav),
            \betav,
            \betav
        } \x_i \r_i
    }
\le{}\\&
\sup_{\r \in \rset}
\sup_{\deltatil \in [0,\delta]}
\nsur^2 \norm{\cumulant{\postdtil}{
            \g(\betav),
            \betav,
            \betav
        }
    }_2
\norm{
    \frac{1}{\nsur^2} \sumsur \x_i^\trans \r_i \sum_{i' \in [\nsur]} \x_{i'} \r_{i'}
}_2
={}\\&
\sup_{\r \in \rset}
\sup_{\deltatil \in [0,\delta]}
\nsur^2 \norm{\cumulant{\postdtil}{
            \g(\betav),
            \betav,
            \betav
        }
    }_2
\norm{
    \frac{1}{\nsur} \sumsur \x_i \r_i
}_2^2   
\le{}\\&
\sup_{\r \in \rset}
\sup_{\deltatil \in [0,\delta]}
\nsur^2 \norm{\cumulant{\postdtil}{
            \g(\betav),
            \betav,
            \betav
        }
    }_2
\frac{1}{\nsur} \sumsur \norm{\x_i \r_i}_2^2   
\le{}\\&
\sup_{\r \in \rset}
\sup_{\deltatil \in [0,\delta]}
\nsur^2 \norm{\cumulant{\postdtil}{
            \g(\betav),
            \betav,
            \betav
        }
    }_2
    \sup_{\r \in \rset}
        \frac{1}{\nsur} \sumsur \norm{\x_i \r_i}_2^2.
\end{aligned}
$$
Here, since $\rset$ is Donkser and satisfies
$\sup_{\r \in \rset} \expect{\psur(\x)}{\x \r(\x)}_2^2 < \infty$ by \cref{assu:balance},
by a uniform law of large numbers we have
$$
\sup_{\r \in \rset}
    \frac{1}{\nsur} \sumsur \norm{\x_i \r_i}_2^2 \rightarrow 
    \expect{\psur(\x)}{\norm{\x_i \r_i}_2^2} \le \rmax^2.
$$
Therefore it will suffice to show that
$$
\sup_{\r \in \rset}
\sup_{\deltatil \in [0,\delta]}
\norm{
\cumulant{\postdtil}{
    \g(\betav),
    \betav,
    \betav
}}_2
= \ordp{\nsur^{-2}}.
$$
To show this, we need to establish a version of Theorem 1 of \textcite{giordano:2024:bayesij}
that holds uniformly over $\delta$ and $\r$.









\subsection{Uniform posterior expansions}

\def\t{t}
\def\tdom{\Omega_\t}
\def\postt{\p(\betav \vert \Y, \t)}

{ 

\def\a{a}
\def\b{b}
\def\c{c}
\def\da{\Delta^\a}
\def\db{\Delta^\b}
\def\dc{\Delta^\c}
\def\abar{\overline{a}\,}
\def\bbar{\overline{b}\,}
\def\cbar{\overline{c}\,}

\begin{assu}\label{assu:bayesij_uniform}
Define a parameterized log likelihood function $\ell(\betav \vert \y, \t)$, for
parameter $\t \in \tdom$.  Let Assumption 1 of \textcite{giordano:2024:bayesij}  
hold for each $\t \in \tdom$ with constants depending on $\t$.  Additionally assume
that items 4, 5, and 6 of Assumption 1 hold uniformly in $\t$ in the following sense:
\begin{itemize}
    \item \textbf{Item 4}: The log likelihood $\ell(\betav \vert \y, \t)$ and its first 
       four partial derivatives are each uniformly continuous over $\t \in \tdom$. 
    \item \textbf{Item 5}: The bound $M(\y; \t)$ from Item 5 holds uniformly in the sense that 
        $\expect{\psur(\y)}{\sup_{\t \in \tdom} M(\y; \t)^2} < \infty$.
    \item \textbf{Item 6 modification 1}: Letting $\lambda_{\info,\t}$ denote the minimum
        eigenvalue of $\info_\t$, assume that $\inf_{\t \in \tdom} \lambda_{\info,\t} > 0$
    \item \textbf{Item 6 modification 2}: The
        empirical log likelihood bound satisfies $\inf_{\t \in \tdom} \varepsilon(\t) > 0$. 
\end{itemize}
\end{assu}

\begin{lem}\label{lem:bayesij_uniform}
Let \cref{assu:bayesij_uniform} hold.
Let $\phi(\betav)$ be a function not depending on $\t$ or $\y$ that is
third-order BCLT okay (\parencite{giordano:2024:bayesij} Definition 2).  Then
\textcite{giordano:2024:bayesij} Theorem 1 holds uniformly over $\t \in \tdom$
in the following sense.  For any target probability $0 < \rho < 1$, there exists 
$C^*$ and $N^*$ not depending on $\t$ for which $\nsur > N^*$ implies that
\begin{align}
\sup_{\t \in \tdom}
\bigg|{}&\,
    \expect{\postt}{\phi(\betav)} - \phi(\betavhat_t) -
\nonumber\\{}&
    \nsur^{-1} \left( 
        \frac{1}{2} \betagrad[2] \phi(\betavhat_t) \infohat_t^{-1} +
        \frac{1}{6} \betagrad \phi(\betavhat_t) 
            \betagrad[3] \hat{\mathcal{L}}_\t(\betavhat_t) \hat{\mathcal{M}}
    \right)
    \bigg| \le \nsur^{-2} C^* 
    \label{eq:ijthm}
\end{align}
with probability at least $\rho$.
\begin{proof}
By assumption, \cref{eq:ijthm} holds for each $\t \in \tdom$ by \textcite{giordano:2024:bayesij} Theorem 1
with some $C^*_\t$ and $N^*_\t$ depending on $\t$.
We now proceed step-by-step through the proof of \textcite{giordano:2024:bayesij} Theorem 1
and show that the constant $C^*_\t$  and $N^*_\t$ depend only on the quantities uniformly controlled by
assumption, from which we have
$$
C^* := \sup_{\t \in \tdom} C^*_\t < \infty
\quad\textrm{and}\quad
N^* := \sup_{\t \in \tdom} N^*_\t < \infty.
$$
We proceed lemma by lemma, indicating how the proof needs to be modified slightly
in order to achieve uniform control over $\t \in \tdom$.

\textbf{Lemma 2.}  We replace the assumption of Lemma 2 with
$$
\sup_{\t \in \tdom} \sup_{\betav \in \betaball} 
\norm{\betagrad[k] \ell(\y \vert \betav)}_2 \le \sup_{\t \in \tdom} M(\y; \t) := M(\y)
\quad\textrm{with}\quad
\expect{\psur(\y)}{M(\y)^2} < \infty,
$$
from which the conclusions of Lemma 2 
(eqs.~20,21, and 22) apply uniformly in $\t \in \tdom$.

\textbf{Lemma 3.}  In Lemma 3, replace $\varepsilon$ with $\inf_{\t \in \tdom} \varepsilon_\t$ and
$\lambda_\info$ with $\inf_{\t \in \tdom} \lambda_{\info,\t}$.  Then by the same proof,
the conclusion of Lemma 3 holds uniformly in $\t$ in the sense that
$$
\sup_{\t \in \tdom} \norm{\betavhat_\t - \betav_{\infty,\t}}_2 \rightarrow 0
\quad\textrm{and}\quad
\inf_{\t \in \tdom} \norm{\infohat_\t}_{op} \ge 
    2 \inf_{\t \in \tdom} \lambda_{\info,\t} := \lambda_\info > 0,
$$
both in probability as $\nsur \rightarrow \infty$.

\textbf{Lemma 4.} Lemma 4 is essentially a convenient re-statement of lemmas 2 and 3.  It follows
that, for any $0 < \rho < 1$, there exists an $N^*$ not depending on $\t$ such that each
event of Lemma 4 holds for all $\t \in \tdom$ with probability at least $\rho$ when $\nsur > N^*$.
Specifically, items 2, 3, and 4 hold with the corresponding inequalities and $\sup_{\t \in \tdom}$,
and items 5 and 6 hold with the corresponding inequalities and $\inf_{\t \in \tdom}$.
There is one caveat --- when the function given in Item 4 of Lemma 4 depends
on $\t$, so that the corresponding neighborhood $\delta_{LLN,\t}$ depends on $\t$,
we must have $\inf_\t \delta_{LLN,\t} > 0$ for each $\epsilon_U$.

\textbf{Lemma 5.} With the above modifications, the upper bound Lemma 5 holds uniformly
in $\tdom$.

\textbf{Lemma 6.}  The quantities $\delta_2$ and $\epsilon_U$ given in Lemma 6 depend
only on the occurrence of the events given in Lemma 4,
and being able to choose $\delta_2$ small enough to make the following
quantities small:
$$
\sup_{\t \in \tdom} \sup_{\betav \in R_{2,t}}\norm{\betagrad[3] \mathcal{L}_\t(\betav)}_2
\quad\textrm{and}\quad
\sup_{\t \in \tdom} \sup_{\betav \in R_{2,t}}\norm{\betagrad[4] \mathcal{L}_\t(\betav)}_2.
$$
Since, by assumption, $\betagrad[k] \mathcal{L}_\t(\betav)$ is
uniformly continuous over $\t \in \tdom$ for $k = 3,4$, $\delta_2$ can be chosen to
make the expressions in the preceding display to satisfy Lemma 6 uniformly in $\t \in \tdom$.

\textbf{Lemma 7} simply relies on Lemma 4.

\textbf{Lemma 8} follows from Lemma 4, again using uniform continuity
of $\betagrad[3] \mathcal{L}_\t(\betav)$ and $\betagrad[4] \mathcal{L}_\t(\betav)$.

\textbf{Lemmas 9, 10, and 11} mostly rearrange terms, relying again on Lemma 4.

The remainder of the proof does not require modification, since we don't
need to consider data dependence in $\phi(\betav)$. 

\end{proof}

\end{lem}


\begin{lem}\label{lem:bayesij_cumulant}
Let \cref{assu:bayesij_uniform} hold.  Consider the posterior cumulant
$\cumulant{\postt}{a(\betav), b(\betav), c(\betav)}$, where
$a(\cdot)$, $b(\cdot)$, and $c(\cdot)$ do not depend on $\t$.  Assume that
each of the functions $a(\betav)$, $b(\betav)$, $c(\betav)$,
their pairwise products $a(\betav) b(\betav)$, $a(\betav) c(\betav)$, and
$c(\betav)$, and the three-way product $a(\betav) b(\betav) c(\betav)$ are
all third-order BCLT okay (\parencite{giordano:2024:bayesij} Definition 2).  

Then, for any probability $0 < \rho < 1$ there 
exists an $N^*$ and $C^*$ such that $\nsur > N^*$ implies that
$$
\nsur^2  \sup_{\t \in \tdom}
\abs{
\cumulant{\postt}{a(\betav), b(\betav), c(\betav)}
} \le C^*
$$
with probability at least $\rho$.
\end{lem}

\begin{proof}
First, for each of $\a$, $\b$, and $\c$, define
$$
\begin{aligned}
\da:={}& \a(\betavhat_\t) - \expect{\postt}{\a(\betav)} \\
\abar(\betav) :={}& \a(\betav) - \a(\betavhat_\t),
\end{aligned}
$$
and so on.  Note that $\expect{\post\t}{\abar} = - \da$.  Then
$$
\begin{aligned}
\MoveEqLeft
\cumulant{\postt}{a(\betav), b(\betav), c(\betav)} 
={}\\&
\expect{\postt}{
    \left(\abar + \da\right)
    \left(\bbar + \db\right)
    \left(\cbar + \dc\right)
}
=\\{}&
\expect{\postt}{\abar \bbar \cbar} +
\\{}&
    \expect{\postt}{\abar \bbar} \dc +
    \expect{\postt}{\abar \cbar} \db +
    \expect{\postt}{\bbar \cbar} \da -
 2 \da \db \dc.
\end{aligned}
$$

Let us first consider $\nsur \da$.  Applying \cref{lem:bayesij_uniform} with
constant $C_\da$,
$$
\begin{aligned}
-\nsur^{-1} C_\da -
\frac{1}{2} \betagrad[2] \phi(\betavhat_t) \infohat_t^{-1} -
        \frac{1}{6} \betagrad \phi(\betavhat_t) 
            \betagrad[3] \hat{\mathcal{L}}_\t(\betavhat_t) \hat{\mathcal{M}}
\le\\{}
\nsur \da
\le\\{}
\nsur^{-1} C_\da -
\frac{1}{2} \betagrad[2] \phi(\betavhat_t) \infohat_t^{-1} -
        \frac{1}{6} \betagrad \phi(\betavhat_t) 
            \betagrad[3] \hat{\mathcal{L}}_\t(\betavhat_t) \hat{\mathcal{M}}.
\end{aligned}
$$

By the modification \textcite{giordano:2024:bayesij} Lemma 3 
given in the proof of \cref{lem:bayesij_uniform}, we have
that $\sup_\t \norm{\betahat_\t}_2$.  From this, and continuity
of $\phi$ and its derivatives, we have
$$
\sup_{\t \in \tdom} \norm{\phi(\betavhat_t)} = \ordp{1}
\quad\textrm{,}\quad
\sup_{\t \in \tdom} \norm{\betagrad \phi(\betavhat_t)} = \ordp{1}
\quad\textrm{,}\quad
\sup_{\t \in \tdom} \norm{\infohat_\t^{-1}}_{op} = \ordp{1}.
$$
Similarly, by the modification of
Lemma 2 in \cref{lem:bayesij_uniform},
together with \cref{assu:bayesij_uniform}
giving uniform continuity of $\mathcal{L}_\t$,
$$
\sup_{\t \in \tdom} \hat{\mathcal{L}}_\t(\betavhat_t) = \ordp{1}.
$$
It follows that
$\sup_{\t \in \tdom} \abs{\nsur^{-1} \da} = \ordp{1}$.  Analogous
results hold for $\db$ and $\dc$.

Similarly, note that $\abar(\betahat_\t) \bbar(\betahat_\t) = 0$,
and $\fracat{\partial \abar(\betav) \bbar(\betav)}{\partial \betav}{\betav=\betahat_\t} = 0$,
from which we have
$$
\nsur^2 \sup_{\t \in \tdom}\abs{\abar \bbar} = \ordp{1},
$$
with analogous results for $\abar \cbar$ and $\bbar \cbar$.  

Finally, $\abar(\betahat_\t) \bbar(\betahat_\t) \cbar(\betahat_\t) = \zerov$,
and 
$$
\fracat{\partial \abar(\betav) \bbar(\betav) \cbar(\betav)}
       {\partial \betav}{\betav=\betahat_\t} = \zerov
$$
so that
$$
\nsur^2 \sup_{\t \in \tdom}\abs{\abar \bbar \cbar} = \ordp{1}.
$$
Combining gives the desired result.
\end{proof}

}


\def\deltardom{(0, \deltamax) \times \rset}
\subsection{Showing that the regressor balance satisfies the conditions}

By \cref{lem:bayesij_cumulant}, it only remains to show that $\postd$ 
satisfies \cref{assu:bayesij_uniform}, where $t = (\delta, \r)$ and
$\tdom = \deltardom$ for some sufficiently small $\deltamax$.

\begin{lem}\label{lem:bayesij_balance}
Let \cref{assu:ij} hold for the original model, and let $\r \in \rset$
(\cref{assu:balance}). Recalling \cref{defn:balance}, define
$$
\ell(\betav; \delta \r) := \expect{\psur(\y, \x)}{\ell(\y \vert \x, \betav, \delta \r)}.
$$
Then $\ell(\betav; \delta \r)$ satisfies \cref{assu:mle} uniformly over
$\delta, \r \in \deltardom$ in the following sense. Define
$$
\begin{aligned}
\betastar_{\delta \r} :={} \argmax{\betav} \ell(\betav; \delta \r)
\quad\textrm{and}\quad
\info_{\delta \r} :={}
    \fracat{\ell(\betav; \delta \r)}
    {\partial \betav \partial \betav^\trans}{\betastar_{\delta \r}}.
\end{aligned}
$$
Let $\lambda_{\delta \r}$ denote the minimum eigenvalue of $\info_{\delta \r}$.
Then there exists a $\deltamax$ such that $\betastar_{\delta \r}$ is
unique for all $\delta, \r \in \deltardom$, and that
$$
\inf_{\delta, \r \in \deltardom} \lambda_{\info, \delta \r} \ge \lambda_\info  > 0.
$$
Furthermore, $\ell(\betav; \delta \r)$ and its derivatives are uniformly
continuous as a function of $\betav$ for $\delta, \r \in \deltardom$.
\end{lem}

Expanding,
$$
\ell(\betav; \delta \r) = 
\ell(\y \vert \x, \betav) + \delta \expect{\psur(\x)}{\r(\x) \x^\trans} \betav.
$$
Since $\sup_{\delta, \r \in \deltardom}
\abs{\delta \expect{\psur(\x)}{\r(\x) \x^\trans}} \le \deltamax \rmax$
by \cref{assu:balance}, $\ell(\betav; \delta \r)$ and its derivatives are uniformly
continuous as a function of $\betav$ for $\delta, \r \in \deltardom$
by \cref{lem:ij} (via the dominated convergence theorem).

We now adapt \cref{lem:ij} to $\ell(\betav; \delta \r)$. First,
for $k \ge 2$, $\betagrad[k] \ell(\betav; \delta \r) = \betagrad[k] \ell(\betav)$,
so \cref{eq:ulln} holds without modification for $k \ge 2$.  

We now adapt the first and zeroth derivatives. Note that by Jensen's inequality,
\begin{align}
\sup_{\r \in \rset}
\norm{\expect{\psur(\x)}{\x \r(\x)}}_2^2
\le \sup_{\r \in \rset} \expect{\psur(\x)}{\norm{\x \r(\x)}_2^2}
\le \rmax^2.
\label{eq:rmaxbound}
\end{align}
Similarly, for $\f(\betav) \ge 0$,
\begin{align}
\MoveEqLeft
\expect{\psur(\x,\y)}{\sup_{\betav \in \betaball} \f(\betav)}
={}
\nonumber\\&
\expect{\psur(\x,\y)}{\sup_{\betav \in \betaball} 
\f(\betav) (\ind{\f(\betav) \ge 1} + \ind{\f(\betav) < 1})}
\nonumber\\\le{}&
\expect{\psur(\x,\y)}{
    \sup_{\betav \in \betaball} 
    \f(\betav)^2 \ind{\f(\betav) \ge 1}} + 
\expect{\psur(\x,\y)}{
    \sup_{\betav \in \betaball} 
    1 \ind{\f(\betav) < 1}}
\nonumber\\\le{}&
\expect{\psur(\x,\y)}{
    \sup_{\betav \in \betaball} 
    \f(\betav)^2 } + 1. 
\label{eq:squnifbound}
\end{align}

Then
$$
\begin{aligned}
\betagrad \ell(\betav; \delta \r) ={}& 
    \betagrad \ell(\betav) + 
    \delta \expect{\psur(\x)}{\r(\x) \x^\trans} 
\Rightarrow \\
\norm{\betagrad \ell(\betav; \delta \r)}_2^2 ={}&
\norm{\betagrad \ell(\betav)}_2^2 + 
\delta^2 \norm{\expect{\psur(\x)}{\r(\x) \x^\trans}}_2^2 +
2 \delta \expect{\psur(\x)}{\r(\x) \x^\trans} \betagrad \ell(\betav) 
\\\le{}&
    \norm{\betagrad \ell(\betav)}_2^2 +
    \deltamax^2 \rmax^2 + 2 \deltamax \rmax 
    \norm{\betagrad \ell(\betav)}_2.
\end{aligned}
$$
The right hand side of the preceding display does not depend on 
$\r$ or $\delta$, and has finite expectation under $\psur(\x, \y)$ by
\cref{lem:ij} and \cref{eq:rmaxbound,eq:squnifbound}, and so we have 
identified $M'_1(\x, \y)$ such that
$$
\sup_{\delta, \r \in \deltardom}
\sup_{\betav \in \betaball} \norm{\betagrad[k] \ell(\betav; \delta \r)}_2^2 
    \le \tilde{M}_1(\x, \y)
    \quad\textrm{and}\quad\expect{\psur(\x, \y)}{\tilde{M}_1(\x, \y)} < \infty.
$$
Similarly,
$$
\begin{aligned}
\ell(\betav; \delta \r)^2 ={}&
    \ell(\betav)^2 + 
    \delta^2 \trace{
        \expect{\psur(\x)}{\r(\x) \x} 
        \expect{\psur(\x)}{\r(\x) \x^\trans} \betav \betav^\trans} +
\\{}&
    2 \delta \ell(\betav) \expect{\psur(\x)}{\r(\x) \x^\trans} \betav
\\\le{}&
\ell(\betav)^2 +
\deltamax^2 \rmax^2 \norm{\betav}_2^2 + 2 \ell(\betav) \rmax \norm{\betav}.
\end{aligned}
$$
As before, we have identified $\tilde{M}(\x, \y)$ not depending on 
$\r$ or $\delta$ satisfying
$$
\sup_{\delta, \r \in \deltardom}
\sup_{\betav \in \betaball} \ell(\betav; \delta \r)_2^2 
    \le \tilde{M}(\x, \y)
    \quad\textrm{and}\quad\expect{\psur(\x, \y)}{\tilde{M}(\x, \y)} < \infty.
$$

{

We now show that we can control $\inf_{\delta, \r \in \deltardom} \lambda_{\delta \r}$.
Let $\minev{\cdot}$ denote the minimum eigenvalue of the argument. Now, we have that
$$
\lambda_{\delta \r} = \minev{\info_{\delta \r}}  =
\minev{
\expect{\psur(\x)}{\Agrad{2}(\x^\trans {\betastar_{\delta \r}}) \x^{\otimes 2}}
}.
$$
Since by \cref{assu:ij}, $\minev{\info} = \lambda_{0 \cdot \r} > 0$,
and by \cref{lem:ij} the map
$$
\betav \mapsto 
\minev{
\expect{\psur(\x)}{\Agrad{2}(\x^\trans \betav) \x^{\otimes 2}}
}
$$
is continuous at $\betastar$, we can choose $\Delta > 0$ small enough
that 
$$
\betav \in \betaball
\quad \Rightarrow \quad
\minev{
\expect{\psur(\x)}{\Agrad{2}(\x^\trans \betav) \x^{\otimes 2}}
} \ge \frac{1}{2} \lambda_\info > 0.
$$ 
So it suffices to take such a $\Delta$ and 
show that there exists $\deltamax$ sufficiently small that
$$
\betastar_{\delta \r} \in \betaball
\quad\textrm{for all}\quad \delta, \r \in \deltadom.
$$
Since $\ell(\betav; \delta \r)$ is continuously differentiable, $\betastar$
is a solution to
$$
\betagrad \ell(\betastar_{\delta \r}; \delta \r) = 
\betagrad \ell(\betastar_{\delta \r}) + \delta \expect{\psur}{\r(\x) \x} =
\zerov.
$$
\def\v{\mathbf{v}}

Note that 
$\norm{\delta \expect{\psur}{\r(\x) \x}}_2 \le \deltamax \rmax$ by
\cref{assu:balance}, so we can define $\betastar_{\v}$ as the solution
to
$$
\betagrad \ell(\betastar_{\v}) + \v = \zerov
\textrm{ for }\norm{\v}_2 \le \deltamax \rmax.  
$$
Since $\v = \zerov$ corresponds to $\betastar$, at which 
$\betagrad[2] \ell(\betastar_{\v})$ is positive definite,
the inverse function theorem implies that
there exists a neighborhood $\{\v: \norm{\v}_2 < \Delta' \}$ such
that the map $\v \mapsto \betastar_{\v}$ is continuous
\parencite[Theorem 3.3.2]{krantz:2002:implicit}.  Taking
$$
\deltamax \le 
\frac{\mathrm{min}\left(\Delta, \Delta'\right)}{\rmax}
$$
thus suffices to guarantee that
$$
\inf_{\delta, \r \in \deltardom} \lambda_{\delta \r} \ge \frac{1}{2} \lambda_\info > 0.
$$
}

{

\def\betabound{\partial \betaball(\delta \r)}
\def\ellhat{\hat{\ell}_{\nsur}}

We have one more condition of \cref{assu:bayesij_uniform} to satisfy,
Item 6.  A small modification of the proof of \cref{lem:loglik_bounded}
gives that, for sufficiently large $\nsur$, there exists a $\gamma > 0$
such that   
$$
\sup_{\delta, \r \in \deltardom} 
\left( 
    \ellhat(\betastar_{\delta \r}; \delta \r)  -
    \sup_{\betav \in \betabound} \ellhat(\betav; \delta, \r) \right)
> \gamma > 0,
$$
where here $\betabound = \{\betav: \norm{\betav - \betastar_{\delta \r}}_2 \le \Delta\}$.
This follows by a uniform law of large numbers applied to both $\betav$ and
to $\delta, \r$. The remainder of the proof of \cref{lem:loglik_bounded}
is unchanged, giving
$$
\sup_{\delta, \r \in \deltardom} 
\sup_{\betav \in \betadom \setminus \betaball} \left(
    \meansur \ell(\y_i \vert \x_i, \betastar_{\delta \r}; \delta \r) -
    \meansur \ell(\y_i \vert \x_i, \betav; \delta \r)
\right) \ge \gamma > 0.
$$

}

It follows that \cref{assu:bayesij_uniform} is satisfied,
and \cref{lem:bayesij_cumulant} gives the desired result.

\section{Generating perturbed binary datasets}\label{app:binary}

\def\mtil{\tilde{m}}
\def\mtiltil{\tiltil{m}}
\def\Mtil{\tilde{\bm{M}}}
\def\Mtiltil{\tiltil{\bm{M}}}
\def\u{u}
\def\unif#1{\mathrm{Uniform}\left(#1\right)}



In this section, we describe one method for constructing randomized binary
datasets that approximately represent a particular continuous perturbation.  How
to do so optimally seems to be an interesting topic for future work.  We then
rerun MCMC for a particular draw of the binary dataset to see whether the linear
prediction given by MrPlew can extrapolate to analytically meaningful
differences in the MrP estimates.  Our results are mixed, and so we strongly
recommend that balance and subgroup contribution plots be used as computationally
efficient but approximate ways to explore the space of potentially problematic
datasets, which are then verified by MCMC.


\subsection{Constructing a binary response vector}

Recall that our goal, by \cref{eq:ytil_expectation}, is to identify
a binary random variable which we will call $\ytiltil$, such that
\begin{align}
\expect{\p(\ytiltil \vert \x)}{\ytiltil} \approx
    \expect{\p(\ytil \vert \x)}{\ytil} =
    \expect{\p(\y \vert \x)}{\y} + \delta \r.
    \label{eq:binary_goal}
\end{align}
Our motivation for doing so will be that, if $\muhat[\mrp](\Ytil) -
\muhat[\mrp](\Y)$ is large, so that $\r$ appears to be ``imbalanced'' according
to \cref{thm:balance}, then we hope that the realizable
$\muhat[\mrp](\Ytiltil) - \muhat[\mrp](\Y)$ may be ``imbalanced'' as well.  We
are further motivated to find a $\Ytiltil$ that is ``close'' to $\Y$ so that
the approximation \cref{eq:balance_linear} remains approximately valid for
$\muhat[\mrp](\Ytiltil) - \muhat[\mrp](\Y)$.

One way to proceed is via a perturbation to the parametric bootstrap
\parencite{efron:1994:bootstrap}.  We begin
with a preliminary guess $\mhat_i := \mhat(\x_i) \approx
\expect{\p(\y\vert\x)}{\y}$ for each $i$ in the survey. If the guess
$\mhat(\cdot)$ is poor, then the perturbed data will not bear the same
relationship to the regressors as would a genuine alternative draw of the data.
However, we emphasize that the role of $\mhat(\x_i)$ is to \emph{define a data
perturbation of interest}, rather than to perform formal inference. A reasonable
choice might be to take $\betavhat := \expect{\post}{\betav}$ and $\mhat(\x_i) =
\m(\betavhat^\trans \x_i)$.  Two other extremes would be to take $\m(\x_i) =
\meansur \y_i$, or to take $\m(\x_i) = \y_i$.  As we will see below, these
extremes represent tradeoffs in their ability to plausibly reproduce the change
\cref{eq:ytil_expectation} on one hand, and to produce estimates that are close
to $\Y$ on the other.

Supposing that $\expecty{\y} = \mhat(\x)$, we wish to draw a binary $\ytiltil$
that is ``not far'' from $\y$, but which has expectation $\mhat(\x) + \delta
\r$.  One way to do that is to couple $\y$ and $\ytiltil$ with a common uniform
random variable $\u \sim \unif{0,1}$: $\y = \ind{\u < \mhat(\x)}$ and $\ytiltil
= \ind{\u < \mhat(\x) + \delta \r}$.  Under this construction,
$\expect{}{\ytiltil} = \mhat(\x) + \delta \r$ and $\expect{}{\y} = \mhat(\x)$ as
desired, as long as $\mhat(\x) + \delta \r \in [0,1]$ and $\mhat(\x) \in
[0,1]$.  Given $\y$ and $\x$, the conditional distribution of $\u_i$ is then
$$
\p(\u \vert \y, \x) =
\begin{cases}
\unif{0,\mhat(\x)} & \textrm{ if }\y = 1\\
\unif{\mhat(\x), 1} & \textrm{ if }\y = 0.
\end{cases}
$$
Then if we draw $\u \sim \p(\u \vert \y, \x)$ and set $\ytiltil = \ind{\u \le
\mhat(\x) + \delta \r}$, then $\ytiltil$ is highly correlated with $\y$
marginally, and satisfies \cref{eq:ytil_expectation} when $\mhat(\x) =
\expect{\p(\y \vert \x)}{\y}$. The procedure, which would be repeated for each
$i$, is shown in \Cref{alg:binary}.

\begin{algorithm}[htbp]
  \SetKw{Assert}{assert}
  \caption{Draw perturbed binary data}\label{alg:binary}
  \KwIn{Estimate $\mhat_i \approx \expect{\p(\y \vert \x_i)}{\y}$,
        perturbation $\delta \r_i$, original response $\y_i$}
  \KwOut{Vector of $\ytil_i \in \{0, 1\}$ with
  $\expect{\p(\ytil | \x_i)}{\ytil} \approx \expect{\p(\y | \x)}{\y_i} + \delta \r_i$}

  \Assert{$\mhat_i + \delta \r_i \in [0,1]$}

  \tcp{Draw $u_i \sim p(u \mid y_i)$}
  \eIf{$y_i = 1$}{
    $u_i \sim \mathrm{Uniform}(\hat{m}_i, 1)$\;
  }{
    $u_i \sim \mathrm{Uniform}(0, \hat{m}_i)$\;
  }
  \tcp{Draw $\ytiltil_i \sim p(\ytiltil \mid y_i)$}
  $\ytiltil_i \gets \ind{\u_i \ge \mhat_i + \delta \r_i}$
\end{algorithm}
Note also that \cref{alg:binary} is random.  In general, for a given
$\mhat(\cdot)$, there are many binary datasets that are consistent with a
particular mean perturbation.

We note that choices of $\mhat(\x_i)$ like $\meansur \y_i$ that ``underfit'' the
data permit larger perturbations, since $\delta$ may be larger before
$\mhat(\x_i)$ leaves $[0,1]$, but at the risk of failing to reproduce the
intended relationship between $\x_i$ and $\ytiltil_i$.  On the other hand,
``overfitting'' the data, say by choosing $\mhat(\x_i) = \y_i$ perfectly
reproduces the relationship between $\x_i$ and $\y_i$, but there may be no
non-zero $\delta$ satisfying $\mhat(\x_i) + \delta \r_i \in [0,1]$ --- for
example, if $\r_i > 0$ and $\y_i = 1$, then we must have $\delta = 0$.
Ultimately, the tradeoff between fidelity to the data generating distribution
and the ability to produce perturbed datasets is a judgement call that defines
the robustness question that is being asked.

\section{Additional Experimental Results}\label{app:experiments}
\subsection{Supplementary Graphs}

\LaxphilipsBalanceInteractions{}

\StoriesBalanceInteractions{}


\BootstrapPlot{}

\PoolingComparisonPlot{}

\subsection{Extrapolating to assess non-local robustness}
\label{app:non_local_robustness}

Following the setup in \Cref{sec:non_local_robustness}, we investigate how well
the MrPlew balance checks assess non-local robustness after extrapolating along
a chosen dimension. 
For our given outcome prediction $\mhat(\cdot)$, we investigate
\Cref{eq:ytil_expectation_local_approximation} by re-running MCMC. Specifically, we choose
a range of $\delta$, up to the largest possible step size
$\delta_{\mathrm{max}}$ that result in  $\mhat_i + \delta \r_i \in [0,1]$ for
all $i$.  At the most extreme points, which effectively set all the responses
in the given category to $1$, we have changed $\LaxPctflipped{}$\% of the
responses in the Same-Sex Marriage and $\AlexanderPctflipped{}$\% of the
responses in the Name Change dataset. For each $\delta$ in each analysis, we ran
\cref{alg:binary} from \cref{app:binary} to produce a draw $\Ytiltil$, and then
re-ran MCMC to compute $\muhat[\mrp](\Ytiltil)$.

\AlexanderRefitPlot{}

\LaxRefitPlot{}


\subsection{Nonlinearity in Same-Sex Marriage analysis}\label{sec:lax_nonlinearity}

In this section we briefly demonstrate that posterior is in fact
locally linear in \cref{fig:laxphilipsrefitplot}, but that
even the least-perturbed binary vectors leave the domain of linearity 
quickly due to a large degree of posterior curvature.

We selected the smallest $\Ytiltil$ not equal to $\Y$ from
\cref{fig:laxphilipsrefitplot}.  This $\Ytiltil$ corresponds to the second
datapoint from the left.  We then define the MrP estimate
\begin{align}
\muhat[\mrp](\epsilon) := \muhat[\mrp]((1 - \epsilon) \Y + \epsilon \Ytiltil).
\label{eq:mrp_epsilon}
\end{align}
Both $\Y$ and $\Ytiltil$ are binary vectors, and \cref{eq:mrp_epsilon}
defines a smooth path between them as $\epsilon$ varies from $0$ to $1$.  For
sufficiently small $\epsilon$ we can estimate $\muhat[\mrp](\epsilon)$ with
self-normalized importance sampling, and we consider $\epsilon$ for which there
are at least $1000$ effective samples in the importance-reweighted posterior.

\Cref{fig:isplot} shows the path of both the MrPlew linear approximation
to $\muhat[\mrp](\epsilon) - \muhat[\mrp]$ and the importance sampling
posterior estimate.  The local linearity is evident for very small
$\epsilon$, but the posterior rapidly deviates from the linear approximation
long before $\epsilon = 1$.

\ImportanceSamplingPlot{}

Though not shown, a randomly selected binary response vector that changes the
same number of entries as $\Ytiltil$ exhibits very little curvature.   There may
be something about the fact that we are perturbing the posterior in a direction
of imbalance that is causing the severe non-linearity.  Theoretically and
practically understanding why the Same-Sex Marriage analysis exhibits such
strong curvature but the Name Change analysis does not remains important future
work.






\end{appendices}

\end{document}